\documentclass[12pt]{article}
\newif\iflongversion
\longversiontrue

\usepackage{fullpage}
\usepackage{times}

\usepackage[bookmarks=false]{hyperref}

\usepackage{doi}

\usepackage{xcolor}
\usepackage{amsthm}
\theoremstyle{plain}
\newtheorem{theorem}{Theorem}
\newtheorem{lemma}[theorem]{Lemma}
\newtheorem{proposition}[theorem]{Proposition}
\newtheorem{corollary}[theorem]{Corollary}

\theoremstyle{definition}
\newtheorem{definition}{Definition}
\theoremstyle{remark}

\newtheorem{example}{Example}

\usepackage{booktabs}   %
\usepackage{subcaption} %

\usepackage{xcolor}

\usepackage{enumitem}
\usepackage{bm}
\usepackage[bbsets,Dfprime,setrelation]{math}

\usepackage{wasysym}
\usepackage[prefixflatinterpret,bracketmodalinterpret,fixformat,silentconst,sidenotecalculus,longseqcontext,seqinsist]{logic}
\usepackage[prefixflatinterpret,bracketmodalinterpret,fixformat,silentconst,differentialdL,simplenames]{dL}
\def\leftrule{L}%
\def\rightrule{R}%
\newcommand{\bebecomes}{\mathrel{::=}}
\newcommand{\alternative}{~|~}

\newcommand{\I}{\dLint[const=I,state=\omega]}
\newcommand{\It}{\dLint[const=I,state=\nu]}
\newcommand{\If}{\DALint[const=I,flow=\varphi]}
\newcommand*{\Iff}[1][\zeta]{\dLint[const=I,state=\varphi(#1)]}%

\newsavebox{\Rval}%
\sbox{\Rval}{$\scriptstyle\mathbb{R}$}
\relax

  \newdimen\linferenceRulehskipamount%
  \newdimen\lcalculuscollectionvskipamount%
\renewcommand{\linferenceRuleNameSeparation}{\hspace{4pt}}

\definecolor{vblue}{rgb}{.1,.15,.62}
\definecolor{vgray}{rgb}{.35,.35,.35}
\renewcommand*{\irrulename}[1]{\text{\textcolor{vgray}{\upshape#1}}}%

\usepackage{amsthm}

\let\determinant\blank
\let\sgn\blank
\let\perm\blank
\let\trace\blank

\let\idmatrix\blank
\DeclareMathOperator{\determinant}{det}
\DeclareMathOperator{\trace}{tr}

\DeclareMathOperator{\idmatrix}{\mathbb{I}}
\DeclareMathOperator{\sgn}{sgn}
\DeclareMathOperator{\perm}{perm}

\usepackage{prettyref}
\newcommand{\rref}[2][]{\prettyref{#2}}
\newrefformat{sec}{Section\,\ref{#1}}
\newrefformat{subsec}{Section\,\ref{#1}}
\newrefformat{def}{Def.\,\ref{#1}}
\newrefformat{thm}{Theorem\,\ref{#1}}
\newrefformat{prop}{Proposition\,\ref{#1}}
\newrefformat{lem}{Lemma\,\ref{#1}}
\newrefformat{cor}{Corollary\,\ref{#1}}
\newrefformat{ex}{Example\,\ref{#1}}
\newrefformat{tab}{Table\,\ref{#1}}
\newrefformat{fig}{Fig.\,\ref{#1}}
\newrefformat{case}{case\,\ref{#1}}
\newrefformat{foot}{Footnote\,\ref{#1}}
\iflongversion
\newrefformat{app}{Appendix\,\ref{#1}}
\newenvironment{proofsketch}[1][TODO]{\proof[Proof Summary (\ifthenelse{\equal{#1}{TODO}}{TODO}{\rref{#1}})]}{\endproof}
\newcommand{\seeapp}[1]{(see #1)}
\else
\newrefformat{app}{\cite{DBLP:journals/corr/abs-1802-01226}}
\newenvironment{proofsketch}[1][]{\proof[Proof Summary \ifthenelse{\equal{#1}{TODO}}{TODO}{\cite{DBLP:journals/corr/abs-1802-01226}}]}{\endproof}
\newcommand{\seeapp}[1]{#1}
\fi

\newcommand{\ie}{i.e.}
\newcommand{\cf}{cf.}
\newcommand{\eg}{e.g.}

\newcommand{\cmp}{\succcurlyeq}

\newcommand{\pmc}{\preccurlyeq}

\renewcommand{\allvars}{\mathbb{V}}
\newcommand{\States}{\mathbb{S}}

\newcommand{\dprogressin}[3][]{%
  {\langle{\pevolvein{#2}{#3}}\rangle}{\ddnext} {#1}%
}
\newcommand{\dprogressinsmall}[3][]{%
  {\langle{\pevolvein{#2}{#3}}\rangle}{\ddnextsmall} {#1}%
}

\newcommand{\ddnext}{\bigcirc}
\newcommand{\ddnextsmall}{\circ}

\newcommand{\initassum}{x{=}y}
\renewcommand*{\der}[2][]{(#2)\ifthenelse{\equal{#1}{}}{'}{^{(#1)}}}

\renewcommand*{\lie}[3][]
{\mathcal{L}_{#2}^{\ifthenelse{\equal{#1}{}}{}{^{(#1)}}}(#3)}
\newcommand*{\lied}[3][]{\overset{\bm .}{#3}\ifthenelse{\equal{#1}{}}{}{^{(#1)}}}

\newcommand{\siglied}[3][]{\overset{\bm .}{#3}^{\Dostar{#1}}}

\newcommand{\Dostar}[1]{\ifthenelse{\equal{#1}{}}{(*)}{-(*)}}
\newcommand{\sigliedgt}[3][]{\siglied[#1]{#2}{#3}>0}
\newcommand{\sigliedgeq}[3][]{\siglied[#1]{#2}{#3}\geq0}

\newcommand{\sigliedzero}[3][]{\siglied[#1]{#2}{#3}=0}

\newcommand{\sigliedsai}[3][]{\siglied[#1]{#2}{#3}}

\newcommand{\sem}[1]{\lenvelope#1\renvelope}
\renewcommand{\int}[1]{\textrm{int}(#1)}

\newcommand{\polyn}[2]{#1}

\renewcommand*{\vec}[1]{\mathbf{#1}}

\newcommand{\vecpolyn}[2]{\vec{#1}}

\newcommand{\matpolyn}[2]{#1}

\newcommand{\fvar}{\phi}
\newcommand{\rfvar}{P}
\newcommand{\rrfvar}{R}
\newcommand{\rtfvar}{S}
\newcommand{\rcfvar}{C}
\newcommand{\coalgvar}{\tilde{\ivr}}

\newcommand{\solvar}{\varphi}
\newcommand{\truncafter}[2]{#1|_{#2}}
\newcommand{\soltrunc}[1]{\truncafter{\solvar}{#1}}

\usepackage{ifpdf}
\ifpdf
\fi

\usepackage{fancyhdr}
\title{Differential Equation Axiomatization\\
  \large The Impressive Power of Differential Ghosts}
\author{Andr\'e Platzer \and
  Yong Kiam Tan \thanks{
  Computer Science Department, Carnegie Mellon University, Pittsburgh, USA
  {\{aplatzer$|$yongkiat\}@cs.cmu.edu}
  }
}
\date{}

\begin{document}
\maketitle
\allowdisplaybreaks
\thispagestyle{empty}

\begin{abstract}
We prove the completeness of an axiomatization for differential equation invariants.
First, we show that the differential equation axioms in differential dynamic logic are complete for all algebraic invariants.
Our proof exploits differential ghosts, which introduce additional variables that can be chosen to evolve freely along new differential equations.
Cleverly chosen differential ghosts are the proof-theoretical counterpart of dark matter.
They create new hypothetical state, whose relationship to the original state variables satisfies invariants that did not exist before.
The reflection of these new invariants in the original system then enables its analysis.

We then show that extending the axiomatization with existence and uniqueness axioms makes it complete for all local progress properties,
and further extension with a real induction axiom makes it complete for all real arithmetic invariants.
This yields a parsimonious axiomatization, which serves as the logical foundation for reasoning about invariants of differential equations.
Moreover, our results are purely axiomatic, and so the axiomatization is suitable for sound implementation in foundational theorem provers.

\end{abstract}

\section{Introduction}
\label{sec:introduction}

Classically, differential equations are studied by analyzing their solutions.
This is at odds with the fact that solutions are often much more complicated than the differential equations themselves.
The stark difference between the simple local description as differential equations and the complex global behavior exhibited by solutions is fundamental to the descriptive power of differential equations!

Poincar\'e's qualitative study of differential equations crucially exploits this difference by deducing properties of solutions \emph{directly from the differential equations}. This paper completes an important step in this enterprise by identifying the logical foundations for proving invariance properties of polynomial differential equations.

We exploit the differential equation axioms of differential dynamic logic (\dL)~\cite{DBLP:conf/lics/Platzer12b,DBLP:journals/jar/Platzer17}. \dL is a logic for deductive verification of hybrid systems that are modeled by hybrid programs combining discrete computation (\eg, assignments, tests and loops), and continuous dynamics specified using systems of ordinary differential equations (ODEs).
By the continuous relative completeness theorem for \dL~\cite[Theorem 1]{DBLP:conf/lics/Platzer12b}, verification of hybrid systems reduces completely to the study of differential equations.
Thus, the hybrid systems axioms of \dL provide a way of lifting our findings about differential equations to hybrid systems.
The remaining practical challenge is to find succinct real arithmetic system invariants; any such invariant, once found, can be proved within our calculus.

To understand the difficulty in verifying properties of ODEs, it is useful to draw an analogy between ODEs and discrete program loops.\footnote{In fact, this analogy can be made precise: \dL also has a converse relative completeness theorem~\cite[Theorem 2]{DBLP:conf/lics/Platzer12b} that reduces ODEs to discrete Euler approximation loops.} Loops also exhibit the dichotomy between global behavior and local description. Although the body of a loop may be simple, it is impractical for most loops to reason about their global behavior by unfolding all possible iterations. Instead, the premier reasoning technique for loops is to study their loop invariants, \ie, properties that are preserved across each execution of the loop body.

Similarly, invariants of ODEs are real arithmetic formulas that describe subsets of the state space from which we cannot escape by following the ODEs. The three basic \dL axioms for reasoning about such invariants are: (1) \emph{differential invariants}, which prove simple invariants by locally analyzing their Lie derivatives, (2) \emph{differential cuts}, which refine the state space with additional provable invariants, and (3) \emph{differential ghosts}, which add differential equations for new ghost variables to the existing system of differential equations.

We may relate these reasoning principles to their discrete loop counterparts: (1) corresponds to loop induction by analyzing the loop body, (2) corresponds to progressive refinement of the loop guards, and (3) corresponds to adding discrete ghost variables to remember intermediate program states.
At first glance, differential ghosts seem counter-intuitive: they increase the dimension of the system, and should be adverse to analyzing it! However, just as discrete ghosts~\cite{DBLP:journals/cacm/OwickiG76} allow the expression of new relationships between variables along execution of a program, differential ghosts that suitably co-evolve with the ODEs crucially allow the expression of new relationships along solutions to the differential equations.
Unlike the case for discrete loops, differential cuts strictly increase the deductive power of differential invariants for proving invariants of ODEs; differential ghosts further increase this deductive power~\cite{DBLP:journals/lmcs/Platzer12}.

This paper has the following contributions:
\begin{enumerate}
\item We show that \emph{all} algebraic invariants, \ie, where the invariant set is described by a formula formed from finite conjunctions and disjunctions of polynomial equations, are provable using only the three ODE axioms outlined above.
\item We introduce axioms internalizing the existence and uniqueness theorems for solutions of differential equations. We show that they suffice for reasoning about all \emph{local progress} properties of ODEs for all real arithmetic formulas.
\item We introduce a real induction axiom that allows us to reduce invariance to local progress. The resulting \dL calculus \emph{decides} all real arithmetic invariants of differential equations.
\item Our completeness results are axiomatic, enabling disproofs.
\end{enumerate}

Just as discrete ghosts can make a program logic relatively complete~\cite{DBLP:journals/cacm/OwickiG76}, our first completeness result shows that differential ghosts achieve completeness for algebraic invariants in \dL. We extend the result to larger classes of hybrid programs, including, \eg, loops that switch between multiple different ODEs.

We note that there already exist prior, complete procedures for checking algebraic, and real arithmetic invariants of differential equations~\cite{DBLP:conf/tacas/GhorbalP14,DBLP:conf/emsoft/LiuZZ11}. Our result identifies a list of axioms that serve as a \emph{logical foundation} from which these procedures can be implemented as derived rules. This logical approach allows us to precisely identify the underlying aspects of differential equations that are needed for sound invariance reasoning. Our axiomatization is \emph{not limited} to proving invariance properties, but also completely axiomatizes disproofs and other qualitative properties such as local progress.

The parsimony of our axiomatization makes it amenable to sound implementation and verification in foundational theorem provers~\cite{DBLP:conf/cade/FultonMQVP15,DBLP:conf/cpp/BohrerRVVP17} using \dL's uniform substitution calculus~\cite{DBLP:journals/jar/Platzer17}, and is in stark contrast to previous highly schematic procedures \cite{DBLP:conf/tacas/GhorbalP14,DBLP:conf/emsoft/LiuZZ11}.

\iflongversion
All proofs are in the appendices~\ref{app:axiomatization} and~ \ref{app:completeness}.
\else
All proofs are in a companion report \rref{app:}.
\fi

\section{Background: Differential Dynamic Logic}
\label{sec:background}

This section briefly reviews the relevant continuous fragment of \dL, and establishes the notational conventions used in this paper. The reader is referred to the literature~\cite{DBLP:conf/lics/Platzer12b,DBLP:journals/jar/Platzer17} and \rref{app:axiomatization} for a complete exposition of \dL, including its discrete fragment.

\subsection{Syntax}
\label{subsec:backgroundsyntax}

Terms in \dL are generated by the following grammar, where $x$ is a variable, and $c$ is a rational constant:
\[
	e \bebecomes x \alternative c \alternative e_1 + e_2 \alternative e_1 \cdot e_2
\]
These terms correspond to polynomials over the variables under consideration. For the purposes of this paper, we write $x$ to refer to a vector of variables $x_1,\dots,x_n$, and we use $p(x),q(x)$ to stand for polynomial terms over these variables. When the variable context is clear, we write $p,q$ without arguments instead. Vectors of polynomials are written in bold $\vecpolyn{p}{x},\vecpolyn{q}{x}$, with $\vecpolyn{p}{x}_i,\vecpolyn{q}{x}_i$ for their $i$-th components.

The formulas of \dL are given by the following grammar, where $\sim$ is a comparison operator $=,\geq,>$, and $\alpha$ is a hybrid program:
\[
  \fvar \bebecomes e_1 \sim e_2 \alternative \fvar_1 \land \fvar_2 \alternative \fvar_1 \lor \fvar_2 \alternative \lnot{\fvar} \alternative \lforall{x}{\fvar} \alternative \lexists{x}{\fvar}  \alternative \dbox{\alpha}{\fvar} \alternative\ddiamond{\alpha}{\fvar}
\]
Formulas can be normalized such that $e_1 \sim e_2$ has $0$ on the right-hand side. We write $p \cmp 0$ if there is a free choice between $\geq$ or $>$. Further, \(p \pmc 0\) is \(-p \cmp 0\), where $\pmc$ stands for $\leq$ or $<$, and $\cmp$ is correspondingly chosen. Other logical connectives, e.g., $\limply,\lbisubjunct$ are definable.
For the formula $\vecpolyn{p}{x}=\vecpolyn{q}{x}$ where both $\vecpolyn{p}{x},\vecpolyn{q}{x}$ have dimension $n$, equality is understood \emph{component-wise} as \m{\landfold_{i=1}^{n} \vecpolyn{p}{x}_i=\vecpolyn{q}{x}_i} and $\vecpolyn{p}{x}\neq\vecpolyn{q}{x}$ as $\lnot{(\vecpolyn{p}{x}=\vecpolyn{q}{x})}$.
We write $\rfvar(x),\ivr(x)$ for first-order formulas of real arithmetic, \ie, formulas not containing the modal connectives. We drop the dependency on $x$ when the variable context is clear. The modal formula $\dibox{\alpha}\fvar$ is true iff $\fvar$ is true after all transitions of $\alpha$, and its dual $\didia{\alpha}\fvar$ is true iff $\fvar$ is true after some transition of $\alpha$.

Hybrid programs $\alpha$ allow us to express both discrete and continuous dynamics. This paper focuses on the continuous fragment\footnote{We only consider weak-test $\dL$, where $\ivr$ is a first-order formula of real arithmetic.}:
\[\alpha \bebecomes \cdots \alternative \pevolvein{\D{x}=\genDE{x}}{\ivr}\]
We write \pevolvein{\D{x}=\genDE{x}}{\ivr} for an autonomous vectorial differential equation system in variables $x_1,\dots,x_n$ where the RHS of the system for each $\D{x_i}$ is a polynomial term $f_i(x)$. The evolution domain constraint $\ivr$ is a formula of real arithmetic, which restricts the set of states in which we are allowed to continuously evolve. We write \m{\pevolve{\D{x}=\genDE{x}}} for \m{\pevolvein{\D{x}=\genDE{x}}{\ltrue}}.
We use a running example (\rref{fig:exampleODE}):
\[\alpha_e \mdefequiv \D{u}=-v + \frac{u}{4} (1 - u^2 - v^2), \D{v} = u + \frac{v}{4} (1 - u^2 - v^2)\]

\begin{figure}
\centering
\includegraphics[width=\columnwidth]{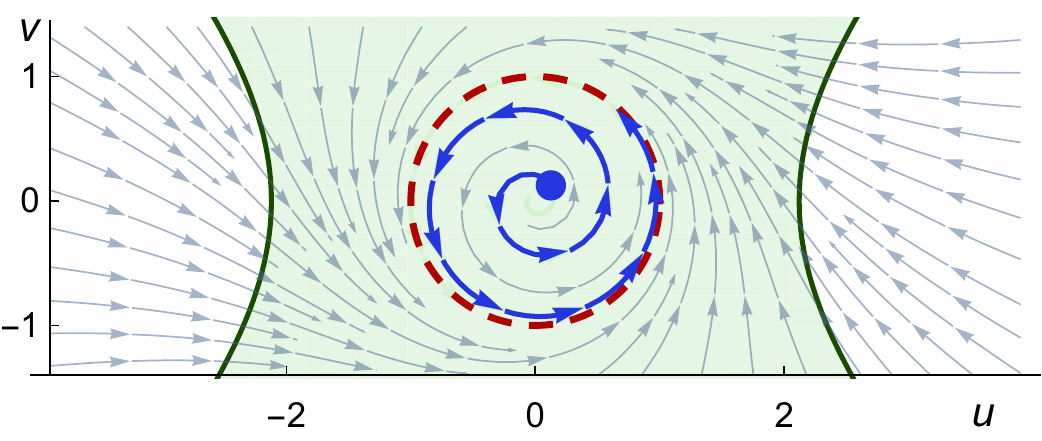}
\caption{The red dashed circle $u^2+v^2=1$ is approached by solutions of $\alpha_e$ from all points except the origin, e.g., the blue trajectory from $(\frac{1}{8},\frac{1}{8})$ spirals towards the circle. The circle, green region $u^2 \leq v^2 + \frac{9}{2}$, and the origin are invariants of the system.}
\label{fig:exampleODE}
\end{figure}
Following our analogy in \rref{sec:introduction}, solutions of $\D{x}=\genDE{x}$ must continuously (locally) follow its RHS, $\genDE{x}$. Figure~\ref{fig:exampleODE} visualizes this with directional arrows corresponding to the RHS of $\alpha_e$ evaluated at points on the plane. Even though the RHS of $\alpha_e$ are polynomials, its solutions, which must locally follow the arrows, already exhibit complex global behavior. Figure~\ref{fig:exampleODE} suggests, \eg, that all points (except the origin) globally evolve towards the unit circle.

\subsection{Semantics}
\label{subsec:backgroundsemantics}

A state $\omega : \allvars \to \reals$ assigns a real value to each variable in $\allvars$. We may let $\allvars = \{x_1,\dots, x_n\}$ since we only need to consider the variables that occur\footnote{Variables $v$ that do not have an ODE $\D{v}=\dots$ also do not change (similar to \(\D{v}=0\)).}. Hence, we shall also write states as $n$-tuples $\omega : \reals^n$ where the $i$-th component is the value of $x_i$ in that state.

The value of term $\polyn{e}{x}$ in state $\iget[state]{\I}$ is written \(\ivaluation{\I}{\polyn{e}{x}}\) and defined as usual. The semantics of comparison operations and logical connectives are also defined in the standard way. We write $\imodel{\I}{\fvar}$ for the set of states in which $\fvar$ is true. For example, $\omega \in \imodel{\I}{e_1 \leq e_2}$ iff $\ivaluation{\I}{\polyn{e_1}{x}} \leq \ivaluation{\I}{\polyn{e_2}{x}}$, and $\omega \in \imodel{\I}{\fvar_1 \land \fvar_2}$ iff $\omega \in \imodel{\I}{\fvar_1}$ and $\omega \in \imodel{\I}{\fvar_2}$.

Hybrid programs are interpreted as transition relations, $\iaccess[\alpha]{\I} \subseteq \reals^n \times \reals^n$, between states. The semantics of an ODE is the set of all pairs of states that can be connected by a solution of the ODE:
\begin{align*}
&\iaccessible[\pevolvein{\D{x}=\genDE{x}}{\ivr}]{\I}{\It} ~\text{iff}~\text{there is a real}~T\geq 0~\text{and a function}\\
&\solvar:[0,T] \to \reals^n ~\text{with}~ \solvar(0)=\iget[state]{\I}, \solvar(T)=\iget[state]{\It}, \solvar \models \pevolvein{\D{x}=\genDE{x}}{\ivr}
\end{align*}
The \m{\solvar \models \pevolvein {\D{x}=\genDE{x}}{\ivr}} condition checks that $\solvar$ is a solution of \m{\D{x}=\genDE{x}}, and that $\solvar(\zeta) \in \imodel{\I}{\ivr}$ for all $\zeta \in [0,T]$.
For any solution $\solvar$, the truncation $\soltrunc{\zeta} :[0,\zeta] \to \reals^n$ defined as \(\soltrunc{\zeta}(\tau)=\solvar(\tau)\) is also a solution.
Thus, $(\omega,\solvar(\zeta)) \in \iaccess[\pevolvein{\D{x}=\genDE{x}}{\ivr}]{\I}$ for all $\zeta \in [0,T]$.

Finally, \m{\imodels{\I}{\dbox{\alpha}{\fvar}}} iff \m{\imodels{\It}{\fvar}} for all $\iget[state]{\It}$ such that \m{\iaccessible[\alpha]{\I}{\It}}. Also, \m{\imodels{\I}{\ddiamond{\alpha}{\fvar}}} iff there is a $\iget[state]{\It}$ such that \m{\iaccessible[\alpha]{\I}{\It}} and \m{\imodels{\It}{\fvar}}.
A formula $\fvar$ is \emph{valid} iff it is true in all states, i.e., \(\imodels{\I}{\fvar}\) for all $\iget[state]{\I}$.

If formula \(P \limply \dbox{\pevolvein{\D{x}=\genDE{x}}{\ivr}}{P}\) is valid, then $P$ is called an \emph{invariant} of $\pevolvein{\D{x}=\genDE{x}}{\ivr}$. By the semantics, that is, from any initial state $\imodels{\I}{P}$, any solution $\solvar$ starting in $\iget[state]{\I}$, which does not leave the evolution domain $\imodel{\I}{\ivr}$, stays in $\imodel{\I}{P}$ for its \emph{entire duration}.

Figure~\ref{fig:exampleODE} suggests several invariants. The unit circle, $u^2+v^2=1$, is an equational invariant because the direction of flow on the circle is always tangential to the circle. The open unit disk $u^2+v^2 < 1$ is also invariant, because trajectories within the disk spiral towards the circle but never reach it. The region described by $u^2 \leq v^2 + \frac{9}{2}$ is invariant but needs a careful proof.

\subsection{Differentials and Lie Derivatives}
\label{subsec:differentials}

The study of invariants relates to the study of time derivatives of the quantities that the invariants involve.
Directly using time derivatives leads to numerous subtle sources of unsoundness, because they are not well-defined in arbitrary contexts (e.g., in isolated states).
\dL, instead, provides differential terms $\der{e}$ that have a local semantics in every state, can be used in any context, and can soundly be used for arbitrary logical manipulations~\cite{DBLP:journals/jar/Platzer17}.
Along an ODE \(\pevolve{\D{x}=\genDE{x}}\), the value of the differential term $\der{e}$ coincides with the time derivative $\D[t]{}$ of the value of $e$ \cite[Lem.\,35]{DBLP:journals/jar/Platzer17}.

The \emph{Lie derivative} of polynomial $p$ along ODE $\D{x}=\genDE{x}$ is:
\[ \lie[]{\genDE{x}}{p} \mdefeq \sum_{x_i\in\allvars} \frac{\partial p}{\partial x_i} f_i(x) = \nabla p \stimes f(x) \]
Unlike time derivatives, Lie derivatives can be written down syntactically. Unlike differentials, they still depend on the ODE context in which they are used. Along an ODE $\D{x}=\genDE{x}$, however, the value of Lie derivative $\lie[]{\genDE{x}}{p}$ coincides with that of the differential $\der{p}$, and \dL allows transformation between the two by proof. For this paper, we shall therefore directly use Lie derivatives, relying under the hood on \dL's axiomatic proof transformation from differentials \cite{DBLP:journals/jar/Platzer17}.
The operator $\lie[]{\genDE{x}}{\cdot}$ inherits the familiar sum and product rules of differentiation from corresponding axioms of differentials.

We reserve the notation $\lie[]{\genDE{x}}{\cdot}$ when used as an operator and simply write $\lied[]{\genDE{x}}{p}$ for $\lie[]{\genDE{x}}{p}$, because $\D{x}=\genDE{x}$ will be clear from the context.
We write $\lied[i]{\genDE{x}}{p}{}$ for the \emph{$i$-th Lie derivative of $p$} along $\D{x}=\genDE{x}$, where higher Lie derivatives are defined by iterating the Lie derivation operator. Since polynomials are closed under Lie derivation w.r.t.\ polynomial ODEs, all higher Lie derivatives of $p$ exist, and are also polynomials in the indeterminates $x$.
\begin{align*}
  \lied[0]{\genDE{x}}{p}{} \mdefeq p,\quad \lied[i+1]{\genDE{x}}{p}{} \mdefeq \lie{\genDE{x}}{\lied[i]{\genDE{x}}{p}{}}{}, \quad \lied[]{\genDE{x}}{p} \mdefeq \lied[1]{\genDE{x}}{p}{}
\end{align*}

\subsection{Axiomatization}
\label{subsec:axiomatizaton}

\irlabel{qear|\usebox{\Rval}}
\irlabel{notr|$\lnot$\rightrule}
\irlabel{notl|$\lnot$\leftrule}
\irlabel{orr|$\lor$\rightrule}
\irlabel{orl|$\lor$\leftrule}
\irlabel{andr|$\land$\rightrule}
\irlabel{andl|$\land$\leftrule}
\irlabel{implyr|$\limply$\rightrule}
\irlabel{implyl|$\limply$\leftrule}
\irlabel{equivr|$\lbisubjunct$\rightrule}
\irlabel{equivl|$\lbisubjunct$\leftrule}
\irlabel{id|id}
\irlabel{cut|cut}
\irlabel{weakenr|W\rightrule}
\irlabel{weakenl|W\leftrule}
\irlabel{existsr|$\exists$\rightrule}
\irlabel{existsrinst|$\exists$\rightrule}
\irlabel{alll|$\forall$\leftrule}
\irlabel{alllinst|$\forall$\leftrule}
\irlabel{allr|$\forall$\rightrule}
\irlabel{existsl|$\exists$\leftrule}
\irlabel{iallr|i$\forall$}
\irlabel{iexistsr|i$\exists$}

The reasoning principles for differential equations in \dL are stated as axioms in its uniform substitution calculus~\cite[Figure 3]{DBLP:journals/jar/Platzer17}. For ease of presentation in this paper, we shall work with a sequent calculus presentation with derived rule versions of these principles. The derivation of these rules from the axioms is shown in~\rref{app:diffaxiomatization}.

We assume a standard classical sequent calculus with all the usual rules for manipulating logical connectives and sequents, \eg, \irref{orl+andr}, and \irref{cut}.
The semantics of sequent \(\lsequent{\Gamma}{\fvar}\) is equivalent to \((\landfold_{A\in\Gamma} A) \limply \fvar\).
When we use an implicational or equivalence axiom, we omit the usual sequent manipulation steps and instead directly label the proof step with the axiom, giving the resulting premises accordingly \cite{DBLP:journals/jar/Platzer17}.
Because first-order real arithmetic is decidable~\cite{Bochnak1998}, we assume access to such a decision procedure, and label steps with \irref{qear} whenever they follow as a consequence of first-order real arithmetic. We use the \irref{existsr} rule over the reals, which allows us to supply a real-valued witness to an existentially quantified succedent. We mark with $\lclose$ the completed branches of sequent proofs.
A proof rule is \emph{sound} iff the validity of all its premises (above the rule bar) imply the validity of its conclusion (below rule bar).

\begin{theorem}[Differential equation axiomatization~\cite{DBLP:journals/jar/Platzer17}]
\label{thm:diffeqax}
The following sound proof rules derive from the axioms of \dL:

\begin{calculus}
\irlabel{dI|dI}
\dinferenceRule[dIeq|dI$_=$]{}
{\linferenceRule
  { \lsequent{\Gamma,\ivr}{p=0} \qquad
    \lsequent{\ivr}{\lied[]{\genDE{x}}{p}=0}
  }
  {\lsequent{\Gamma}{\dbox{\pevolvein{\D{x}=\genDE{x}}{\ivr}}{p=0}} }
}{}
\dinferenceRule[dIgeq|dI$_\cmp$]{}
{\linferenceRule
  { \lsequent{\Gamma,\ivr}{p\cmp0} \qquad
    \lsequent{\ivr}{\lied[]{\genDE{x}}{p}\geq0}
  }
  {\lsequent{\Gamma}{\dbox{\pevolvein{\D{x}=\genDE{x}}{\ivr}}{p\cmp0}} }
  \quad
}{where $\cmp$ is either $\geq$ or $>$}
\dinferenceRule[dC|dC]{}
{\linferenceRule
  {\lsequent{\Gamma}{\dbox{\pevolvein{\D{x}=\genDE{x}}{\ivr}}{\rcfvar}} \qquad
   \lsequent{\Gamma}{\dbox{\pevolvein{\D{x}=\genDE{x}}{\ivr \land \rcfvar}}{\rfvar}}
  }
  {\lsequent{\Gamma}{\dbox{\pevolvein{\D{x}=\genDE{x}}{\ivr}}{\rfvar}}}
}{}
\dinferenceRule[dW|dW]{}
{\linferenceRule
  {\lsequent{\ivr}{\rfvar}}
  {\lsequent{\Gamma}{\dbox{\pevolvein{\D{x}=\genDE{x}}{\ivr}}{\rfvar}}}
}{}

\dinferenceRule[dG|dG]{}
{\linferenceRule
  {\lsequent{\Gamma}{\lexists{y}{\dbox{\pevolvein{\D{x}=\genDE{x},\D{y}=a(x)\stimes y+b(x)}{\ivr}}{\rfvar}}}}
  {\lsequent{\Gamma}{\dbox{\pevolvein{\D{x}=\genDE{x}}{\ivr}}{\rfvar}}}
}{}
\end{calculus}
\end{theorem}
\noindent
\emph{Differential invariants} (\irref{dI}) reduce questions about invariance of $p=0,p\cmp0$ (globally along solutions of the ODE) to local questions about their respective Lie derivatives. We only show the two instances (\irref{dIeq+dIgeq}) of the more general \irref{dI} rule~\cite{DBLP:journals/jar/Platzer17} that will be used here. They internalize the mean value theorem\footnote{Note that for rule \irref{dIgeq}, we only require $\lied[]{\genDE{x}}{p}\geq0$ even for the $p > 0$ case.} \seeapp{\rref{app:diffaxiomatization}}.
These \emph{derived rules} are schematic because $\lied[]{\genDE{x}}{p}$ in their premises are dependent on the ODEs $\D{x}=\genDE{x}$. This exemplifies our point in \rref{subsec:differentials}: differentials allow the principles underlying \irref{dIeq+dIgeq} to be stated as axioms~\cite{DBLP:journals/jar/Platzer17} rather than complex, schematic proof rules.

\emph{Differential cut} (\irref{dC}) expresses that if we can separately prove that the system never leaves $\rcfvar$ while staying in $\ivr$ (the left premise), then we may additionally assume $\rcfvar$ when proving the postcondition $\rfvar$ (the right premise).
Once we have sufficiently enriched the evolution domain using \irref{dI+dC}, \emph{differential weakening} (\irref{dW}) allows us to drop the ODEs, and prove the postcondition $\rfvar$ directly from the evolution domain constraint $\ivr$.
Similarly, the following derived rule and axiom from \dL will be useful to manipulate postconditions:
\[\dinferenceRule[Mb|M${\dibox{\cdot}}$]{}
{\linferenceRule
  {\lsequent{\fvar_2}{\fvar_1} \qquad \lsequent{\Gamma}{\dbox{\alpha}{\fvar_2}}}
  {\lsequent{\Gamma}{\dbox{\alpha}{\fvar_1}}}
}{} \qquad \dinferenceRule[band|${\dibox{\cdot}\land}$]{}
{\linferenceRule[equiv]
  {\dbox{\alpha}{\fvar_1} \land \dbox{\alpha}{\fvar_2}}
  {\dbox{\alpha}{(\fvar_1 \land \fvar_2)}}
}{}\]
The \irref{Mb} monotonicity rule allows us to strengthen the postcondition to $\fvar_2$ if it implies $\fvar_1$. The derived axiom \irref{band} allows us to prove conjunctive postconditions separately, e.g.,~\irref{dIeq} derives from \irref{dIgeq} using \irref{band} with the equivalence $p=0 \lbisubjunct p\geq0 \land -p \geq 0$.

Even if \irref{dC} increases the deductive power over \irref{dI}, the deductive power increases even further~\cite{DBLP:journals/lmcs/Platzer12} with the \emph{differential ghosts} rule (\irref{dG}).
It allows us to add a \emph{fresh} variable $y$ to the system of equations. The main soundness restriction of \irref{dG} is that the new ODE must be linear\footnote{%
Linearity prevents the newly added equation from unsoundly restricting the duration of existence for solutions to the differential equations.} in $y$. This restriction is enforced by ensuring that $a(x),b(x)$ do not mention $y$. For our purposes, we will allow $\vec{y}$ to be vectorial, i.e., we allow the existing differential equations to be extended by a system that is linear in the new vector of variables $\vec{y}$. In this setting, $a(x)$ (resp.~$b(x)$) is a matrix (resp.~vector) of polynomials in $x$.

Adding differential ghost variables by \irref{dG} for the sake of the proof crucially allows us to express new relationships between variables along the differential equations. The next section shows how \irref{dG} can be used along with the rest of the \dL rules to prove a class of invariants satisfying Darboux-type properties. We exploit this increased deductive power in full in later sections.

\section{Darboux Polynomials}
\label{sec:darboux}

This section illustrates the use of \irref{dG} in proving invariance properties involving Darboux polynomials~\cite{darboux1878memoire}. A polynomial $\polyn{p}{x}$ is a \emph{Darboux polynomial} for the system $\D{x}=\genDE{x}$ iff it satisfies the polynomial identity $\lied[]{\genDE{x}}{p} = \polyn{g}{x}\polyn{p}{x}$ for some polynomial cofactor $\polyn{g}{x}$.

\subsection{Darboux Equalities}
\label{subsec:darbouxeq}

As in algebra, $\polynomials{\reals}{x}$ is the ring of polynomials in indeterminates $x$.
\begin{definition}[Ideal~\cite{Bochnak1998}]
\label{def:ideal}
The \emph{ideal} generated by the polynomials $p_1,\dots,p_s \in \polynomials{\reals}{x}$ is defined as the set of polynomials:
\[ \ideal{p_1,\dots,p_s} \mdefeq \{\text{\large$\Sigma$}_{i=1}^s g_ip_i \with g_i \in \polynomials{\reals}{x}\}\]
\end{definition}

Let us assume that $\polyn{p}{x}$ satisfies the Darboux polynomial identity $\lied[]{\genDE{x}}{p} = \polyn{g}{x}\polyn{p}{x}$.
Taking Lie derivatives on both sides, we get:
\[
\lied[2]{\genDE{x}}{p} = \lie[]{\genDE{x}}{\lied[]{\genDE{x}}{p}} = \lie[]{\genDE{x}}{gp}
= \lied[]{\genDE{x}}{g}p + g\lied[]{\genDE{x}}{p} = (\lied[]{\genDE{x}}{\polyn{g}{x}} + \polyn{g}{x}^2) \polyn{p}{x} \in \ideal{\polyn{p}{x}}
\]
By repeatedly taking Lie derivatives, it is easy to see that all higher Lie derivatives of $\polyn{p}{x}$ are contained in the ideal $\ideal{\polyn{p}{x}}$. Now, consider an initial state $\omega$ where $p$ evaluates to $\ivaluation{\I}{\polyn{p}{x}} = 0$, then:
\[\ivaluation{\I}{\polyn{\lied[]{\genDE{x}}{p}}{x}} = \ivaluation{\I}{\polyn{gp}{x}} = \ivaluation{\I}{\polyn{g}{x}} \cdot \ivaluation{\I}{\polyn{p}{x}} = 0\]
Similarly, because every higher Lie derivative of a Darboux polynomial is contained in the ideal generated by $\polyn{p}{x}$, all of them are simultaneously $0$ in state $\omega$. Thus, it should be the case\footnote{This requires the solution to be an analytic function of time, which is the case here.} that $\polyn{p}{x}=0$ stays invariant along solutions to the ODE starting at $\omega$. The above intuition motivates the following proof rule for invariance of $\polyn{p}{x}=0$:
\[
\dinferenceRule[dbx|dbx]{Darboux}
{\linferenceRule
  {\lsequent{\ivr} {\lied[]{\genDE{x}}{p} = \polyn{g}{x}\polyn{p}{x}}}
  {\lsequent{\polyn{p}{x}=0} {\dbox{\pevolvein{\D{x}=\genDE{x}}{\ivr}}{\polyn{p}{x}=0}}}
}{}
\]

Although we can derive \irref{dbx} directly, we opt for a detour through a proof rule for Darboux \emph{in}equalities instead. The resulting proof rule for invariant inequalities is crucially used in later sections.

\subsection{Darboux Inequalities}
\label{subsec:darbouxineq}

Assume that $\polyn{p}{x}$ satisfies a Darboux \emph{inequality} $\lied[]{\genDE{x}}{p} \geq \polyn{g}{x}\polyn{p}{x}$ for some cofactor polynomial $\polyn{g}{x}$. Semantically, in an initial state $\omega$ where $\ivaluation{\I}{\polyn{p}{x}} \geq 0$, an application of Gr\"onwall's lemma~\cite[\S29.VI]{Walter1998,DBLP:journals/mathann/Gronwall19} allows us to conclude that $\polyn{p}{x} \geq 0$ stays invariant along solutions starting at $\omega$. Indeed, if $\polyn{p}{x}$ is a Darboux polynomial with cofactor $\polyn{g}{x}$, then it satisfies both Darboux inequalities $\lied[]{\genDE{x}}{p} \geq \polyn{g}{x}\polyn{p}{x}$ and $\lied[]{\genDE{x}}{p} \leq \polyn{g}{x}\polyn{p}{x}$, which yields an alternative semantic argument for the invariance of $p=0$.
In our derivations below, we show that these Darboux invariance properties can be proved purely syntactically using \irref{dG}.

\begin{lemma}[Darboux (in)equalities are differential ghosts] \label{lem:Darboux}
The proof rules for Darboux equalities (\irref{dbx}) and inequalities (\irref{dbxineq}) derive from \irref{dG} (and \irref{dI+dC}):
\[
\dinferenceRule[dbxineq|dbx${\cmp}$]{Darboux inequality}
{\linferenceRule
  {\lsequent{\ivr} {\lied[]{\genDE{x}}{p}\geq \polyn{g}{x}\polyn{p}{x}}}
  {\lsequent{\polyn{p}{x}\cmp0} {\dbox{\pevolvein{\D{x}=\genDE{x}}{\ivr}}{\polyn{p}{x}\cmp0}}}
  ~~
}{\text{where $\cmp$ is either $\geq$ or $>$}}
\]
\end{lemma}
\begin{proof}
We first derive \irref{dbxineq}, let \textcircled{1} denote the use of its premise, and \textcircled{2} abbreviate the right premise in the following derivation.
{\footnotesize%
\begin{sequentdeduction}[array]
\linfer[dG]
  {\linfer[Mb+existsr]
    {\linfer[dC]
      {
        \lsequent{\polyn{p}{x}{\cmp}0,y{>}0} {\dbox{\pevolvein{\D{x}=\genDE{x}\syssep\D{y}=-\polyn{g}{x}y}{\ivr\land y>0}}{\polyn{p}{x}y\cmp0}}
      \qquad\hfill
        \textcircled{2}
      }
      {\lsequent{\polyn{p}{x}\cmp0,y>0} {\dbox{\pevolvein{\D{x}=\genDE{x}\syssep\D{y}=-\polyn{g}{x}y}{\ivr}}{(y>0 \land \polyn{p}{x}y\cmp0)}}}
    }%
    {\lsequent{\polyn{p}{x}\cmp0} {\lexists{y}{\dbox{\pevolvein{\D{x}=\genDE{x}\syssep\D{y}=-\polyn{g}{x}y}{\ivr}}{\polyn{p}{x}\cmp0}}}}
  }
  {\lsequent{\polyn{p}{x}\cmp0} {\dbox{\pevolvein{\D{x}=\genDE{x}}{\ivr}}{\polyn{p}{x}\cmp0}}}
\end{sequentdeduction}
}%
In the first \irref{dG} step, we introduce a new ghost variable $y$ satisfying a carefully chosen differential equation \(\D{y}=-\polyn{g}{x}y\) as a counterweight.
Next, \irref{existsr} allows us to pick an initial value for $y$. We simply pick any $y>0$. We then observe that in order to prove $\polyn{p}{x}\cmp0$, it suffices to prove the stronger invariant $y>0 \land \polyn{p}{x}y\cmp0$, so we use the monotonicity rule \irref{Mb} to strengthen the postcondition. Next, we use \irref{dC} to first prove $y>0$ in \textcircled{2}, and assume it in the evolution domain constraint in the left premise. This sign condition on $y$ is crucially used when we apply \textcircled{1} in the proof for the left premise:
{\footnotesize\renewcommand{\arraystretch}{1.3}%
\begin{sequentdeduction}[array]
\linfer[dI]
{
  \linfer[qear]{\lclose}
  {\lsequent{\polyn{p}{x}{\cmp}0,y{>}0}{\polyn{p}{x}y\cmp0}}
  !
  \linfer[cut]
    {\textcircled{1}!
    \linfer[qear]
      {\lclose}
      {\lsequent{\lied[]{\genDE{x}}{p}\geq \polyn{g}{x}\polyn{p}{x}, y>0} {\lied[]{\genDE{x}}{p}y-\polyn{g}{x}y\polyn{p}{x}\geq0}}
    }%
    {\lsequent{\ivr \land y>0} {\lied[]{\genDE{x}}{p}y-\polyn{g}{x}y\polyn{p}{x}\geq0}}
}
{\lsequent{\polyn{p}{x}{\cmp}0,y{>}0} {\dbox{\pevolvein{\D{x}=\genDE{x}\syssep\D{y}=-\polyn{g}{x}y}{\ivr\land y>0}}{\polyn{p}{x}y\cmp0}}}
\end{sequentdeduction}
}%
We use \irref{dI} to prove the inequational invariant $\polyn{p}{x}y\cmp0$; its left premise is a consequence of real arithmetic. On the right premise, we compute the Lie derivative of $\polyn{p}{x}y$ using the usual product rule as follows:
\[\lie[]{\genDE{x},-\polyn{g}{x}y}{\polyn{p}{x}y} = \lie[]{\genDE{x},-\polyn{g}{x}y}{\polyn{p}{x}}y + \polyn{p}{x}\lie[]{\genDE{x},-\polyn{g}{x}y}{y} = \lied[]{\genDE{x}}{p}y-\polyn{g}{x}y\polyn{p}{x} \]
 We complete the derivation by cutting in the premise of \irref{dbxineq} (\textcircled{1}). Note that the differential ghost \m{\D{y}=-\polyn{g}{x}y} was precisely chosen so that the final arithmetic step closes trivially.

We continue on premise \textcircled{2} with a second ghost \(\D{z}=\frac{\polyn{g}{x}}{2}z\):
{\footnotesize\renewcommand{\arraystretch}{1.3}%
\begin{sequentdeduction}[array]
\linfer[dG]
{\linfer[Mb+existsr]
  {\linfer[dI]
    {\linfer[qear]
          {\lclose}
        {\lsequent{\ivr} {(-\polyn{g}{x}y)z^2+y(2z(\frac{\polyn{g}{x}}{2}z))=0}}
    }%
    {\lsequent{yz^2=1} {\dbox{\pevolvein{\D{x}=\genDE{x}\syssep\D{y}=-\polyn{g}{x}y\syssep\D{z}=\frac{\polyn{g}{x}}{2}z}{\ivr}}{yz^2=1}}}
  }%
  {\lsequent{y>0} {\lexists{z}{\dbox{\pevolvein{\D{x}=\genDE{x}\syssep\D{y}=-\polyn{g}{x}y\syssep\D{z}=\frac{\polyn{g}{x}}{2}z}{\ivr}}{y>0}}}}
}%
{\lsequent{y>0} {\dbox{\pevolvein{\D{x}=\genDE{x}\syssep\D{y}=-\polyn{g}{x}y}{\ivr}}{y>0}}}
\end{sequentdeduction}
}%
This derivation is analogous to the one for the previous premise. In the \irref{Mb+existsr} step, we observe that if $y>0$ initially, then there exists $z$ such that $yz^2=1$. Moreover, $yz^2=1$ is sufficient to imply $y > 0$ in the postcondition. The differential ghost \m{\D{z}=\frac{\polyn{g}{x}}{2}z} is constructed so that $yz^2=1$ can be proved invariant \emph{along the differential equation}.

The \irref{dbx} proof rule derives from rule \irref{dbxineq} using the equivalence \m{p=0 \lbisubjunct p\geq0 \land -p \geq 0} and derived axiom \irref{band}:
{\footnotesize%
\begin{sequentdeduction}[default]
\linfer[Mb+qear]{
\linfer[band+andr]
{
  \linfer[dbxineq]{
  \linfer[qear]{
    \lsequent{\ivr} {\lied[]{\genDE{x}}{p} = \polyn{g}{x}\polyn{p}{x}}
  }
    {\lsequent{\ivr} {\lied[]{\genDE{x}}{p}\geq \polyn{g}{x}\polyn{p}{x}}}
  }
  {\lsequent{\polyn{p}{x}\geq 0} {\dbox{\pevolvein{\D{x}=\genDE{x}}{\ivr}}{\polyn{p}{x}\geq 0}}}
! \qquad
  \linfer[dbxineq]{
  \linfer[qear]{
    \lsequent{\ivr} {\lied[]{\genDE{x}}{p} = \polyn{g}{x}\polyn{p}{x}}
  }
    {\lsequent{\ivr} {\lied[]{\genDE{x}}{(-p)}\geq -\polyn{g}{x}\polyn{p}{x}}}
  }
  {\lsequent{-\polyn{p}{x}\geq 0} {\dbox{\pevolvein{\D{x}=\genDE{x}}{\ivr}}{{-}\polyn{p}{x}\geq 0}}}
}
  {\lsequent{\polyn{p}{x}\geq 0 \land -\polyn{p}{x}\geq 0 } {\dbox{\pevolvein{\D{x}=\genDE{x}}{\ivr}}{(\polyn{p}{x}\geq 0 \land -\polyn{p}{x}\geq 0)}}}
}
  {\lsequent{\polyn{p}{x}=0} {\dbox{\pevolvein{\D{x}=\genDE{x}}{\ivr}}{\polyn{p}{x}=0}}}
\end{sequentdeduction}
}%
\end{proof}

\begin{example}[Proving continuous properties in \dL]
\label{ex:continuousproperties}
In the running example, \irref{dbxineq} directly proves that the open disk $1-u^2-v^2 > 0$ is an invariant for $\alpha_e$ using cofactor $g\mnodefeq-\frac{1}{2}(u^2+v^2)$:
{\footnotesize%
\begin{sequentdeduction}[default]
\linfer[dbxineq]{
  \linfer[qear]{\lclose}{\lsequent{} {\lie[]{\alpha_e}{1 - u^2 - v^2} \geq -\frac{1}{2}(u^2+v^2) (1 - u^2 - v^2)}}
}
  {\lsequent{1 - u^2 - v^2 > 0} {\dbox{\alpha_e}{1 - u^2 - v^2 > 0}}}
\end{sequentdeduction}
}%

\begin{figure}
\centering
\includegraphics[width=\columnwidth]{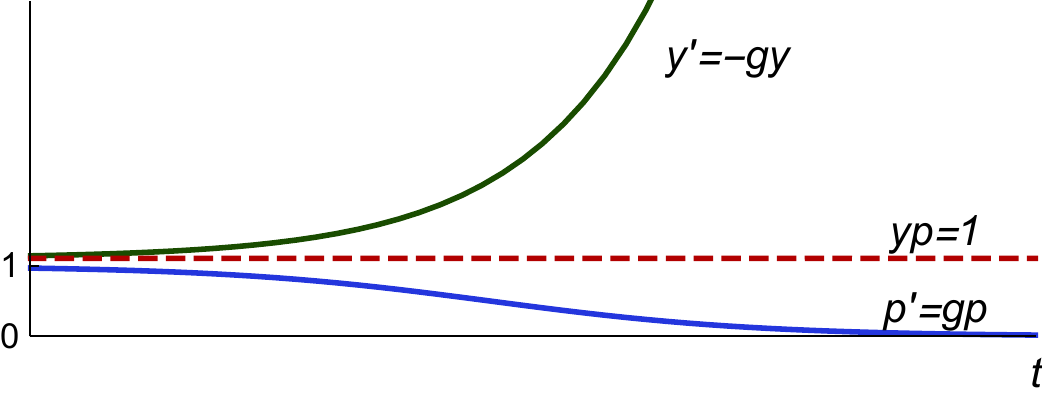}
\caption{The differential ghost $y'=-gy$ (in green) balances out $p'=gp$ (in blue) so that the value of $yp$ (the red dashed line) remains constant at $1$. The horizontal axis tracks the evolution of time $t$.}
\label{fig:darboux}
\end{figure}

Figure~\ref{fig:exampleODE} indicated that trajectories in the open disk spiral towards $1-u^2-v^2=0$, i.e., they evolve towards leaving the invariant region. Intuitively, this makes a direct proof of invariance difficult. The proof of \irref{dbxineq} instead introduces the differential ghost $y'=-gy$. Its effect for our example is illustrated in \rref{fig:darboux}, which plots the value of $p \mnodefeq 1-u^2-v^2$ and ghost $y$ along the solution starting from the point $(\frac{1}{8},\frac{1}{8})$. Although $p$ decays towards $0$, the ghost $y$ balances this by growing away from $0$ so that $yp$ remains constant at its initial value $1$, which implies that $p$ never reaches $0$.
\end{example}

These derivations demonstrate the clever use of differential ghosts. In fact, we have already exceeded the deductive power of \irref{dI+dC} because the formula $y > 0 \limply \dbox{\pevolve{\D{y}=-y}}{y > 0}$ is valid but not provable with \irref{dI+dC} alone but needs a \irref{dG}~\cite{DBLP:journals/lmcs/Platzer12}. It is a simple consequence of \irref{dbxineq}, since the polynomial $y$ satisfies the Darboux equality $\lied[]{}{y} = -y$ with cofactor $-1$. For brevity, we showed the same derivation for both $\geq$ and $>$ cases of \irref{dbxineq} even though the latter case only needs one ghost. Similarly, \irref{dbx} derives directly using two ghosts rather than the four ghosts incurred using \irref{band}. All of these cases, however, only introduce one differential ghost at a time. In the next section, we exploit the full power of \emph{vectorial} \irref{dG}.

\section{Algebraic Invariants}
\label{sec:alginvs}

We now consider polynomials that are not Darboux for the given differential equations, but instead satisfy a \emph{differential radical property}~\cite{DBLP:conf/tacas/GhorbalP14} with respect to its higher Lie derivatives. Let $g_i$ be cofactor polynomials, $N \geq 1$, assume that $p$ satisfies the polynomial identity:
\begin{equation}
\lied[N]{\genDE{x}}{p} = \sum_{i=0}^{N-1} g_i \lied[i]{\genDE{x}}{p}
\label{eq:differential-rank}
\end{equation}

With the same intuition, again take Lie derivatives on both sides:
\begin{align*}
\lied[N+1]{\genDE{x}}{p} =& \lie[]{\genDE{x}}{\lied[N]{\genDE{x}}{p}} = \lie[]{\genDE{x}}{\sum_{i=0}^{N-1} g_i \lied[i]{\genDE{x}}{p}} = \sum_{i=0}^{N-1}  \lie[]{\genDE{x}}{g_i \lied[i]{\genDE{x}}{p}}\\
=& \sum_{i=0}^{N-1}  \left(\lied[]{\genDE{x}}{g_i} \lied[i]{\genDE{x}}{p} + g_i\lied[i+1]{\genDE{x}}{p}\right) \in \ideal{p,\lied[]{\genDE{x}}{p}, \dots, \lied[N-1]{\genDE{x}}{p}}
\end{align*}
In the last step, ideal membership follows by observing that, by \rref{eq:differential-rank}, $\lied[N]{\genDE{x}}{p}$ is contained in the ideal generated by the lower Lie derivatives. By repeatedly taking Lie derivatives on both sides, we again see that $\lied[N]{\genDE{x}}{p}, \lied[N+1]{\genDE{x}}{p}, \dots$ are all contained in the ideal $\ideal{p,\lied[]{\genDE{x}}{p}, \dots, \lied[N-1]{\genDE{x}}{p}}$. Thus, if we start in state $\omega$ where $\ivaluation{\I}{\polyn{p}{x}},\ivaluation{\I}{\lied[]{\genDE{x}}{p}},\dots, \ivaluation{\I}{\lied[N-1]{\genDE{x}}{p}}$ all simultaneously evaluate to $0$, then $\polyn{p}{x}=0$ (\emph{and all higher Lie derivatives}) must stay invariant along (analytic) solutions to the ODE.

This section shows how to axiomatically prove this invariance property using (vectorial) \irref{dG}. We shall see at the end of the section that this allows us to prove \emph{all} true algebraic invariants.

\subsection{Vectorial Darboux Equalities}
\label{subsec:vecdarbouxeq}

We first derive a vectorial generalization of the Darboux rule \irref{dbx}, which will allow us to derive the rule for algebraic invariants as a special case by exploiting a vectorial version of \rref{eq:differential-rank}.
Let us assume that the $n$-dimensional vector of polynomials $\vecpolyn{p}{x}$ satisfies the vectorial polynomial identity $\lied[]{\genDE{x}}{\vecpolyn{p}{x}}=\matpolyn{G}{x}\itimes\vecpolyn{p}{x}$, where $\matpolyn{G}{x}$ is an $n\times n$ matrix of polynomials, and $\lied[]{\genDE{x}}{\vec{p}}$ denotes component-wise Lie derivation of $\vecpolyn{p}{x}$.
If all components of $\vecpolyn{p}{x}$ start at 0, then they stay 0 along \m{\pevolve{\D{x}=\genDE{x}}}.

\begin{lemma}[Vectorial Darboux equalities are vectorial ghosts] \label{lem:vdbx}
The vectorial Darboux proof rule derives from vectorial \irref{dG} (and \irref{dI+dC}).
\[
\dinferenceRule[vdbx|vdbx]{vectorial Darboux}
{\linferenceRule
  {\lsequent{\ivr} {\lied[]{\genDE{x}}{\vec{p}}=\matpolyn{G}{x}\itimes\vecpolyn{p}{x}}}
  {\lsequent{\vecpolyn{p}{x}=0} {\dbox{\pevolvein{\D{x}=\genDE{x}}{\ivr}}{\vecpolyn{p}{x}=0}}}
}{}
\]
\end{lemma}
\begin{proof}
Let $\matpolyn{G}{x}$ be an $n \times n$ matrix of polynomials, and $\vecpolyn{p}{x}$ be an $n$-dimensional vector of polynomials satisfying the premise of \irref{vdbx}.

First, we develop a proof that we will have occasion to use repeatedly. This proof adds an $n$-dimensional vectorial ghost \m{\D{\vec{y}}=-\transpose{\matpolyn{G}{x}}\vec{y}} such that the vanishing of the scalar product, \ie, $\vecpolyn{p}{x}\stimes\vec{y} = 0$, is invariant. In the derivation below, we suppress the initial choice of values for $\vec{y}$ till later.
\textcircled{1} denotes the use of the premise of \irref{vdbx}.
In the \irref{dC} step, we mark the remaining open premise with \textcircled{2}.
{\footnotesize\renewcommand*{\arraystretch}{1.3}%
\begin{sequentdeduction}[array]
\linfer[dG]
{\linfer[dC]
  {\textcircled{2} ! \linfer[dI]
      {\linfer[cut]
        {\textcircled{1}
        !\linfer[qear]%
          {\linfer[qear]
            {\lclose}
            {\lsequent{\ivr} {\matpolyn{G}{x}\vecpolyn{p}{x}\stimes\vec{y}-\matpolyn{G}{x}\vecpolyn{p}{x}\stimes\vec{y}=0}}
          }%
          {\lsequent{\ivr} {\matpolyn{G}{x}\vecpolyn{p}{x}\stimes\vec{y}-\vecpolyn{p}{x}\stimes{\transpose{\matpolyn{G}{x}}\vec{y}}=0}}
        }%
        {\lsequent{\ivr} {\lied[]{\genDE{x}}{\vecpolyn{p}{x}}\stimes\vec{y}-\vecpolyn{p}{x}\stimes{\transpose{\matpolyn{G}{x}}\vec{y}}=0}}
      }%
    {\lsequent{\vecpolyn{p}{x}\stimes\vec{y}=0} {\dbox{\pevolvein{\D{x}=\genDE{x}\syssep\D{\vec{y}}=-\transpose{\matpolyn{G}{x}}\vec{y}}{\ivr}}{\vecpolyn{p}{x}\stimes\vec{y}=0}}}
  }
  {\lsequent{\vecpolyn{p}{x}=0} {\lexists{\vec{y}}{\dbox{\pevolvein{\D{x}=\genDE{x}\syssep\D{\vec{y}}=-\transpose{\matpolyn{G}{x}}\vec{y}}{\ivr}}{\vecpolyn{p}{x}=0}}}}
}
{\lsequent{\vecpolyn{p}{x}=0} {\dbox{\pevolvein{\D{x}=\genDE{x}}{\ivr}}{\vecpolyn{p}{x}=0}}}
\end{sequentdeduction}
}%
The open premise \textcircled{2} now includes $\vecpolyn{p}{x}\stimes\vec{y}=0$ in the evolution domain:
\[
\textcircled{2}\quad
\lsequent{\vecpolyn{p}{x}=0} {\dbox{\pevolvein{\D{x}=\genDE{x}\syssep\D{\vec{y}}=-\transpose{\matpolyn{G}{x}}\vec{y}}{\ivr \land \vecpolyn{p}{x}\stimes\vec{y}=0}}{\vecpolyn{p}{x}=0}} \]

So far, the proof is similar to the first ghost step for \irref{dbxineq}. Unfortunately, for $n>1$, the postcondition \(\vecpolyn{p}{x}=0\) does \emph{not} follow from the evolution domain constraint \(\vecpolyn{p}{x}\stimes\vec{y}=0\) even when $\vec{y}\neq0$, because \(\vecpolyn{p}{x}\stimes\vec{y}=0\) merely implies that $\vecpolyn{p}{x}$ and $\vec{y}$ are orthogonal, not that $\vecpolyn{p}{x}$ is 0.

The idea is to repeat the above proof sufficiently often to obtain an entire matrix $Y$ of independent differential ghost variables such that both \(Y \vecpolyn{p}{x}=0\) and \(\determinant(Y)\neq0\) can be proved invariant.\footnote{For a square matrix of polynomials $Y$, $\determinant(Y)$ is its determinant, $\trace(Y)$ its trace, and, of course, $\transpose{Y}$ is its transpose.} The latter implies that $Y$ is invertible, so that \(Y \vecpolyn{p}{x}=0\) implies \(\vecpolyn{p}{x}=0\).
The matrix $Y$ is obtained by repeating the derivation above on premise \textcircled{2}, using \irref{dG} to add $n$ copies of the ghost vectors, $\vec{y}_1, \dots, \vec{y}_n$, each satisfying the ODE system $\vec{y}_i' = -\transpose{\matpolyn{G}{x}}\vec{y}_i$. By the derivation above, each $\vec{y}_i$ satisfies the provable invariant $\vec{y}_i \cdot \vecpolyn{p}{x} = 0$, or more concisely:
{\footnotesize%
\[
\overbrace{
\left(\begin{array}{cccc}
y_{11} & y_{12} & \dots & y_{1n} \\
y_{21} & y_{22} & \dots & y_{2n} \\
\vdots & \vdots & \ddots & \vdots \\
y_{n1} & y_{n2} & \dots & y_{nn} \\
\end{array}\right)}^{Y}
\itimes
\overbrace{\begin{avector}\polyn{p_1}{x} \\ \polyn{p_2}{x} \\ \vdots \\ \polyn{p_n}{x} \end{avector}}^{\vecpolyn{p}{x}} = 0
\]}%

Streamlining the proof, we first perform the \irref{dG} steps that add the $n$ ghost vectors $\vec{y}_i$, before combining \irref{band+dI} to prove:
\[
{\lsequent{\vecpolyn{p}{x}{=}0} {\dbox{\pevolvein{\D{x}=\genDE{x}\syssep\D{\vec{y}_1}=-\transpose{\matpolyn{G}{x}}\vec{y}_1\syssep \dots \syssep\D{\vec{y}_n}=-\transpose{\matpolyn{G}{x}}\vec{y}_n}{\ivr}}{\landfold_{i=1}^n \vec{y}_i\stimes\vecpolyn{p}{x}{=}0}}}
\]
which we summarize using the above matrix notation as:
\[
\textcircled{3}\quad
{\lsequent{\vecpolyn{p}{x}=0} {\dbox{\pevolvein{\D{x}=\genDE{x}\syssep Y' = - Y\matpolyn{G}{x}}{\ivr}}{Y\vecpolyn{p}{x}=0}}}
\]
because when $Y'$ is the component-wise derivative of $Y$, all the differential ghost equations are summarized as $Y' = - Y\matpolyn{G}{x}$.\footnote{The entries on both sides of the differential equations satisfy $Y'_{ij} = (\vec{y}_{ij})' = -(\transpose{\matpolyn{G}{x}}\vec{y}_i)_j = -\sum_{k=1}^{n} \transpose{\matpolyn{G}{x}}_{jk} y_{ik} = -\sum_{k=1}^{n} \matpolyn{G}{x}_{kj} y_{ik} = -\sum_{k=1}^{n} y_{ik} \matpolyn{G}{x}_{kj} = -(Y \matpolyn{G}{x})_{ij}$.}
Now that we have the invariant $Y\vecpolyn{p}{x}=0$ from \textcircled{3}, it remains to prove the invariance of \m{\determinant(Y)>0} to complete the proof.

Since $Y$ only contains $y_{ij}$ variables, $\determinant(Y)$ is a polynomial term in the variables $y_{ij}$.
These $y_{ij}$ are ghost variables that we have introduced by \irref{dG}, and so we are free to pick their initial values. For convenience, we shall pick initial values forming the identity matrix $Y=\idmatrix$, so that $\determinant(Y) = \determinant(\idmatrix) = 1 > 0$ is true initially.

In order to show that $\determinant(Y)>0$ is an invariant, we use rule \irref{dbxineq} with the critical polynomial identity \m{\lied[]{Y'=-YG}{\determinant(Y)} = - \trace{(\matpolyn{G}{x})}\determinant(Y)} that follows from Liouville's formula~\cite[\S15.III]{Walter1998}, where the Lie derivatives are taken with respect to the extended system of equations $\D{x}=\genDE{x}\syssep Y' = -YG$. For completeness, we give an arithmetic proof of Liouville's formula in \rref{app:alginvariants}.
Thus, $\determinant(Y)$ is a Darboux polynomial over the variables $y_{ij}$, with polynomial cofactor \(-\trace(\matpolyn{G}{x})\):
\[
\textcircled{4}\qquad\qquad
\linfer[dbxineq]
{\lsequent{\ivr} {\lied[]{Y'=-YG}{\determinant(Y)} = -\trace(\matpolyn{G}{x}) \determinant(Y)}}
{\lsequent{\determinant(Y){>}0} {\dbox{\pevolvein{\D{x}=\genDE{x}\syssep Y' = - Y\matpolyn{G}{x}}{\ivr}}{\determinant(Y){>}0}}}
\]

Combining \textcircled{3} and \textcircled{4} completes the derivation for the invariance of $\vecpolyn{p}{x} = 0$. We start with the \irref{dG} step and abbreviate the ghost matrix.
{\footnotesize
\renewcommand*{\arraystretch}{1.3}
\begin{sequentdeduction}[array]
\linfer[dG]{
\linfer[]{
  \lsequent{\vecpolyn{p}{x}=0} {\lexists{Y}\dbox{\pevolvein{\D{x}=\genDE{x}\syssep Y' = - Y\matpolyn{G}{x}}{\ivr}}{\vecpolyn{p}{x}=0}}
}
{\lsequent{\vecpolyn{p}{x}=0} {\lexists{\vec{y}_1,\dots,\vec{y}_n}\dbox{\pevolvein{\D{x}=\genDE{x}\syssep\D{\vec{y}_1}=-\transpose{\matpolyn{G}{x}}\vec{y}_1\syssep \dots \syssep\D{\vec{y}_n}=-\transpose{\matpolyn{G}{x}}\vec{y}_n}{\ivr}}{\vecpolyn{p}{x}=0}}}
}
  {\lsequent{\vecpolyn{p}{x}=0} {\dbox{\pevolvein{\D{x}=\genDE{x}}{\ivr}}{\vecpolyn{p}{x}=0}}}
\end{sequentdeduction}
}%
Now, we carry out the rest of the proof as outlined earlier.
{\footnotesize%
\begin{sequentdeduction}[array]
\linfer[existsr]{
\linfer[cut]{
\linfer[dC]{
\linfer[dC]{
\linfer[dW]{
\linfer[qear]{
  \lclose
}
  {\lsequent{\ivr \land Y \vecpolyn{p}{x}{=}0 \land \determinant(Y){>}0}{\vecpolyn{p}{x}{=}0}}
}
  {\textcircled{4}\hfill
  \lsequent{\vecpolyn{p}{x}{=}0} {\dbox{\pevolvein{\D{x}=\genDE{x}\syssep Y' = - Y\matpolyn{G}{x}}{\ivr \land Y \vecpolyn{p}{x}{=}0 \land \determinant(Y){>}0}}{\vecpolyn{p}{x}{=}0}}}
}
  {\textcircled{3}\hfill
  \lsequent{\vecpolyn{p}{x}{=}0, \determinant(Y){>}0} {\dbox{\pevolvein{\D{x}=\genDE{x}\syssep Y' = - Y\matpolyn{G}{x}}{\ivr \land Y \vecpolyn{p}{x} = 0}}{\vecpolyn{p}{x}{=}0}}}
}
  {\lsequent{\vecpolyn{p}{x}{=}0, \determinant(Y){>}0} {\dbox{\pevolvein{\D{x}=\genDE{x}\syssep Y' = - Y\matpolyn{G}{x}}{\ivr}}{\vecpolyn{p}{x}{=}0}}}
}
  {\lsequent{\vecpolyn{p}{x}{=}0, Y=\idmatrix} {\dbox{\pevolvein{\D{x}=\genDE{x}\syssep Y' = - Y\matpolyn{G}{x}}{\ivr}}{\vecpolyn{p}{x}{=}0}}}
}
  {\lsequent{\vecpolyn{p}{x}{=}0} {\lexists{Y}\dbox{\pevolvein{\D{x}=\genDE{x}\syssep Y' = - Y\matpolyn{G}{x}}{\ivr}}{\vecpolyn{p}{x}{=}0}}}
\end{sequentdeduction}
}%
The order of the differential cuts \textcircled{3} and \textcircled{4} is irrelevant.
\end{proof}

Since $\det{(Y)} \neq 0$ is invariant, the $n \times n$ ghost matrix $Y$ in this proof corresponds to a basis for $\reals^n$ that \emph{continuously evolves} along the differential equations. To see what $Y$ does geometrically, let $\vec{p}_0$ be the initial values of $\vec{p}$, and $Y=\idmatrix$ initially. With our choice of $Y$, a variation of step \textcircled{3} in the proof shows that $Y\vecpolyn{p}{x}{=}\vec{p}_0$ is invariant. Thus, the evolution of $Y$ \emph{balances out} the evolution of $\vec{p}$, so that $\vec{p}$ remains constant with respect to the continuously evolving change of basis $Y^{-1}$. This generalizes the intuition illustrated in~\rref{fig:darboux} to the $n$-dimensional case. Crucially, differential ghosts let us soundly express this time-varying change of basis purely axiomatically.

\subsection{Differential Radical Invariants}
\label{subsec:diffradicalinvariants}

We now return to polynomials $p$ satisfying property \rref{eq:differential-rank}, and show how to prove $p=0$ invariant using an instance of \irref{vdbx}.

\begin{theorem}[Differential radical invariants are vectorial Darboux] \label{thm:DRI}
The differential radical invariant proof rule derives from \irref{vdbx} (which in turn derives from vectorial \irref{dG}).
\[
\dinferenceRule[dRI|dRI]{differential radical invariants}
{\linferenceRule
  {\lsequent{\Gamma,\ivr} {\landfold_{i=0}^{N-1}  \lied[i]{\genDE{x}}{p} = 0} & \lsequent{\ivr} {\lied[N]{\genDE{x}}{p} = \sum_{i=0}^{N-1} g_i \lied[i]{\genDE{x}}{p}}}
  {\lsequent{\Gamma} {\dbox{\pevolvein{\D{x}=\genDE{x}}{\ivr}}{\polyn{p}{x}=0}}}
}{}
\]
\end{theorem}
\begin{proofsketch}[app:alginvariants]
Rule \irref{dRI} derives from rule \irref{vdbx} with:
{\footnotesize%
\[\matpolyn{G}{x}= \left(\begin{array}{ccccc}
0      & 1      & 0      & \dots & 0      \\
0      & 0      & \ddots & \ddots & \vdots \\
\vdots & \vdots & \ddots & \ddots & 0      \\
0      & 0      & \dots & 0      & 1\\
g_0    & g_1    & \dots & g_{N-2}& g_{N-1} \end{array}\right),
\quad
\vecpolyn{p}{x} = \left(\begin{array}{l}p\\ \lied[1]{\genDE{x}}{p}\\ \vdots \\\lied[N-2]{\genDE{x}}{p}\\ \lied[N-1]{\genDE{x}}{p}\end{array}\right)
\]}%
The matrix $\matpolyn{G}{x}$ has $1$ on its superdiagonal, and the $g_i$ cofactors in the last row.
The left premise of \irref{dRI} is used to show $\vecpolyn{p}{x} = 0$ initially, while the right premise is used to show the premise of \irref{vdbx}.
\end{proofsketch}

\subsection{Completeness for Algebraic Invariants}
\label{subsec:completenessalg}

Algebraic formulas are formed from finite conjunctions and disjunctions of polynomial equations, but, over $\reals$, can be normalized to a single equation $p=0$ using the real arithmetic equivalences:
\begin{align*}
p=0 \land q =0 \lbisubjunct p^2 + q^2 = 0,\quad p=0 \lor q =0 \lbisubjunct pq = 0
\end{align*}

The key insight behind completeness of \irref{dRI} is that higher Lie derivatives stabilize.
Since the polynomials $\polynomials{\reals}{x}$ form a Noetherian ring, for every polynomial $p$ and polynomial ODE \m{\D{x}=\genDE{x}},
there is a smallest natural number\footnote{%
The only polynomial satisfying \rref{eq:differential-rank} for $N=0$ is the 0 polynomial, which gives correct but trivial  invariants $0=0$ for any system (and 0 can be considered to be of rank 1).}
$N{\geq}1$ called \emph{rank}~\cite{novikov1999trajectories,DBLP:conf/tacas/GhorbalP14} such that $p$ satisfies the polynomial identity \rref{eq:differential-rank} for some cofactors $g_i$.
This $N$ is computable by successive ideal membership checks \cite{DBLP:conf/tacas/GhorbalP14}.

Thus, some suitable rank at which the right premise of \irref{dRI} proves exists for any polynomial $p$.\footnote{%
\rref{thm:DRI} shows $\ivr$ can be assumed when proving ideal membership of $\lied[N]{\genDE{x}}{p}$. A finite rank exists either way, but assuming $\ivr$ may reduce the number of higher Lie derivatives of $p$ that need to be considered.}
The succedent in the remaining left premise of \irref{dRI} entails that \emph{all} Lie derivatives evaluate to zero.

\begin{definition}[Differential radical formula]
The \emph{differential radical formula} $\sigliedzero{\genDE{x}}{p}$ of a polynomial $p$ with rank $N{\geq}1$ from \rref{eq:differential-rank} and Lie derivatives with respect to \m{\D{x}=\genDE{x}} is defined to be:
\begin{align*}
\sigliedzero{\genDE{x}}{p} ~\mdefequiv~ \landfold_{i=0}^{N-1}  \lied[i]{\genDE{x}}{p} = 0
\end{align*}
\end{definition}
The completeness of \irref{dRI} can be proved semantically~\cite{DBLP:conf/tacas/GhorbalP14}. However, using the extensions developed in \rref{sec:extaxioms}, we derive the following characterization for algebraic invariants axiomatically.

\begin{theorem}[Algebraic invariant completeness]
\label{thm:algcomplete}
The following is a derived axiom in \dL when $\ivr$ characterizes an open set:
\[\dinferenceRule[DRI|DRI]{differential radical invariant axiom}
{\linferenceRule[equiv]
  {\big(\ivr \limply \sigliedzero{\genDE{x}}{p}\big)}
  {\axkey{\dbox{\pevolvein{\D{x}=\genDE{x}}{\ivr}}{p=0}}}
}{}\]
\end{theorem}
\begin{proofsketch}[app:alginvariants]
The ``$\lylpmi$" direction follows by an application of \irref{dRI} (whose right premise closes by \rref{eq:differential-rank} for any $\ivr$). The ``$\limply$" direction relies on existence and uniqueness of solutions to differential equations, which are internalized as axioms in \rref{sec:extaxioms}.
\end{proofsketch}
For the proof of \rref{thm:algcomplete}, we emphasize that additional axioms are \emph{only required} for syntactically deriving the ``$\limply$" direction (completeness) of \irref{DRI}. Hence, the base \dL axiomatization with differential ghosts is complete for proving properties of the form $\dbox{\pevolvein{\D{x}=\genDE{x}}{\ivr}}{p=0}$ because \irref{dRI} reduces all such questions to $\ivr \limply \sigliedzero{\genDE{x}}{p}$, which is a formula of real arithmetic, and hence, decidable.
The same applies for our next result, which is a corollary of \rref{thm:algcomplete}, but applies beyond the continuous fragment of \dL.

\begin{corollary}[Decidability]
\label{cor:testfree}
For algebraic formulas $\rfvar$ and hybrid programs $\alpha$ whose tests and domain constraints are negations of algebraic formulas \seeapp{\rref{app:alginvariants}}, it is possible to compute a polynomial $q$ such that the equivalence $\dbox{\alpha}{\rfvar} \lbisubjunct q=0$ is derivable in \dL.
\end{corollary}
\begin{proofsketch}[app:alginvariants]
By structural induction on $\alpha$ analogous to \cite[Thm.\ 1]{DBLP:conf/lics/Platzer12b}, using \rref{thm:algcomplete} for the differential equations case.
\end{proofsketch}

\section{Extended Axiomatization}
\label{sec:extaxioms}

In this section, we present the axiomatic extension that is used for the rest of this paper. The extension requires that the system \m{\D{x}=\genDE{x}} locally evolves $x$, i.e., it has no fixpoint at which $\genDE{x}$ is the 0 vector. This can be ensured syntactically, e.g., by requiring that the system contains a clock variable $\D{x_1}=1$ that tracks the passage of time, which can always first be added using \irref{dG} if necessary.

\subsection{Existence, Uniqueness, and Continuity}
\label{subsec:existenceuniqcont}

The differential equations considered in this paper have polynomial right-hand sides. Hence, the Picard-Lindel\"{o}f theorem~\cite[\S10.VI]{Walter1998} guarantees that for any initial state $\iget[state]{\I} \in \reals^n $, a \emph{unique} solution of the system $\pevolve{\D{x}=\genDE{x}}$, i.e., $\solvar : [0,T] \to \reals^n$ with $\solvar(0) = \iget[state]{\I}$, \emph{exists} for some duration $T > 0$. The solution $\solvar$ can be extended (uniquely) to its maximal open interval of existence~\cite[\S10.IX]{Walter1998} and $\solvar(\zeta)$ is differentiable, and hence continuous with respect to $\zeta$.

\begin{lemma}[Continuous existence, uniqueness, and differential adjoints]
\label{lem:uniqcont}
The following axioms are sound. In \irref{Cont} and \irref{Dadjoint}, $y$ are fresh variables (not in \m{\pevolvein{\D{x}=\genDE{x}}{\ivr(x)}} or $p$).

\begin{calculus}
\cinferenceRule[Uniq|Uniq]{uniqueness}
{\linferenceRule[impll]
  {\big(\ddiamond{\pevolvein{\D{x}=\genDE{x}}{\ivr_1}}{\rfvar_1}\big) \land
  \big(\ddiamond{\pevolvein{\D{x}=\genDE{x}}{\ivr_2}}{\rfvar_2}\big)}
  {\axkey{\ddiamond{\pevolvein{\D{x}=\genDE{x}}{\ivr_1 \land \ivr_2}}{(\rfvar_1 \lor \rfvar_2)}}}
}{}
\cinferenceRule[Cont|Cont]{continuous existence}
{
 \linferenceRule[impl]
  {x = y}
  {\big(p>0 \limply \axkey{\ddiamond{\pevolvein{\D{x}=\genDE{x}}{p > 0}}{x \not= y}}\big)}
}{}
\cinferenceRule[Dadjoint|Dadj]{differential adjoints}
{\linferenceRule[equiv]
  {\ddiamond{\pevolvein{\D{y}=-\genDE{y}}{\ivr(y)}}{\,y=x}}
  {\axkey{\ddiamond{\pevolvein{\D{x}=\genDE{x}}{\ivr(x)}}{\,x=y}}}
}{}
\end{calculus}%
\end{lemma}
\begin{proofsketch}[app:extaxiomatization]
\irref{Uniq} internalizes uniqueness, \irref{Cont} internalizes continuity of the values of $p$ and existence of solutions, and \irref{Dadjoint} internalizes the group action of time on ODE solutions, which is another consequence of existence and uniqueness.
\end{proofsketch}

The \emph{uniqueness axiom} \irref{Uniq} can be intuitively read as follows. If we have two solutions $\solvar_1,\solvar_2$ respectively staying in evolution domains $\ivr_1,\ivr_2$ and whose endpoints satisfy $\rfvar_1,\rfvar_2$, then one of $\solvar_1$ or $\solvar_2$ is a prefix of the other, and therefore, the prefix stays in both evolution domains so $\ivr_1 \land \ivr_2$ and satisfies $\rfvar_1 \lor \rfvar_2$ at its endpoint.

\emph{Continuity axiom} \irref{Cont} expresses a notion of \emph{local progress} for differential equations. It says that from an initial state satisfying $x=y$, the system can locally evolve to another state satisfying $x \neq y$ while staying in the \emph{open set} of states characterized by $p>0$. This uses the assumption that the system locally evolves $x$ at all.

The \emph{differential adjoints} axiom \irref{Dadjoint} expresses that $x$ can flow forward to $y$ iff $y$ can flow backward to $x$ along an ODE. It is at the heart of the ``there and back again" axiom that equivalently expresses properties of differential equations with evolution domain constraints in terms of properties of forwards and backwards differential equations without evolution domain constraints~\cite{DBLP:conf/lics/Platzer12b}.

To make use of these axioms, it will be useful to derive rules and axioms that allow us to work directly in the diamond modality, rather than the box modality.

\begin{corollary}[Derived diamond modality rules and axioms]
\label{cor:diadiffeqax}
The following derived axiom and derived rule are provable in \dL:

\begin{calculus}
\dinferenceRule[dDR|DR${\didia{\cdot}}$]{}
{\linferenceRule[impll]
  {\dbox{\pevolvein{\D{x}=\genDE{x}}{\rrfvar}}{\ivr}}
  {\big(\ddiamond{\pevolvein{\D{x}=\genDE{x}}{\rrfvar}}{\rfvar} \limply \axkey{\ddiamond{\pevolvein{\D{x}=\genDE{x}}{\ivr}}{\rfvar}}\big)}
}{}

\dinferenceRule[gddR|dRW${\didia{\cdot}}$]{}
{\linferenceRule
  {
  \lsequent{\rrfvar}{\ivr} &
  \lsequent{\Gamma}{\ddiamond{\pevolvein{\D{x}=\genDE{x}}{\rrfvar}}{\rfvar}}
  }
  {\lsequent{\Gamma}{\ddiamond{\pevolvein{\D{x}=\genDE{x}}{\ivr}}{\rfvar}}}
}{}
\end{calculus}
\end{corollary}
\iflongversion
\begin{proofsketch}[app:diaaxioms]
Axiom \irref{dDR} is the diamond version of the \dL refinement axiom that underlies \irref{dC}; if we never leave $\ivr$ when staying in $\rrfvar$ (first assumption), then any solution staying in $\rrfvar$ (second assumption) must also stay in $\ivr$ (conclusion).
The rule \irref{gddR} derives from \irref{dDR} using \irref{dW} on its first assumption.
\end{proofsketch}
\fi

\subsection{Real Induction}
\label{subsec:realind}

Our final axiom is based on the real induction principle~\cite{clark2012instructor}. It internalizes the topological properties of solutions. For space reasons, we only present the axiom for systems without evolution domain constraints, leaving the general version to
\iflongversion
\rref{app:extaxiomatization}.
\else
the report \rref{app:extaxiomatization}.
\fi

\begin{lemma}[Real induction]
\label{lem:realindODE}
The real induction axiom is sound, where $y$ is fresh in \m{\dbox{\pevolve{\D{x}=\genDE{x}}}{\rfvar}}.
\[
\cinferenceRule[RealInd|RI]{real induction axiom}
{\linferenceRule[equivl2]
{\lforall{y}{\dbox{\pevolvein{\D{x}=\genDE{x}}{\rfvar \lor x=y}}{
\big(x=y \limply \rfvar \land \ddiamond{\pevolvein{\D{x}=\genDE{x}}{\rfvar}}{x \neq y}\big)}}}
{\axkey{\dbox{\pevolve{\D{x}=\genDE{x}}}{\rfvar}}}
}{}
\]
\end{lemma}
\begin{proofsketch}[app:extaxiomatization]
The \irref{RealInd} axiom follows from the real induction principle~\cite{clark2012instructor} and the Picard-Lindel\"{o}f theorem~\cite[\S10.VI]{Walter1998}.
\end{proofsketch}

\begin{figure}
\centering
\includegraphics[width=\columnwidth]{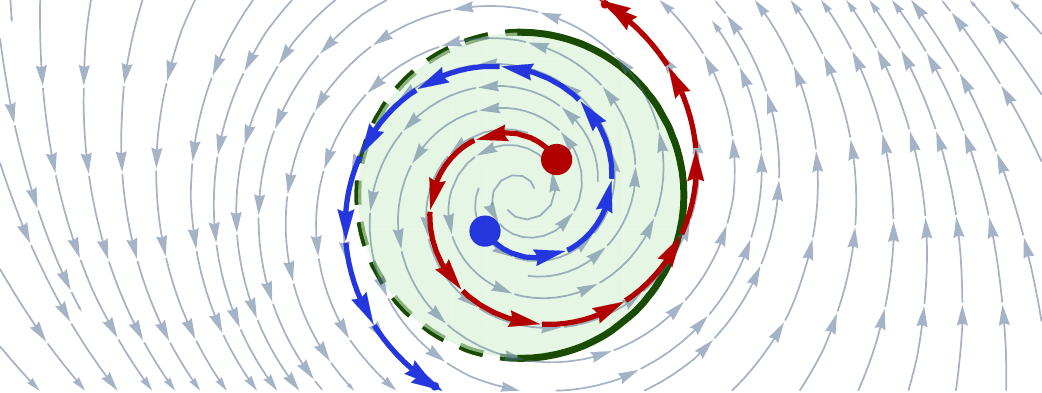}
\caption{The half-open disk $u^2+v^2 < \frac{1}{4} \lor u^2+v^2 = \frac{1}{4}\land u \geq 0$ is not invariant for $\alpha_e$ because the red and blue trajectories spiral out of it towards the unit circle at a closed or open boundary, respectively.}
\label{fig:realinduct}
\end{figure}

To see the topological significance of \irref{RealInd}, recall the running example and consider a set of points that is \emph{not invariant}. Figure~\ref{fig:realinduct} illustrates two trajectories that leave the candidate invariant disk $S$. These trajectories must stay in $S$ before leaving it through its boundary, and only in one of two ways: either at a point which is also in $S$ (red trajectory exiting right) or is not (the blue trajectory).

Real induction axiom \irref{RealInd} can be understood as $\lforall{y}{\dbox{\dots}{\big(x=y \limply \dots\big)}}$ quantifying over all final states ($x=y$) reachable by trajectories still within $\rfvar$ except possibly at the endpoint $x=y$.
The left conjunct under the modality expresses that $\rfvar$ is still true at such an endpoint, while the right conjunct expresses that the ODE still remains in $\rfvar$ locally. The left conjunct rules out trajectories like the blue one exiting left in \rref{fig:realinduct}, while the right conjunct rules out trajectories like the red trajectory exiting right.

The right conjunct suggests a way to use \irref{RealInd}: it reduces invariants to local progress properties under the box modality. This motivates the following syntactic modality abbreviations for \emph{progress within a domain $\ivr$} (with the initial point) or \emph{progress into $\ivr$} (without):
\[\begin{aligned}
\dprogressinsmall{\D{x}=\genDE{x}}{\ivr}{} &\mdefequiv \ddiamond{\pevolvein{\D{x}=\genDE{x}}{\ivr}}{\,x\neq y} \\
\dprogressin{\D{x}=\genDE{x}}{\ivr}{} &\mdefequiv \ddiamond{\pevolvein{\D{x}=\genDE{x}}{\ivr \lor x=y}}{\,x\neq y}
\end{aligned}\]
All remaining proofs in this paper only use these two modalities with an initial assumption $x=y$. In this case, where \(\ivaluation{\I}{x}=\ivaluation{\I}{y}\), the $\ddnext$ modality has the following semantics:
\begin{align*}
&\imodels{\I}{\dprogressin{\D{x}=\genDE{x}}{\ivr}{}}~\text{iff}~\text{there is a function}~\solvar:[0,T] \to \reals^n\\&\text{with}~T>0, \solvar(0)=\omega, \solvar~\text{is a solution of the system}~\D{x}=\genDE{x}, \text{and}\\&\solvar(\zeta) \in \imodel{\I}{\ivr}~\text{for all $\zeta$ in the half-open interval}~(0,T]
\end{align*}
For \(\dprogressinsmall{\D{x}=\genDE{x}}{\ivr}{}\) it is the closed interval~$[0,T]$ instead of~$(0,T]$.
Both $\ddnext$ and $\ddnextsmall$ resemble continuous-time versions of the next modality of temporal logic with the only difference being whether the initial state already needs to start in $\ivr$.
Both coincide if \(\imodels{\I}{\ivr}\).

The motivation for separating these modalities is topological: $\dprogressinsmall{\D{x}=\genDE{x}}{\ivr}$ is uninformative (trivially true) if the initial state \(\imodels{\I}{\ivr}\) and $\ivr$ describes an open set, because existence and continuity already imply local progress. Excluding the initial state as in $\dprogressin{\D{x}=\genDE{x}}{\ivr}{}$ makes this an insightful question, because it allows the possibility of starting on the topological boundary before entering the open set.

For brevity, we leave the $x=y$ assumption in the antecedents and axioms implicit in all subsequent derivations. For example, we shall elide the implicit $x=y$ assumption and write axiom \irref{Cont} as:
\[
\dinferenceRule[Contabbrev|Cont]{continuous existence}
{
 \linferenceRule[impl]
  {p>0}
  {\dprogressinsmall{\D{x}=\genDE{x}}{p > 0}}
}{}
\]

\begin{corollary}[Real induction rule]
\label{cor:realind}
This rule derives from \irref{RealInd+Dadjoint}.
\[\dinferenceRule[realind|rI]{}
{\linferenceRule
  {
   \lsequent{\rfvar}{\dprogressin{\D{x}=\genDE{x}}{\rfvar}} & \lsequent{\lnot{\rfvar}}{\dprogressin{\D{x}=-\genDE{x}}{\lnot{\rfvar}}}
  }
  {\lsequent{\rfvar}{\dbox{\pevolve{\D{x}=\genDE{x}}}{\rfvar}}}
}{}\]
\end{corollary}
\begin{proofsketch}[app:diaaxioms]
The rule derives from \irref{RealInd}, where we have used \irref{Dadjoint} to axiomatically flip the signs of its second premise.
\end{proofsketch}

Rule \irref{realind} shows what our added axioms buys us: \irref{RealInd} reduces global invariance properties of ODEs to local progress properties. These properties will be provable with \irref{Cont+Uniq} and existing \dL axioms. Both premises of \irref{realind} allow us to assume that the formula we want to prove local progress for is true initially. Thus, we could have equivalently stated the succedent with $\ddnextsmall$ modalities instead of $\ddnext$ in both premises. The choice of $\ddnext$ will be better for strict inequalities.

\section{Semialgebraic Invariants}
\label{sec:semialg}

From now on, we simply assume domain constraint $\ivr \equiv \ltrue$ since $\ivr$ is not fundamental \cite{DBLP:conf/lics/Platzer12b} and not central to our discussion.\footnote{We provide the case of arbitrary semialgebraic evolution domain $\ivr$ in~\rref{app:completeness}.}
Any first-order formula of real arithmetic, $\rfvar$, characterizes a \emph{semialgebraic set}, and by quantifier elimination~\cite{Bochnak1998} may equivalently be written as a finite, quantifier-free formula with polynomials $p_{ij},q_{ij}$:
\begin{equation}
\rfvar \mequiv \lorfold_{i=0}^{M} \Big(\landfold_{j=0}^{m(i)} p_{ij} \geq 0 \land \landfold_{j=0}^{n(i)} q_{ij} > 0\Big)
\label{eq:normalform}
\end{equation}
$\rfvar$ is also called a \emph{semialgebraic} formula, and the first step in our invariance proofs for semialgebraic $\rfvar$ will be to apply rule \irref{realind}, yielding premises of the form \(\lsequent{\rfvar}{\dprogressin{\D{x}=\genDE{x}}{\rfvar}}\) (modulo sign changes and negation). The key insight then is that local progress can be completely characterized by a finite formula of real arithmetic.

\subsection{Local Progress}
\label{subsec:localprogress}

Local progress was implicitly used previously for semialgebraic invariants~\cite{DBLP:conf/emsoft/LiuZZ11,DBLP:journals/cl/GhorbalSP17}. Here, we show how to derive the characterization syntactically in the \dL calculus, starting from atomic inequalities. We observe interesting properties, e.g., self-duality, along the way.

\subsubsection{Atomic Non-strict Inequalities}

Let $\rfvar$ be $p \geq 0$. Intuitively, since we only want to show \emph{local progress}, it is sufficient to locally consider the \emph{first} (significant) Lie derivative of $p$. This is made precise with the following key lemma.

\begin{lemma}[Local progress step]
The following axiom derives from \irref{Cont} in \dL.
\[\dinferenceRule[Lpgeq|LPi$_\geq$]{}
{
\linferenceRule[impll]
  {p \geq 0 \land \big(p=0 \limply \dprogressinsmall{\D{x}=\genDE{x}}{\lied[]{\genDE{x}}{p} \geq 0}\big)}
  {\axkey{\dprogressinsmall{\D{x}=\genDE{x}}{p \geq 0}}}
}{}\]
\end{lemma}
\begin{proof}
The proof starts with a case split since $p \geq 0$ is equivalent to \(p > 0 \lor p = 0\). In the $p>0$ case, \irref{Contabbrev} and \irref{gddR} close the premise. The premise from the $p=0$ case is abbreviated with \textcircled{1}.

{\footnotesize\renewcommand*{\arraystretch}{1.3}%
\begin{sequentdeduction}[array]
\linfer[qear+orl]{
    \linfer[gddR]{
      \linfer[Contabbrev]{\lclose}
      {\lsequent{p > 0}{\dprogressinsmall{\D{x}=\genDE{x}}{p > 0}}}
    }
    {\lsequent{p > 0}{\dprogressinsmall{\D{x}=\genDE{x}}{p \geq 0}}} !
    \textcircled{1}
}
  {\lsequent{p \geq 0, p=0 \limply \dprogressinsmall{\D{x}=\genDE{x}}{\lied[]{\genDE{x}}{p} \geq 0}}{\dprogressinsmall{\D{x}=\genDE{x}}{p \geq 0}}}
\end{sequentdeduction}
}%
We continue on \textcircled{1} with \irref{dDR} and finish the proof using \irref{dI}:
{\footnotesize\renewcommand*{\arraystretch}{1.3}%
\begin{sequentdeduction}[array]
\linfer[implyl]{
  \linfer[dDR]{
  \linfer[dI]{
    \lclose
  }
  {\lsequent{p=0}{\dbox{\pevolvein{\D{x}=\genDE{x}}{\lied[]{\genDE{x}}{p} \geq 0}}{p\geq 0}}}
  }
  {\lsequent{p=0,\dprogressinsmall{\D{x}=\genDE{x}}{\lied[]{\genDE{x}}{p} \geq 0}}{\dprogressinsmall{\D{x}=\genDE{x}}{p \geq 0}}}
}
  {\lsequent{p=0, p=0 \limply \dprogressinsmall{\D{x}=\genDE{x}}{\lied[]{\genDE{x}}{p} \geq 0}}{\dprogressinsmall{\D{x}=\genDE{x}}{p \geq 0}}}
\end{sequentdeduction}
}%
\end{proof}

Observe that \irref{Lpgeq} allows us to pass from reasoning about local progress for $p \geq 0$ to local progress for its Lie derivative $\lied[]{\genDE{x}}{p} \geq 0$ whilst accumulating $p=0$ in the antecedent. Furthermore, this can be iterated for higher Lie derivatives, as in the following derivation:
{\footnotesize\renewcommand{\linferPremissSeparation}{~~~}%
\begin{sequentdeduction}[array]
\linfer[Lpgeq]{
  \lsequent{\Gamma}{p \geq 0} !
  \linfer[Lpgeq]{
  \lsequent{\Gamma,p=0}{\lied[]{\genDE{x}}{p} \geq 0} !
  \linfer[Lpgeq]{
  \lsequent{\Gamma,p=0,\dots}{\dprogressinsmall{\D{x}=\genDE{x}}{\lied[k]{\genDE{x}}{p} \geq 0}}
}
  {\dots}
}
  {\lsequent{\Gamma,p=0}{\dprogressinsmall{\D{x}=\genDE{x}}{\lied[]{\genDE{x}}{p} \geq 0}}}
}
  {\lsequent{\Gamma}{\dprogressinsmall{\D{x}=\genDE{x}}{p \geq 0}}}
\end{sequentdeduction}
}%
Indeed, if we could prove $\lied[k]{\genDE{x}}{p} > 0$ from the antecedent, \irref{Contabbrev+gddR} finish the proof, because we must then locally enter $\lied[k]{\genDE{x}}{p} > 0$:
{\footnotesize\renewcommand{\linferPremissSeparation}{~~~}%
\begin{sequentdeduction}[array]
\linfer[cut]{
  \lsequent{\Gamma,p=0,\dots,\lied[k-1]{\genDE{x}}{p}= 0}{\lied[k]{\genDE{x}}{p} > 0} !
  \linfer[gddR]{
  \linfer[Contabbrev]{
    \lclose
  }
    {\lsequent{\lied[k]{\genDE{x}}{p} > 0}{\dprogressinsmall{\D{x}=\genDE{x}}{\lied[k]{\genDE{x}}{p} > 0}}}
  }
  {\lsequent{\lied[k]{\genDE{x}}{p} > 0}{\dprogressinsmall{\D{x}=\genDE{x}}{\lied[k]{\genDE{x}}{p} \geq 0}}}
}
  \lsequent{\Gamma,p=0,\dots,\lied[k-1]{\genDE{x}}{p}= 0}{\dprogressinsmall{\D{x}=\genDE{x}}{\lied[k]{\genDE{x}}{p} \geq 0}}
\end{sequentdeduction}
}%

This derivation repeatedly examines higher Lie derivatives when lower ones are indeterminate ($p=0,\dots,\lied[k-1]{\genDE{x}}{p}=0$), until we find the \emph{first} significant derivative with a definite sign ($\lied[k]{\genDE{x}}{p} > 0$). Fortunately, we already know that this terminates: when $N$ is the rank of $p$, then once we gathered $p=0,\dots,\lied[N-1]{\genDE{x}}{p}=0$, i.e., $\sigliedzero{\genDE{x}}{p}$ in the antecedents, \irref{dRI} proves the invariant $p=0$, and ODEs always locally progress in invariants.
The following definition gathers the open premises above to obtain the \emph{first significant Lie derivative}.

\begin{definition}[Progress formula]
The \emph{progress formula} \m{\sigliedgt{\genDE{x}}{p}} for a polynomial $p$ with rank $N{\geq}1$ is defined as the following formula, where Lie derivatives are with respect to $\D{x}=\genDE{x}$:
\begin{align*}
\sigliedgt{\genDE{x}}{p} \mdefequiv &p\geq 0 \land \big(p=0 \limply \lied[]{\genDE{x}}{p} \geq 0\big) \land \big(p=0 \land \lied[]{\genDE{x}}{p} = 0  \limply \lied[2]{\genDE{x}}{p} \geq 0\big) \\
\land& \dots\\
\land& \big(p=0 \land \lied[]{\genDE{x}}{p} = 0 \land \dots \land \lied[N-2]{\genDE{x}}{p} = 0 \limply \lied[N-1]{\genDE{x}}{p} > 0\big)
\end{align*}
We define $\sigliedgeq{\genDE{x}}{p}$ as \m{\sigliedgt{\genDE{x}}{p} \lor \sigliedzero{\genDE{x}}{p}}. We write $\sigliedgt[-]{\genDE{x}}{p}$ (or $\sigliedgeq[-]{\genDE{x}}{p}$) when taking Lie derivatives w.r.t.\ \(\D{x}=-\genDE{x}\).
\end{definition}

\begin{lemma}[Local progress $\geq$]
\label{lem:localprogressgeq}
This axiom derives from \irref{Lpgeq}:
\[\dinferenceRule[Lpgeqfull|LP$_{\geq^*}$]{Progress Conditions}
{
\linferenceRule[impl]
  {\sigliedgeq{\genDE{x}}{p}}
  {\axkey{\dprogressinsmall{\D{x}=\genDE{x}}{p \geq 0}}}
}{}
\]
\end{lemma}
\begin{proofsketch}[app:localprogress]
This follows by the preceding discussion with iterated use of derived axioms \irref{Lpgeq} and \irref{dRI}.
\end{proofsketch}

In order to prove $\dprogressinsmall{\D{x}=\genDE{x}}{p \geq 0}$, it is not always necessary to consider the entire progress formula for $p$. The iterated derivation shows that once the antecedent ($\Gamma,p=0,\dots,\lied[k-1]{\genDE{x}}{p}=0$) implies that the next Lie derivative is significant ($\lied[k]{\genDE{x}}{p} > 0$), the proof can stop early without considering the remaining higher Lie derivatives.

\subsubsection{Atomic Strict Inequalities}

Let $P$ be $p > 0$. Unlike the above non-strict cases, where $\ddnext$ and $\ddnextsmall$ were equivalent, we now exploit the $\ddnext$ modality. The reason for this difference is that the set of states satisfying $p > 0$ is topologically open and, as mentioned earlier, it is possible to \emph{locally enter} the set from an initial point on its boundary. This becomes important when we generalize to the case of semialgebraic $\rfvar$ in normal form \rref{eq:normalform} because it allows us to move between its outer disjunctions.

\begin{lemma}[Local progress $>$]
\label{lem:localprogressgt}
This axiom derives from \irref{Lpgeq}:
\[\dinferenceRule[Lpgtfull|LP$_{>^*}$]{Progress Conditions}
{
\linferenceRule[impl]
  {\sigliedgt{\genDE{x}}{p}}
  {\axkey{\dprogressin{\D{x}=\genDE{x}}{p > 0}}}
}{}
\]
\end{lemma}
\begin{proofsketch}[app:localprogress]
We start by unfolding the syntactic abbreviation of the $\ddnext$ modality, and observing that we can reduce to the non-strict case with \irref{gddR} and the real arithmetic fact\footnote{Here, $|x-y|^2$ is the squared Euclidean norm $(x_1-y_1)^2+\dots+(x_n-y_n)^2$} $p \geq |x-y|^{2N} \limply p > 0 \lor x=y$, where $N{\geq}1$ is the rank of $p$. The appearance of $N$ in this latter step corresponds to the fact that we only need to inspect the first $N-1$ Lie derivatives of $p$ with $\sigliedgt{\genDE{x}}{p}$. We further motivate this choice in the full proof \seeapp{\rref{app:localprogress}}.
{\footnotesize%
\begin{sequentdeduction}[array]
\linfer[]{
  \linfer[gddR]{
    \linfer[qear]{ \lclose }
    {\lsequent{p \geq |x-y|^{2N}}{p > 0 \lor x=y}}!
    \lsequent{\Gamma}{\dprogressinsmall{\D{x}=\genDE{x}}{p \geq |x-y|^{2N}}}
  }
  {\lsequent{\Gamma}{\dprogressinsmall{\D{x}=\genDE{x}}{p > 0 \lor x=y}}}
}
  {\lsequent{\Gamma}{\dprogressin{\D{x}=\genDE{x}}{p > 0}}}
\end{sequentdeduction}
}%
We continue on the remaining open premise with iterated use of \irref{Lpgeq}, similar to the derivation for \rref{lem:localprogressgeq}.
\end{proofsketch}

\subsubsection{Semialgebraic Case}

We finally lift the progress formulas for atomic inequalities to the general case of an arbitrary semialgebraic formula in normal form.
\begin{definition}[Semialgebraic progress formula]
The \emph{semialgebraic progress formula} $\sigliedsai{\genDE{x}}{\rfvar}$ for a semialgebraic formula $\rfvar$ written in normal form~\rref{eq:normalform} is defined as follows:
\begin{align*}
\sigliedsai{\genDE{x}}{\rfvar} ~\mdefequiv~ \lorfold_{i=0}^{M} \Big(\landfold_{j=0}^{m(i)} \sigliedgeq{\genDE{x}}{p_{ij}} \land \landfold_{j=0}^{n(i)} \sigliedgt{\genDE{x}}{q_{ij}}\Big)
\end{align*}
We write $\sigliedsai[-]{\genDE{x}}{\rfvar}$ when taking Lie derivatives w.r.t.\ \(\D{x}=-\genDE{x}\).
\end{definition}

\begin{lemma}[Semialgebraic local progress]
\label{lem:localprogresssemialg}
Let $\rfvar$ be a semialgebraic formula in normal form \rref{eq:normalform}. The following axiom derives from \dL extended with $\irref{Cont+Uniq}$.
\[
\dinferenceRule[LpRfull|LP\usebox{\Rval}]{Progress Condition}
{\linferenceRule[impl]
  {\sigliedsai{\genDE{x}}{\rfvar}}
  {\axkey{\dprogressin{\D{x}=\genDE{x}}{\rfvar}}}
}{}
\]
\end{lemma}
\begin{proofsketch}[app:localprogress]
We decompose $\sigliedsai{\genDE{x}}{\rfvar}$ according to its outermost disjunction, and accordingly decompose $\rfvar$ in the local progress succedent with \irref{gddR}. We then use \irref{Uniq+band} to split the conjunctive local progress condition in the resulting succedents of open premises, before finally utilizing \irref{Lpgeqfull} or \irref{Lpgtfull}, respectively.
\end{proofsketch}

\rref{lem:localprogresssemialg} implies that the implication in \irref{LpRfull} can be strengthened to an equivalence. It also justifies our syntactic abbreviation $\ddnext$, recalling that the $\ddnext$ modality of temporal logic is self-dual.
\begin{corollary}[Local progress completeness]
\label{cor:localprogresscomplete}
Let $\rfvar$ be a semialgebraic formula in normal form \rref{eq:normalform}. The following axioms derive from \dL extended with $\irref{Cont+Uniq}$.

\begin{calculus}
\dinferenceRule[Lpiff|LP]{Iff Progress Condition}
{\linferenceRule[equiv]
  {\sigliedsai{\genDE{x}}{\rfvar}}
  {\axkey{\dprogressin{\D{x}=\genDE{x}}{\rfvar}}}
}{}

\dinferenceRule[duality|$\lnot{\ddnext}$]{Duality}
{\linferenceRule[equiv]
  {\lnot{\dprogressin{\D{x}=\genDE{x}}{\lnot{\rfvar}}}}
  {\axkey{\dprogressin{\D{x}=\genDE{x}}{\rfvar}}}
}{}
\end{calculus}
\end{corollary}
\begin{proofsketch}[app:localprogress]
Both follow because any $\rfvar$ in normal form \rref{eq:normalform} has a corresponding normal form for $\lnot{\rfvar}$ such that the equivalence $\lnot{(\sigliedsai{\genDE{x}}{\rfvar})} \lbisubjunct \sigliedsai{\genDE{x}}{(\lnot{\rfvar})}$ is provable. Then apply \irref{Uniq+LpRfull}.
\end{proofsketch}

In continuous time, there is no discrete next state, so unlike the $\ddnext$ modality of discrete temporal logic, local progress is \emph{idempotent}.

\subsection{Completeness for Semialgebraic Invariants}
\label{subsec:completenesssemialg}

We summarize our results with the following derived rule.
\begin{theorem}[Semialgebraic invariants] \label{thm:sAI}
For semialgebraic $\rfvar$ with progress formulas $\sigliedsai{\genDE{x}}{\rfvar}, \sigliedsai[-]{\genDE{x}}{(\lnot{\rfvar})}$ w.r.t.\ their respective normal forms \rref{eq:normalform}, this rule derives from the \dL calculus with $\irref{RealInd+Dadjoint+Cont+Uniq}$.
\[
\dinferenceRule[sAI|sAI]{semialgebraic invariants}
{\linferenceRule
  {
  \lsequent{\rfvar} {\sigliedsai{\genDE{x}}{\rfvar}} & \lsequent{\lnot{\rfvar}} {\sigliedsai[-]{\genDE{x}}{(\lnot{\rfvar})}}
  }
  {\lsequent{\rfvar}{\dbox{\pevolve{\D{x}=\genDE{x}}}{\rfvar}}}
}{}
\]
\end{theorem}
\begin{proof}
Straightforward application of \irref{realind+Lpiff}.
\end{proof}

Completeness of \irref{sAI} was proved semantically in~\cite{DBLP:conf/emsoft/LiuZZ11} making crucial use of semialgebraic sets and analytic solutions to polynomial ODE systems. We showed that the \irref{sAI} proof rule can be \emph{derived} syntactically in the \dL calculus and derive its completeness, too:

\begin{theorem}[Semialgebraic invariant completeness]
\label{thm:semialgcompleteness}
For semialgebraic $\rfvar$ with progress formulas $\sigliedsai{\genDE{x}}{\rfvar}, \sigliedsai[-]{\genDE{x}}{(\lnot{\rfvar})}$ w.r.t.\ their respective normal forms \rref{eq:normalform}, this axiom derives from \dL with $\irref{RealInd+Dadjoint+Cont+Uniq}$.
\[\dinferenceRule[semialgiff|SAI]{Semialgebraic invariant axiom}
{\linferenceRule[equiv]
  {\lforall{x}{\big(\rfvar \limply \sigliedsai{\genDE{x}}{\rfvar}\big)} \land\lforall{x}{\big(\lnot{\rfvar} \limply \sigliedsai[-]{\-genDE{x}}{(\lnot{\rfvar})}\big)}}
  {\axkey{\lforall{x}{(\rfvar \limply \dbox{\pevolve{\D{x}=\genDE{x}}}{\rfvar})}}}
}{}\]
\end{theorem}
In~\rref{app:completeness}, we prove a generalization of \rref{thm:semialgcompleteness} that handles semialgebraic evolution domains $\ivr$ using \irref{Lpiff} and a corresponding generalization of axiom \irref{RealInd}. Thus, \dL decides invariance properties for all first-order real arithmetic formulas $\rfvar$, because quantifier elimination~\cite{Bochnak1998} can equivalently rewrite $\rfvar$ to normal form \rref{eq:normalform} first. Unlike for \rref{thm:algcomplete}, which can decide algebraic postconditions from any semialgebraic precondition, \rref{thm:semialgcompleteness} (and its generalized version) are still limited to proving invariants, the search of which is the only remaining challenge.

Of course, \irref{sAI} can be used to prove all the invariants considered in our running example. However, we had a significantly simpler proof for the invariance of $1-u^2-v^2>0$ with \irref{dbxineq}. This has implications for implementations of \irref{sAI}: simpler proofs help minimize dependence on real arithmetic decision procedures. Similarly, we note that if $\rfvar$ is either topologically open (resp.\ closed), then the left (resp.\ right) premise of \irref{sAI} closes trivially. Logically, this follows by the finiteness theorem~\cite[Theorem 2.7.2]{Bochnak1998}, which implies that formula $\rfvar \limply \sigliedsai{\genDE{x}}{\rfvar}$ is provable in real arithmetic for open semialgebraic $\rfvar$. Topologically, this corresponds to the fact that only one of the two exit trajectory cases in \rref{subsec:realind} can occur.

\section{Related Work}
\label{sec:relatedwork}

We focus our discussion on work related to deductive verification of hybrid systems. Readers interested in ODEs~\cite{Walter1998}, real analysis~\cite{clark2012instructor}, and real algebraic geometry~\cite{Bochnak1998} are referred to the respective cited texts. Orthogonal to our work is the question of how invariants can be efficiently generated, e.g.~\cite{DBLP:conf/tacas/GhorbalP14,DBLP:conf/emsoft/LiuZZ11,DBLP:journals/fmsd/SankaranarayananSM08}.

\paragraph{Proof Rules for Invariants.}
There are numerous useful but incomplete proof rules for ODE invariants \cite{DBLP:conf/hybrid/PrajnaJ04,DBLP:journals/fmsd/SankaranarayananSM08,DBLP:conf/fsttcs/TalyT09}.
An overview can be found in~\cite{DBLP:journals/cl/GhorbalSP17}.
The soundness and completeness theorems for \irref{dRI+sAI} were first shown in~\cite{DBLP:conf/tacas/GhorbalP14} and~\cite{DBLP:conf/emsoft/LiuZZ11} respectively.

In their original presentation, \irref{dRI} and \irref{sAI}, are \emph{algorithmic procedures} for checking invariance, requiring \eg, checking ideal membership for all polynomials in the semialgebraic decomposition.
This makes them very difficult to implement soundly as part of a small, trusted axiomatic core, such as the implementation of \dL in \KeYmaeraX~\cite{DBLP:conf/cade/FultonMQVP15}.
We instead show that these rules can be \emph{derived} from a small set of axiomatic principles.
Although we also leverage ideal computations, they are only used in \emph{derived rules}.
With the aid of a theorem prover, derived rules can be implemented as tactics that crucially remain \emph{outside} the soundness-critical axiomatic core.
Our completeness results are axiomatic, so complete for disproofs.

\paragraph{Deductive Power and Proof Theory.} The derivations shown in this paper are fully general, which is necessary for completeness of the resulting derived rules. The number of conjuncts in the progress and differential radical formulas, for example, are equal to the rank of $p$. Known upper bounds for the rank of $p$ in $n$ variables are doubly exponential in $n^2 \ln{n}$~\cite{novikov1999trajectories}. Fortunately, many simpler classes of invariants can be proved using simpler derivations.
This is where a study of the deductive power of various sound, but incomplete, proof rules~\cite{DBLP:journals/cl/GhorbalSP17} comes into play. If we know that an invariant of interest is of a simpler class, then we could simply use the proof rule that is complete for that class. This intuition is echoed in~\cite{DBLP:journals/lmcs/Platzer12}, which studies the relative deductive power of differential invariants (\irref{dI}) and differential cuts (\irref{dC}). Our first result shows, in fact, that \dL with \irref{dG} is already complete for algebraic invariants.
Other proof-theoretical studies of \dL~\cite{DBLP:conf/lics/Platzer12b} reveal surprising correspondences between its hybrid, continuous and discrete aspects in the sense that each aspect can be axiomatized completely relative to any other aspect. Our \rref{cor:testfree} is a step in this direction.

\section{Conclusion and Future Work}
\label{sec:conclusion}

The first part of this paper demonstrates the impressive deductive power of differential ghosts: they prove all algebraic invariants and Darboux inequalities. We leave open the question of whether their deductive power extends to larger classes of invariants.
The second part of this paper introduces extensions to the base \dL axiomatization, and shows how they can be used together with the existing axioms to decide real arithmetic invariants syntactically.

It is instructive to examine the mathematical properties of solutions and terms that underlie our axiomatization. In summary:
\begin{center}
\begin{tabular}{ll}
\hline
\textbf{Axiom}           & \textbf{Property}  \\ \hline
\irref{dI}       & Mean value theorem\\
\irref{dC}       & Prefix-closure of solutions\\
\irref{dG}       & Picard-Lindel\"of\\
\irref{Cont}     & Existence of solutions\\
\irref{Uniq}     & Uniqueness of solutions\\
\irref{Dadjoint} & Group action on solutions\\
\irref{RealInd}  & Completeness of $\reals$\\
\hline
\end{tabular}
\end{center}
The soundness of our axiomatization, therefore, easily extends to term languages beyond polynomials, \eg, continuously differentiable terms satisfy the above properties. We may, of course, lose completeness and decidable arithmetic in the extended language, but we leave further exploration of these issues to future work.

\section*{Acknowledgments}
We thank Brandon Bohrer, Khalil Ghorbal, Andrew Sogokon, and the anonymous reviewers for their detailed feedback on this paper.
This material is based upon work supported by the National Science Foundation under NSF CAREER Award CNS-1054246.
The second author was also supported by A*STAR, Singapore.

Any opinions, findings, and conclusions or recommendations expressed in this publication are those of the author(s) and do not necessarily reflect the views of the National Science Foundation.
\bibliographystyle{plainurl}
\bibliography{diffaxiomatic-arXiv}

\iflongversion
\else
\fi
\clearpage
\appendix
\newcommand{\usb}{b}
\newcommand{\usf}{f}
\newcommand{\usg}{g}
\newcommand{\usp}{p}
\newcommand{\usP}{\rfvar}
\newcommand{\usR}{\rrfvar}
\newcommand{\usQ}{\ivr}
\newcommand{\usC}{C}
\newcommand{\usa}{a}
\newcommand{\usDE}{\usf\argx}
\newcommand{\usDEy}{\usf\argy}
\newcommand{\argx}{(x)}
\newcommand{\argy}{(y)}
\newcommand{\argxx}{(x,x')}
\newcommand{\interp}{I}
\newcommand{\itsem}[1]{\interp\omega\sem{#1}}
\newcommand{\itsemst}[2]{\interp{#1}\sem{#2}}
\newcommand{\ifsem}[1]{\interp\sem{#1}}
\newcommand{\ipsem}[1]{\interp\sem{#1}}
\newcommand{\solmodels}[2]{\interp,#1 \models #2}

\section{Differential Dynamic Logic Axiomatization}
\label{app:axiomatization}
We work with \dL's uniform substitution calculus presented in~\cite{DBLP:journals/jar/Platzer17}. The calculus is based on the uniform substitution inference rule:
\[
\cinferenceRule[US|US]{Uniform Substitution}
{\linferenceRule[formula]
  {\fvar}
  {\sigma(\fvar)}
}{}
\]
The uniform substitution calculus requires a few extensions to the syntax and semantics presented in \rref{sec:background}. Firstly we extend the term language with differential terms $\D{(e)}$ and $k$-ary function symbols $\usf$, where $e_1,\dots,e_k$ are terms. The formulas are similarly extended with $k$-ary predicate symbols $\usp$:
\begin{align*}
e &\bebecomes \cdots \alternative \D{(e)} \alternative \usf(e_1,\dots,e_k)\\
\fvar &\bebecomes \cdots \alternative \usp(e_1,\dots,e_k)
\end{align*}
The grammar of \dL programs is as follows ($a$ is a program symbol):
\[ \alpha,\beta \bebecomes \usa \alternative \pumod{x}{e}  \alternative \ptest{\fvar} \alternative \pevolvein{\D{x}=\genDE{x}}{\ivr}  \alternative \pchoice{\alpha}{\beta} \alternative \alpha;\beta \alternative \prepeat{\alpha}\]

We refer readers to~\cite{DBLP:journals/jar/Platzer17} for the complete, extended semantics. Briefly, for each variable $x$, there is an associated differential variable $\D{x}$, and states map all of these variables (including differential variables) to real values; we write $\States$ for the set of all states. The semantics also requires an interpretation $\interp$ for the uniform substitution symbols. The term semantics, $\itsem{e}$, gives the value of $e$ in state $\omega$ and interpretation $\interp$.
Differentials have a differential-form semantics \cite{DBLP:journals/jar/Platzer17} as the sum of all partial derivatives by all variables $x$ multiplied by the corresponding values of $\D{x}$:
\[
\itsem{\der{\theta}}
=
\sum_{x} \iget[state]{\I}(\D{x}) \Dp[x]{\itsem{\theta}}
\]
The formula semantics, $\ifsem{\fvar}$, is the set of states where $\fvar$ is true in interpretation $\interp$, and the transition semantics of hybrid programs $\ipsem{\alpha}$ is given with respect to interpretation $I$. The transition semantics for $\D{x}=\genDE{x}$ requires:
\begin{align*}
&(\omega,\nu) \in \ipsem{\pevolvein{\D{x}=\genDE{x}}{\ivr}} ~\text{iff}~\text{there is}~T\geq 0~\text{and a function}~\\
&\solvar:[0,T] \to \States ~\text{with}~ \solvar(0)=\omega ~\text{on}~ \scomplement{\{\D{x}\}}, \solvar(T)=\nu,~\text{and}\\
&\solmodels{\solvar}{\pevolvein{\D{x}=\genDE{x}}{\ivr}}
\end{align*}
The \m{\solmodels{\solvar}{\pevolvein {\D{x}=\genDE{x}}{\ivr}}} condition checks \m{\solvar(\zeta) \in \ifsem{x'=f(x) \land \ivr}}, $\solvar(0) = \solvar(\zeta)$ on $\scomplement{\{x,\D{x}\}}$ for $0 \leq \zeta \leq T$, and, if $T>0$, then $\frac{d\solvar(t)(x)}{dt}(\zeta)$ exists, and is equal to $\solvar(\zeta)(\D{x})$ for all $0 \leq \zeta \leq T$. In other words, $\solvar$ is a solution of the differential equations $\D{x}=\genDE{x}$ that stays in the evolution domain constraint. It is also required to hold all variables other than $x,\D{x}$ constant. Most importantly, the values of the differential variables $\D{x}$ is required to match the value of the RHS of the differential equations along the solution. We refer readers to~\cite[Definition 7]{DBLP:journals/jar/Platzer17} for further details.

The \dL calculus allows all its axioms (\cf~\cite[Figures 2 and 3]{DBLP:journals/jar/Platzer17}) to be stated as \emph{concrete} instances, which are then instantiated by uniform substitution. In this appendix, we take the same approach: all of our (new) axioms will be stated as concrete instances as well. We will need to be slightly more careful, and write down explicit variable dependencies for all the axioms.
We shall directly use vectorial notation when presenting the axioms.
To make this paper self-contained, we state all of the axioms used in the paper and the appendix. However, we only provide justification for derived rules and axioms that are not already justified in~\cite{DBLP:journals/jar/Platzer17}.

\subsection{Base Axiomatization}
\label{app:baseaxiomatization}

The following are the base axioms and axiomatic proof rules for \dL from~\cite[Figure 2]{DBLP:journals/jar/Platzer17} where $\usall$ is the vector of all variables.

\begin{theorem}[Base axiomatization~\cite{DBLP:journals/jar/Platzer17}]
\label{thm:baseeqaxioms}
The following are sound axioms and proof rules for \dL.

    \begin{calculus}
      \cinferenceRule[diamond|$\didia{\cdot}$]{diamond axiom}
      {\linferenceRule[equiv]
        {\lnot\dbox{a}{\lnot p(\usall)}}
        {\axkey{\ddiamond{a}{p(\usall)}}}
      }
      {}
      \cinferenceRule[assignb|$\dibox{:=}$]{assignment / substitution axiom}
      {\linferenceRule[equiv]
        {p(f)}
        {\axkey{\dbox{\pupdate{\umod{x}{f}}}{p(x)}}}
      }
      {}%
      \irlabel{Dassignb|$\dibox{:=}$}%
      \cinferenceRule[testb|$\dibox{?}$]{test}
      {\linferenceRule[equiv]
        {(q \limply p)}
        {\axkey{\dbox{\ptest{q}}{p}}}
      }{}
      \cinferenceRule[choiceb|$\dibox{\cup}$]{axiom of nondeterministic choice}
      {\linferenceRule[equiv]
        {\dbox{a}{p(\usall)} \land \dbox{b}{p(\usall)}}
        {\axkey{\dbox{\pchoice{a}{b}}{p(\usall)}}}
      }{}
      \cinferenceRule[composeb|$\dibox{{;}}$]{composition} %
      {\linferenceRule[equiv]
        {\dbox{a}{\dbox{b}{p(\usall)}}}
        {\axkey{\dbox{a;b}{p(\usall)}}}
      }{}
      \cinferenceRule[iterateb|$\dibox{{}^*}$]{iteration/repeat unwind} %
      {\linferenceRule[equiv]
        {p(\usall) \land \dbox{a}{\dbox{\prepeat{a}}{p(\usall)}}}
        {\axkey{\dbox{\prepeat{a}}{p(\usall)}}}
      }{}
      \cinferenceRule[K|K]{K axiom / modal modus ponens} %
      {\linferenceRule[impl]
        {\dbox{a}{(p(\usall)\limply q(\usall))}}
        {(\dbox{a}{p(\usall)}\limply\axkey{\dbox{a}{q(\usall)}})}
      }{}
      \cinferenceRule[I|I]{loop induction}
      {\linferenceRule[equiv]
        {p(\usall) \land \dbox{\prepeat{a}}{(p(\usall)\limply\dbox{a}{p(\usall)})}}
        {\axkey{\dbox{\prepeat{a}}{p(\usall)}}}
      }{}
      \cinferenceRule[V|V]{vacuous $\dbox{}{}$}
      {\linferenceRule[impl]
        {p}
        {\axkey{\dbox{a}{p}}}
      }{}%
      \cinferenceRule[G|G]{$\dbox{}{}$ generalisation} %
      {\linferenceRule[formula]
        {p(\usall)}
        {\dbox{a}{p(\usall)}}
      }{}
    \end{calculus}
\end{theorem}
In sequent calculus, these axioms (and axiomatic proof rules) are instantiated by uniform substitution and then used by congruence reasoning for equivalences (and equalities). All of the substitutions that we require are admissible~\cite[Definition 19]{DBLP:journals/jar/Platzer17}. We use weakening to elide assumptions from the antecedent without notice and, e.g., use G\"odel's rule \irref{G} directly as:
\[\linfer[G]
  {\lsequent{}{\fvar}}
  {\lsequent{\Gamma}{\dbox{\alpha}{\fvar}}}
\]

The \irref{band} axiom derives from \irref{G+K} \cite{DBLP:conf/lics/Platzer12b}.
The \irref{Mb} rule derives using \irref{G+K} as well \cite{DBLP:journals/jar/Platzer17}.
The loop induction rule derives from the induction axiom \irref{I} using \irref{G} on its right conjunct \cite{DBLP:conf/lics/Platzer12b}.
\[\dinferenceRule[loop|loop]{}
{\linferenceRule
  {\lsequent{\fvar}{\dbox{\alpha}{\fvar}}}
  {\lsequent{\fvar}{\dbox{\prepeat{\alpha}}{\fvar}}}
}{}\]

The presentation of the base axiomatization in \rref{thm:baseeqaxioms} follows~\cite{DBLP:journals/jar/Platzer17}, where $\usp(\usall)$ is used to indicate a predicate symbol which takes, $\usall$, the vector of all variables. In the sequel, in order to avoid notational confusion with earlier parts of this paper, we return to using $\usP$ for predicate symbols, reserving $\usp$ for polynomial terms. Correspondingly, we return to using $x$ for the vector of variables appearing in ODE $\D{x}=\genDE{x}$ when stating the \dL axioms for differential equations.

\subsection{Differential Equation Axiomatization}
\label{app:diffaxiomatization}

The following are axioms for differential equations and differentials from~\cite[Figure 3]{DBLP:journals/jar/Platzer17}. Note that when $x$ is a vector of variables $x_1,x_2,\dots,x_n$, then $x'$ is the corresponding vector of differential variables $x_1',x_2',\dots,x_n'$, and $\usDE$ is a vector of $n$-ary function symbols $\usf_1\argx,\usf_2\argx,\dots,\usf_n\argx$.

\begin{theorem}[Differential equation axiomatization~\cite{DBLP:journals/jar/Platzer17}]
\label{thm:diffeqaxiomatization}
The following are sound axioms of \dL.

\begin{calculus}
      \cinferenceRule[DW|DW]{differential evolution domain} %
      {\axkey{\dbox{\pevolvein{\D{x}=\usDE}{\usQ\argx}}{\usQ\argx}}}
      {}

\cinferenceRule[DIeq|DI$_=$]{}
{\linferenceRule[impll]
  {(\usQ\argx \limply \dbox{\pevolvein{\D{x}=\usDE}{\usQ\argx}}{(\usp\argx)'=0})}
  {(\axkey{\dbox{\pevolvein{\D{x}=\usDE}{\usQ\argx}}{\usp\argx=0}} \lbisubjunct \dbox{\ptest{\usQ\argx}}{\usp\argx=0})}
}{}

\cinferenceRule[DIgeq|DI$_\cmp$]{}
{\linferenceRule[impll]
  {(\usQ\argx \limply \dbox{\pevolvein{\D{x}=\usDE}{\usQ\argx}}{(\usp\argx)'\geq0})}
  {(\axkey{\dbox{\pevolvein{\D{x}=\usDE}{\usQ\argx}}{\usp\argx\cmp0}} \lbisubjunct \dbox{\ptest{\usQ\argx}}{\usp\argx\cmp0})}
}{}

      \cinferenceRule[DE|DE]{differential effect} %
      {\linferenceRule[equivl]
        {\dbox{\pevolvein{\D{x}=\usDE}{\usQ\argx}}{\dbox{\Dupdate{\Dumod{\D{x}}{\usDE}}}{\usP(x,\D{x})}}}
        {\axkey{\dbox{\pevolvein{\D{x}=\usDE}{\usQ\argx}}{\usP(x,\D{x})}}}
      }
      {}%

      \cinferenceRule[Dconst|$c'$]{derive constant}
      {\linferenceRule[eq]
        {0}
        {\axkey{\der{f}}}
      }
      {}%
      \cinferenceRule[Dvar|$x'$]{derive variable}
      {\linferenceRule[eq]
        {\D{x}}
        {\axkey{\der{x}}}
      }
      {}%
      \cinferenceRule[Dplus|$+'$]{derive sum}
      {\linferenceRule[eq]
        {\der{\usf(\usall)}+\der{\usg(\usall)}}
        {\axkey{\der{\usf(\usall)+\usg(\usall)}}}
      }
      {}
      \cinferenceRule[Dtimes|$\cdot'$]{derive product}
      {\linferenceRule[eq]
        {\der{\usf(\usall)}\cdot \usg(\usall)+\usf(\usall)\cdot\der{\usg(\usall)}}
        {\axkey{\der{\usf(\usall)\cdot \usg(\usall)}}}
      }
      {}
\end{calculus}
\end{theorem}

We additionally use the Barcan axiom~\cite{DBLP:conf/lics/Platzer12b} specialized to ODEs in the diamond modality:
\[\cinferenceRule[dBarcan|B$'$]{}
{\linferenceRule[equivl]
  {\lexists{y}{\ddiamond{\pevolvein{\D{x}=\usDE}{\usQ\argx}}{\usP(x,y)}}}
  {\axkey{\ddiamond{\pevolvein{\D{x}=\usDE}{\usQ\argx}}{\exists{y}\usP(x,y)}}}
}{y \not\in x}\]
The ODE axiom \irref{dBarcan} requires the variables $y$ be fresh in $x$ ($y \not\in x$).

\irlabel{Dall|$e'$}
Syntactic differentiation under differential equations is performed using the \irref{DE} axiom along with the axioms for working with differentials~\irref{Dconst+Dvar+Dplus+Dtimes}~\cite[Lemmas 36-37]{DBLP:journals/jar/Platzer17}, and the assignment axiom~\irref{Dassignb} for differential variables. We label the exhaustive use of the differential axioms as \irref{Dall}. The following derivation is sound for any polynomial term $p$ (where $\lied[]{\genDE{x}}{p}$ is the polynomial term for the Lie derivative of $p$). We write $\sim$ for a free choice between $=$ and $\cmp$:
{\footnotesize\begin{sequentdeduction}[array]
\linfer[DE]{
\linfer[Dall+Dassignb+qear]{
  \lsequent{} {\dbox{\pevolvein{\D{x}=\genDE{x}}{\ivr}}{\lied[]{\genDE{x}}{p} \sim 0}}
}
  {\lsequent{} {\dbox{\pevolvein{\D{x}=\genDE{x}}{\ivr}}{\Dusubst{\D{x}}{\genDE{x}}{\der[]{p} \sim 0}}}}
}
  {\lsequent{} {\dbox{\pevolvein{\D{x}=\genDE{x}}{\ivr}}{\der[]{p} \sim 0}}}
\end{sequentdeduction}
}%
The \irref{Dall+Dassignb+qear} step first performs syntactic Lie derivation on $p$, and then additionally uses \irref{qear} to rearrange the resulting term into $\lied[]{\genDE{x}}{p}$ as required.
To see this more concretely, we perform the above derivation with a polynomial from the running example.
\begin{example}[Using syntactic derivations]
\label{ex:derivation}
Let $p_e \mdefeq v^2-u^2+\frac{9}{2}$, unfolding the Lie derivative, we have:
\begin{align*}
\lie[]{\alpha_e}{p_e} &= \frac{\partial p_e}{\partial u}(-v + \frac{u}{4} (1 - u^2 - v^2)) + \frac{\partial p_e}{\partial v}(u + \frac{v}{4} (1 - u^2 - v^2))\\
&= -2u(-v + \frac{u}{4} (1 - u^2 - v^2)) + 2v (u + \frac{v}{4} (1 - u^2 - v^2))\\
&= 4uv + \frac{1}{2}(1-u^2-v^2)(v^2-u^2) = \lied[]{\alpha_e}{p_e}
\end{align*}
In \dL, we have the following derivation:
{\footnotesize\begin{sequentdeduction}[array]
\linfer[DE]{
\linfer[Dall]{
\linfer[Dassignb]{
\linfer[qear]{
  \lsequent{} {\dbox{\alpha_e}{\lied[]{\alpha_e}{p_e} \sim 0}}
}
  {\lsequent{} {\dbox{\alpha_e}{2v(u + \frac{v}{4} (1 - u^2 - v^2))- 2u(-v + \frac{u}{4} (1 - u^2 - v^2))\sim 0}}}
}
  {\lsequent{} {\dbox{\alpha_e}{\Dusubst{\D{u}}{-v + \frac{u}{4} (1 - u^2 - v^2)}{\Dusubst{\D{v}}{u + \frac{v}{4} (1 - u^2 - v^2)}{2vv'-2uu' \sim 0}}}}}
}
  {\lsequent{} {\dbox{\alpha_e}{\Dusubst{\D{u}}{-v + \frac{u}{4} (1 - u^2 - v^2)}{\Dusubst{\D{v}}{u + \frac{v}{4} (1 - u^2 - v^2)}{\der[]{p_e} \sim 0}}}}}
}
  {\lsequent{} {\dbox{\alpha_e}{\der[]{p_e} \sim 0}}}
\end{sequentdeduction}
}%
Note that we needed the \irref{qear} step to rearrange the result from syntactically differentiating $p_e$ to match the expression $\lied[]{\alpha_e}{p_e}$ for the Lie derivative. Since the two notions must coincide under the ODEs, this rearrangement step is always possible.
\end{example}

We also use \dL's coincidence lemmas~\cite[Lemma 10,11]{DBLP:journals/jar/Platzer17}.
\begin{lemma}[Coincidence for terms and formulas~\cite{DBLP:journals/jar/Platzer17}]
\label{lem:coincide}
The following are coincidence properties of \dL, where free variables $\freevars{e},\freevars{\phi}$ are as defined in~\cite{DBLP:journals/jar/Platzer17}.
\begin{itemize}
\item If the states $\iget[state]{\I}, \iget[state]{\It}$ agree on the free variables of term $e$ (\m{\freevars{e}}), then $\itsemst{\iget[state]{\I}}{\polyn{e}{x}} = \itsemst{\iget[state]{\It}}{\polyn{e}{x}}$.
\item If the states $\iget[state]{\I}, \iget[state]{\It}$ agree on the free variables of formula $\phi$ (\m{\freevars{\phi}}), then $\iget[state]{\I} \in \ifsem{\phi}$ iff $\iget[state]{\It} \in \ifsem{\phi}$.
\end{itemize}
\end{lemma}

We prove generalized versions of axioms from~\cite{DBLP:journals/jar/Platzer17}. These are the vectorial differential ghost axioms (\irref{DG} and \irref{DGall})\footnote{We do not actually need \irref{DGall} in this paper. We prove it for completeness, because~\cite{DBLP:journals/jar/Platzer17} proves a similar axiom for single variable \irref{DG}.} which were proved only for the single variable case, and the differential modus ponens axiom, \irref{DMP}, which was specialized for differential cuts.
\begin{lemma}[Generalized axiom soundness]
\label{lem:dgdmp}
The following axioms are sound. Note that $\vec{y}$ is an $m$-dimensional vector of variables, $\D{\vec{y}}$ is its corresponding vector of differential variables, and $\usa\argx$ (resp.~$\usb\argx$) is an $m \times m$ matrix (resp.~$m$-dimensional vector) of function symbols.

\begin{calculus}
\cinferenceRule[DG|DG]{vectorial differential ghost variables}
{\linferenceRule[equivl]
  {\lexists{\vec{y}}{\dbox{\pevolvein{\D{x}=\usDE,\D{\vec{y}}=\usa\argx\stimes \vec{y}+\usb\argx}{\usQ\argx}}{\usP\argx}}}
  {\axkey{\dbox{\pevolvein{\D{x}=\usDE}{\usQ\argx}}{\usP\argx}}}
}{}
\cinferenceRule[DGall|DG$_\forall$]{vectorial differential ghost variables, universally}
{\linferenceRule[equivl]
  {\lforall{\vec{y}}{\dbox{\pevolvein{\D{x}=\usDE,\D{\vec{y}}=\usa\argx\stimes \vec{y}+\usb\argx}{\usQ\argx}}{\usP\argx}}}
  {\axkey{\dbox{\pevolvein{\D{x}=\usDE}{\usQ\argx}}{\usP\argx}}}
}{}
\cinferenceRule[DMP|DMP]{differential modus ponens}
{\linferenceRule[impll]
  {\dbox{\pevolvein{\D{x}=\usDE}{\usQ\argx}}{(\usQ\argx \limply \usR\argx)}}
  {(\dbox{\pevolvein{\D{x}=\usDE}{\usR\argx}}{\usP\argx} \limply \axkey{\dbox{\pevolvein{\D{x}=\usDE}{\usQ\argx}}{\usP\argx}})}
}{}
\end{calculus}
\end{lemma}
\begin{proof}
We use $\omega$ for the initial state, and $\nu$ for the state reached at the end of a continuous evolution. The valuations for matrix and vectorial terms are applied component-wise.

We first prove vectorial \irref{DG} and \irref{DGall}. Our proof is specialized to ODEs that are (inhomogeneous) linear in $\vec{y}$.
We only need to prove the ``$\limply$'' direction for \irref{DGall}, because $\lforall{\vec{y}}{\fvar}$ implies $\lexists{\vec{y}}{\fvar}$ over the reals, and so we get the ``$\limply$'' direction for \irref{DG} from the ``$\limply$'' direction of \irref{DGall}. Conversely, we only need to prove the ``$\lylpmi$'' direction for \irref{DG}, because the ``$\lylpmi$'' direction for \irref{DGall} follows from it.

\begin{enumerate}
\item[``$\limply$''] We need to show the RHS of \irref{DGall} assuming its LHS. Let $\omega_\vec{y}$ be identical to $\omega$ except where the values for variables $\vec{y}$ are replaced with any initial values $d \in \reals^m$. Consider any solution $\solvar_\vec{y} : [0,T] \to \States$ where $\solvar_\vec{y}(0) = \omega_\vec{y}$ on $\scomplement{\{\D{x},\D{\vec{y}}\}}$, $\solvar_\vec{y}(T) = \nu$, and $\solmodels{\solvar_\vec{y}}{\pevolvein{\D{x}=\usDE,\D{\vec{y}}=\usa\argx\stimes \vec{y}+\usb\argx}{\usQ\argx}}$.

Define $\solvar : [0,T] \to \States$ satisfying:
\[ \solvar(t)(z) \mdefeq
  \begin{cases}
  \solvar_\vec{y}(t)(z)  & z \in  \scomplement{\{\vec{y},\D{\vec{y}}\}}\\
  \omega(z)   & z \in \{\vec{y},\D{\vec{y}}\}
  \end{cases} \]
In other words, $\solvar$ is identical to $\solvar_\vec{y}$ except it holds all of $\vec{y},\D{\vec{y}}$ constant at their initial values in $\omega$. By construction, $\solvar(0)=\omega$ on $\scomplement{\{\D{x}\}}$, and moreover, because $\vec{y}$ is fresh \ie, not mentioned in $\usQ\argx,\usDE$, by~\rref{lem:coincide}, we have that:
\[\solmodels{\solvar}{\pevolvein{\D{x}=\usDE}{\usQ\argx}}\]

Therefore, $\solvar(T) \in \ifsem{\usP\argx}$ from the LHS of \irref{DGall}. Since $\solvar(T)$ coincides with $\solvar_\vec{y}(T)$ on $x$ (since $\vec{y}$ is fresh), by~\rref{lem:coincide} we also have $\solvar_\vec{y}(T) \in \ifsem{\usP\argx}$ as required.

\item[``$\lylpmi$''] We need to show the LHS of \irref{DG} assuming its RHS. Consider a solution $\solvar : [0,T] \to \States$ where $\solvar(0) = \omega$ on $\scomplement{\{\D{x}\}}$, $\solvar(T)=\nu$, and $\solmodels{\solvar}{\pevolvein{\D{x}=\usDE}{\usQ\argx}}$.
Let $\solvar_a(t) \mdefeq I\solvar(t)\sem{\usa\argx}$, and $\solvar_b(t) \mdefeq I\solvar(t)\sem{\usb\argx}$ be the valuation of $\usa\argx,\usb\argx$ along $\solvar$ respectively. Recall that $\solvar_a : [0,T] \to \reals^m \times \reals^m$ and $\solvar_b : [0,T] \to \reals^m$.

By~\cite[Definition 5]{DBLP:journals/jar/Platzer17}, $\solvar_a(t) = I(\usa)\big(I(\solvar(t))(x)\big)$, where $I(\usa)$ is continuous (and similarly for $\solvar_b(t)$). Since $\solvar$ is a continuous function in $t$, both $\solvar_a(t),\solvar_b(t)$ are compositions of continuous functions, and are thus, also continuous functions in $t$. Consider the $m$-dimensional initial value problem:
\[\D{\vec{y}} = \solvar_a(t) \vec{y} + \solvar_b(t),\quad \vec{y}(0) = \itsem{\vec{y}} \]
By \cite[\S14.VI]{Walter1998}, there exists a unique solution $\psi : [0,T] \to \reals^m$ for this system that is defined on the \emph{entire} interval $[0,T]$.
Therefore, we may construct the extended solution $\solvar_\vec{y}$ satisfying:
\[ \solvar_\vec{y}(t)(z) \mdefeq
  \begin{cases}
  \solvar(t)(z)  & z \in  \scomplement{\{\vec{y},\D{\vec{y}}\}} \\
  \psi(t)(z)     & z \in \vec{y} \\
  \frac{d\psi(t)(w)}{dt}     & z = \D{w} \in \D{\vec{y}}
  \end{cases} \]
By definition, $\solvar_\vec{y}(0) = \omega$ on $\scomplement{\{\D{x},\D{\vec{y}}\}}$, and by construction and~\rref{lem:coincide},
\[\solmodels{\solvar_\vec{y}}{\pevolvein{\D{x}=\usDE,\D{\vec{y}}=\usa\argx\stimes \vec{y}+\usb\argx}{\usQ\argx}}\]
Thus, we have $\solvar_\vec{y}(T)\in \ifsem{\usP\argx}$ from the RHS of \irref{DG}. Since $\solvar(T)$ coincides with $\solvar_\vec{y}(T)$ on $x$, again by~\rref{lem:coincide} we have $\nu = \solvar(T) \in \ifsem{\usP\argx}$ as required.
\end{enumerate}

To prove soundness of \irref{DMP} consider any initial state $\omega$ satisfying
\begin{align*}
&\textcircled{1}~\omega \in \ifsem{\dbox{\pevolvein{\D{x}=\usDE}{\usQ\argx}}{(\usQ\argx \limply \usR\argx)}}, \text{and}\\
&\textcircled{2}~\omega \in \ifsem{\dbox{\pevolvein{\D{x}=\usDE}{\usR\argx}}{\usP\argx}}
\end{align*}
We need to show $\omega \in \ifsem{\dbox{\pevolvein{\D{x}=\usDE}{\usQ\argx}}{\usP\argx}}$, \ie, for any solution $\solvar : [0,T] \to \States$ where $\solvar(0) = \omega$ on $\scomplement{\{\D{x}\}}$, and $\solmodels{\solvar}{\pevolvein{\D{x}=\usDE}{\usQ\argx}}$, we have $\solvar(T) \in \ifsem{\usP\argx}$. By definition, we have $\solvar(t) \in \ifsem{\usQ\argx}$ for $t \in [0,T]$, but by \textcircled{1}, we also have that $\solvar(t) \in \ifsem{\usQ\argx \limply \usR\argx}$ for all $t \in [0,T]$. Therefore, $\solvar(t) \in \ifsem{\usR\argx}$ for all $t \in [0,T]$, and hence, $\solmodels{\solvar}{\pevolvein{\D{x}=\usDE}{\usR\argx}}$. Thus, by \textcircled{2}, we have $\solvar(T) \in \ifsem{\usP\argx}$ as required.
\end{proof}

Using the axiomatization from \rref{thm:diffeqaxiomatization} and \rref{lem:dgdmp}, we now derive all of the rules shown in \rref{thm:diffeqax}.

\begin{proof}[Proof of \rref{thm:diffeqax}]
For each rule, we show a derivation from the \dL axioms. The open premises in these derivations correspond to the open premises for each rule.

\begin{itemize}
\item[\irref{dW}]
By \irref{DMP} we obtain two premises corresponding to the two formulas on the left of its implications. The right premise closes using \irref{DW}. The left premise uses \irref{G}, which leaves the open premise of \irref{dW}.
{\footnotesize
\begin{sequentdeduction}[array]
\linfer[DMP]{
  \linfer[G+implyr]{
    \lsequent{\ivr}{\rfvar}
  }
  {\lsequent{\Gamma}{\dbox{\pevolvein{\D{x}=\genDE{x}}{\ivr}}{(\ivr \limply \rfvar)}}} !
  \linfer[DW]{\lclose}
  {\lsequent{\Gamma}{\dbox{\pevolvein{\D{x}=\genDE{x}}{\rfvar}}{\rfvar}}}
}
  {\lsequent{\Gamma}{\dbox{\pevolvein{\D{x}=\genDE{x}}{\ivr}}{\rfvar}}}
\end{sequentdeduction}
}%
\item[\irref{dIgeq}] This rule follows from the \irref{DIgeq} axiom, and also using the equivalence between Lie derivatives and differentials within the context of the ODEs.
{\footnotesize
\renewcommand{\arraystretch}{1.3}
\begin{sequentdeduction}[array]
\linfer[cut]{
  \linfer[testb+implyr]{\lsequent{\Gamma,\ivr}{p\cmp0}}
  {\lsequent{\Gamma}{\dbox{?\ivr}{p\cmp0}}}
  !
  \linfer[cut+DIgeq]{
  \linfer[DE+Dall+qear]{
  \linfer[dW+implyr]{
    \lsequent{\ivr}{\lied[]{\genDE{x}}{p}\geq0}
  }
    {\lsequent{\Gamma,\ivr}{\dbox{\pevolvein{\D{x}=\genDE{x}}{\ivr}}{\lied[]{\genDE{x}}{p}\geq0}}}
  }
    {\lsequent{\Gamma,\ivr}{\dbox{\pevolvein{\D{x}=\genDE{x}}{\ivr}}{\D{(p)}\geq0}}}
  }
  {\lsequent{\Gamma,\dbox{?\ivr}{p\cmp0}}{\dbox{\pevolvein{\D{x}=\genDE{x}}{\ivr}}{p\cmp0}}}
}
  {\lsequent{\Gamma}{\dbox{\pevolvein{\D{x}=\genDE{x}}{\ivr}}{p\cmp0}}}
\end{sequentdeduction}
}%
\item[\irref{dIeq}] The derivation is similar to \irref{dIgeq} using \irref{DIeq} instead of \irref{DIgeq}.
\irref{dIeq} also derives from \irref{dIgeq} using \irref{band} with the equivalence $p=0 \lbisubjunct p\geq0 \land p \leq 0$.

\item[\irref{dC}] We cut in a premise with postcondition $\ivr \limply (\ivr \land \rcfvar)$ and then reduce this postcondition to $\rcfvar$ by \irref{Mb} using the propositional tautology $\rcfvar \limply (\ivr \limply (\ivr \land \rcfvar))$. The right premise after the cut is abbreviated by \textcircled{1}.
{\footnotesize
\begin{sequentdeduction}[array]
\linfer[cut]{
  \linfer[Mb]{
    \lsequent{\Gamma}{\dbox{\pevolvein{\D{x}=\genDE{x}}{\ivr}}{\rcfvar}}
  }
  {\lsequent{\Gamma}{\dbox{\pevolvein{\D{x}=\genDE{x}}{\ivr}}{(\ivr \limply (\ivr \land \rcfvar))}}} !
  \textcircled{1}
  }
  {\lsequent{\Gamma}{\dbox{\pevolvein{\D{x}=\genDE{x}}{\ivr}}{\rfvar}}}
\end{sequentdeduction}
}%
Continuing on \textcircled{1} we use \irref{DMP} to refine the domain constraint, which leaves open the remaining premise of \irref{dC}:
{\footnotesize
\begin{sequentdeduction}[array]
\linfer[DMP]{
    \lsequent{\Gamma}{\dbox{\pevolvein{\D{x}=\genDE{x}}{\ivr \land \rcfvar}}{\rfvar}}
  }{\lsequent{\Gamma,\dbox{\pevolvein{\D{x}=\genDE{x}}{\ivr}}{(\ivr \limply (\ivr \land \rcfvar))}}{\dbox{\pevolvein{\D{x}=\genDE{x}}{\ivr}}{\rfvar}}}
\end{sequentdeduction}
}%

\item[\irref{dG}] This derives by rewriting the RHS with (vectorial) \irref{DG}.
{\footnotesize
\begin{sequentdeduction}[array]
\linfer[DG]{
  \lsequent{\Gamma}{\lexists{\vec{y}}{\dbox{\pevolvein{\D{x}=\genDE{x},\D{\vec{y}}=a(x)\stimes \vec{y}+b(x)}{\ivr}}{\rfvar}}}
}
  {\lsequent{\Gamma}{\dbox{\pevolvein{\D{x}=\genDE{x}}{\ivr}}{\rfvar}}}
\end{sequentdeduction}
}\qedhere%
\end{itemize}
\end{proof}

The soundness of \irref{DIgeq} is proved \cite[Theorem 38]{DBLP:journals/jar/Platzer17} from the mean value theorem as follows. Briefly, consider any solution \m{\solvar : [0,T] \to \States}, and let \m{\solvar_p(t) \mdefeq I\solvar(t)\sem{p}} be the value of $p$ along $\solvar$. We may, without loss of generality, assume $T>0$, and $\solvar_p(0) \cmp 0$. By assumption on the left of the implication in \irref{DIgeq}, $(\usp\argx)'\geq0$, which, by the differential lemma~\cite[Lemma 35]{DBLP:journals/jar/Platzer17}, means $\frac{d\solvar_p(t)}{dt}(\zeta) \geq 0$ for $0 \leq \zeta \leq T$. The mean value theorem implies \(\solvar_p(T) - \solvar_p(0) = \frac{d\solvar_p(t)}{dt}(\zeta)(T-0)\) for some $0 < \zeta < T$. Since $\solvar_p(0) \cmp 0, T>0$, and $\frac{d\solvar_p(t)}{dt}(\zeta) \geq 0$, we have $\solvar_p(T) \cmp 0$ as required.
Conversely, a logical version of the mean value theorem derives from \irref{DIgeq}:
\begin{corollary}[Mean value theorem]
\label{cor:meanvalue}
The following analogue of the mean value theorem derives from \irref{DIgeq}:
\[
\dinferenceRule[MVT|MVT]{}
{\linferenceRule[impl]
  {p \geq 0 \land \ddiamond{\pevolvein{\D{x}=\genDE{x}}{\ivr}}{p<0}}
  {\ddiamond{\pevolvein{\D{x}=\genDE{x}}{\ivr}}{\der[]{p}<0}}
}{}
\]
\end{corollary}
\begin{proof}
This follows immediately by taking contrapositives, dualizing with \irref{diamond}, and then applying \irref{DIgeq}. In the \irref{cut} steps, we weakened the antecedents to implications since $\rfvar \limply (\ivr \limply \rfvar)$ is a propositional tautology for any formula $\rfvar$.
{\footnotesize
\begin{sequentdeduction}[array]
\linfer[andl+diamond+notl+notr]{
\linfer[cut+testb]{
\linfer[cut]{
\linfer[DIgeq]{
  \lclose
}
  {\lsequent{\dbox{\ptest{\ivr}}{p\geq 0},\ivr \limply \dbox{\pevolvein{\D{x}=\genDE{x}}{\ivr}}{\der[]{p}\geq0}}{\dbox{\pevolvein{\D{x}=\genDE{x}}{\ivr}}{p\geq0}}}
}
  {\lsequent{\dbox{\ptest{\ivr}}{p\geq 0},\dbox{\pevolvein{\D{x}=\genDE{x}}{\ivr}}{\der[]{p}\geq0}}{\dbox{\pevolvein{\D{x}=\genDE{x}}{\ivr}}{p\geq0}}}
}
  {\lsequent{p \geq 0,\dbox{\pevolvein{\D{x}=\genDE{x}}{\ivr}}{\der[]{p}\geq0}}{\dbox{\pevolvein{\D{x}=\genDE{x}}{\ivr}}{p\geq0}}}
}
  {\lsequent{p \geq 0 \land \ddiamond{\pevolvein{\D{x}=\genDE{x}}{\ivr}}{p<0}}{\ddiamond{\pevolvein{\D{x}=\genDE{x}}{\ivr}}{\der[]{p}<0}}}
\end{sequentdeduction}
}%
\end{proof}
Intuitively, this version of the mean value theorem asserts that if $p$ changes sign from $p\geq 0$ to $p<0$ along a solution, then its (Lie) derivative must have been negative somewhere along the solution.

\subsection{Extended Axiomatization}
\label{app:extaxiomatization}

We prove the soundness of the axioms shown in \rref{sec:extaxioms} in uniform-substitution style after restating them as concrete \dL formula instances. The axioms of \rref{sec:extaxioms} then derive as uniform substitution instances, because we only use them for concrete instances where the predicates involved ($\rfvar,\ivr$) mention all (proper) variables $x$ changing in the respective system $\D{x}=\genDE{x}$.

For these proofs we will often need to take truncations of solutions $\soltrunc{t}$ (defined in \rref{subsec:backgroundsemantics}). For any solution $\solvar : [0,T] \to \States$, we write $\solvar([a,b]) \in \ifsem{\rfvar}$ to mean $\solvar(\zeta) \in \ifsem{\rfvar}$ for all $a \leq \zeta \leq b$. We use $\solvar((a,b))$ instead when the interval is open, and similarly for the half-open cases. For example, if $\solvar$ obeys the evolution domain constraint $\ivr$ on the interval $[0,T]$, we write $\solvar([0,T]) \in \ifsem{\ivr}$. We will only use this notation when $[a,b]$ is a subinterval of $[0,T]$.

As explained in \rref{sec:extaxioms}, the soundness of the extended axioms require that the system \m{\D{x}=\genDE{x}} always locally evolves $x$. In a uniform substitution formulation for \irref{ContAx+RealIndInAx}, the easiest syntactic check ensuring this condition is that the system contains an equation \(\D{x_1}=1\). But our proofs are more general and only use the assumption that the system locally evolves $x$.
The requirement that \(\D{x_1}=1\) occurs is minor, since such a clock variable can always be added using \irref{DG} if necessary before using the axioms. We elide these \irref{DG} steps for subsequent derivations.

\subsubsection{Existence, Uniqueness, and Continuity}

We prove soundness for concrete versions of the axioms in \rref{lem:uniqcont}.

\begin{lemma}[Continuous existence, uniqueness, and differential adjoints for \rref{lem:uniqcont}]
\label{lem:uniqcontUS}
The following axioms are sound.

\begin{calculus}
\cinferenceRule[UniqAx|Uniq]{}
{\linferenceRule[impll]
  {\big(\ddiamond{\pevolvein{\D{x}=\usDE}{\usQ_1\argx}}{\usP_1\argx}\big) \land
   \big(\ddiamond{\pevolvein{\D{x}=\usDE}{\usQ_2\argx}}{\usP_2\argx}\big) }
  {\axkey{\ddiamond{\pevolvein{\D{x}=\usDE}{\usQ_1\argx \land \usQ_2\argx}}{(\usP_1\argx \lor \usP_2\argx)}}}
}{}

\cinferenceRule[ContAx|Cont]{}
{
 \linferenceRule[impl]
  {x = y}
  {\big(\usg\argx>0 \limply \axkey{\ddiamond{\pevolvein{\D{x}=\genDE{x}}{\usg\argx > 0}}{x \not= y}}\big)}
}{}

\cinferenceRule[DadjointAx|Dadj]{}
{\linferenceRule[equiv]
  {\ddiamond{\pevolvein{\D{y}=-\genDE{y}}{\ivr(y)}}{\,y=x}}
  {\axkey{\ddiamond{\pevolvein{\D{x}=\genDE{x}}{\ivr(x)}}{\,x=y}}}
}{}
\end{calculus}
\end{lemma}
\begin{proof}
For the ODE system $\D{x}=\usDE$, the RHSes, when interpreted as functions on $x$ are continuously differentiable. Therefore, by the Picard-Lindel\"{o}f theorem~\cite[\S10.VI]{Walter1998}, from \emph{any} state $\omega$, there is an interval $[0,\tau), \tau > 0$ on which there is a unique, continuous solution $\solvar : [0,\tau)\to \States$ with $\solvar(0) = \omega$ on $\scomplement{\{\D{x}\}}$. Moreover, the solution may be uniquely extended in time (to the right), up to its maximal open interval of existence~\cite[\S10.IX]{Walter1998}.

We first prove axiom \irref{UniqAx}. Consider an initial state $\omega$, satisfying both conjuncts on the left of the implication in \irref{UniqAx}. Expanding the definition of the diamond modality, this means that there exist two solutions $\solvar_1 : [0,T_1] \to \States$, $\solvar_2 : [0,T_2] \to \States$ from $\omega$ where $\solmodels{\solvar_1}{\pevolvein{\D{x}=\usDE}{\usQ_1\argx}}$ and $\solmodels{\solvar_2}{\pevolvein{\D{x}=\usDE}{\usQ_2\argx}}$, with $\solvar(T_1) \in \ifsem{\usP_1}$ and $\solvar(T_2)\in \ifsem{\usP_2\argx}$.

Now let us first assume $T_1 \leq T_2$. Since both $\solvar_1,\solvar_2$ are solutions starting from $\omega$, the uniqueness of solutions implies that $\solvar_1(t) = \solvar_2(t)$ for $t \in [0,T_1]$. Therefore, since $\solvar_2([0,T_2]) \in \ifsem{\usQ_2\argx}$ and $T_1 \leq T_2$, we have $\solmodels{\solvar_1}{\pevolvein{\D{x}=\usDE}{\usQ_1\argx \land \usQ_2\argx}}$. Since $\solvar_1(T_1) \in \ifsem{\usP_1\argx}$, which implies $\solvar_1(T_1) \in \ifsem{\usP_1\argx \lor \usP_2\argx}$, we therefore have $\omega \in \ifsem{\ddiamond{\pevolvein{\D{x}=\usDE}{\usQ_1\argx \land \usQ_2\argx}}{(\usP_1\argx \lor \usP_2\argx)}}$.

The case for $T_2<T_1$ is similar, except now we have $\solvar_2(T_2) \in \ifsem{\usP_2\argx}$. In either case, we have the required RHS of \irref{UniqAx}:
\[\omega \in \ifsem{\ddiamond{\pevolvein{\D{x}=\usDE}{\usQ_1\argx \land \usQ_2\argx}}{(\usP_1\argx \lor \usP_2\argx)}}\]

Next, we prove axiom \irref{ContAx}. Consider an arbitrary initial state $\iget[state]{\I}$, with \(\iget[state]{\I} \in \ifsem{x=y \land \usg\argx > 0}\).
By Picard-Lindel\"of, there is a solution $\iget[flow]{\If} : [0,\tau) \to \States$ with $\iget[state]{\Iff[0]} = \iget[state]{\I}$ on $\scomplement{\{\D{x}\}}$ for some $\tau>0$ such that \(\solmodels{\iget[flow]{\If}}{\pevolve{\D{x}=\usDE}}\).
Since $\D{x}\not\in\usg\argx$, coincidence (\rref{lem:coincide}) implies \(\itsemst{\iget[state]{\Iff[0]}}{\usg\argx}>0\).
As a composition of continuous evaluation \cite[Definition 5]{DBLP:journals/jar/Platzer17} with the continuous solution $\iget[flow]{\If}$, $\itsemst{\iget[state]{\Iff[t]}}{\usg\argx}$ is a continuous function of time $t$.
Thus, \(\itsemst{\iget[state]{\Iff[0]}}{\usg\argx}>0\) implies \(\itsemst{\iget[state]{\Iff[t]}}{\usg\argx}>0\) for all $t$ in some interval $[0,T]$ with $0 < T < \tau$.
Hence, the truncation $\truncafter{\iget[flow]{\If}}{T}$ satisfies
\[
\solmodels{\truncafter{\iget[flow]{\If}}{T}}{\pevolvein{\D{x}=\usDE}{\usg\argx>0}}
\]
Since $y$ is constant in the ODE but \(\pevolve{\D{x}=\usDE}\) was assumed to locally evolve (for example with \(\D{x_1}=1\)), there is a time $0<\epsilon\leq T$ at which \(\iget[state]{\Iff[\epsilon]}\in \ifsem{x\neq y}\).
Thus, the truncation \(\truncafter{\iget[flow]{\If}}{\epsilon}\) witnesses
\(\iget[state]{\I} \in \ifsem{\ddiamond{\pevolvein{\D{x}=\usDE}{\usg\argx>0}}{x\neq y}}\).

Finally, we prove axiom \irref{DadjointAx}. The ``$\lylpmi$'' direction follows immediately from the ``$\limply$'' direction by swapping the names $x,y$, because \(-(-\usDE)=\usDE\). Therefore, we only prove the ``$\limply$'' direction. Consider an initial state $\omega$ where $\omega \in \ifsem{\ddiamond{\pevolvein{\D{x}=\usDE}{\usQ\argx}}{\,x=y}}$. Unfolding the semantics, there is a solution $\solvar : [0,T] \to \States$, of the system $\D{x}=\usDE$, with $\solvar(0) = \omega$ on $\scomplement{\{\D{x}\}}$, with \(\solvar(t) \in \ifsem{\usQ\argx}\) for all $t$, and $\solvar(T) \in \ifsem{x=y}$.

Note that since the variables $y$ do not appear in the differential equations, its value is held constant along the solution $\solvar$. Now, let us consider the time- and variable reversal $\psi : [0,T]$, where
\[\psi(\tau)(z) \mdefeq
  \begin{cases}
  \solvar(T-\tau)(x_i)  & z = y_i \\
  -\solvar(T-\tau)(\D{x_i}) & z = \D{y_i} \\
  \omega(z)             & \text{otherwise}
  \end{cases}\]
By construction, $\psi(0)$ agrees with $\omega$ on $\scomplement{\{\D{y}\}}$, because \(\solvar(T) \in \ifsem{x=y}\). Moreover, we have explicitly negated the signs of the differential variables $\D{y_i}$ along $\psi$. By uniqueness, the solutions of $\D{x}=-\usDE$ are exactly the time-reversed solutions of $\D{x}=\usDE$. As we have constructed, $\psi$ is the time-reversed solution for $\D{x}=\usDE$ except we have replaced variables $x$ by $y$ instead. Moreover, since $\solvar([0,T]) \in \ifsem{\usQ\argx}$, we also have $\psi([0,T]) \in \ifsem{\usQ\argy}$ by construction and \rref{lem:coincide}. Therefore, $\solmodels{\psi}{\pevolvein{\D{y}=-\usDEy}{\usQ\argy}}$.
Finally, observe that $\psi(T)(y) = \solvar(0)(x)$, but $\psi$ holds the values of $x$ constant, thus $\psi(T)(x) = \solvar(0)(x)$ and so $\psi(T) \in \ifsem{y=x}$. Therefore, $\psi$ is a witness for
\[
\omega \in \ifsem{\ddiamond{\pevolvein{\D{y}=-\usDEy}{\usQ\argy}}{y=x}}
\qedhere
\]
\end{proof}

\subsubsection{Real Induction}

For completeness, we state and prove a succinct version of the real induction principle that we use. This and other principles are in~\cite{clark2012instructor}.
\begin{definition}[Inductive subset~\cite{clark2012instructor}]
\label{def:indsubset}
The subset $S\subseteq\interval{[a,b]}$ is called an \emph{inductive} subset of the compact interval $\interval{[a,b]}$ iff for all $a \leq x \leq b$ and $[a,x) \subseteq S$,
\begin{enumerate}
\item [\textcircled{1}] $x \in S$.
\item [\textcircled{2}] If $x < b$ then $[x,x+\epsilon] \subseteq S$ for some $0 < \epsilon$.
\end{enumerate}
Here, $[a,a)$ is the empty interval, hence \textcircled{1} requires $a \in S$.
\end{definition}

\begin{proposition}[Real induction principle~\cite{clark2012instructor}]
\label{prop:realindR}
The subset $S \subseteq [a,b]$ is inductive if and only if $S=[a,b]$.
\end{proposition}
\begin{proof}
In the ``$\mimply$'' direction, if $S = [a,b]$, then $S$ is inductive by definition. For the ``$\mylpmi$'' direction, let $S \subseteq [a,b]$ be inductive.
Suppose that $S \neq [a,b]$, so that the complement set $\scomplement{S} = [a,b] \setminus S$ is nonempty.
Let $x$ be the infimum of $\scomplement{S}$, and note that $x \in [a,b]$ since $[a,b]$ is left-closed.

First, we note that $[a,x) \subseteq S$. Otherwise, $x$ is not an infimum of $\scomplement{S}$, because there would exist $a \leq y < x$, such that $y \in \scomplement{S}$. By \textcircled{1}, $x \in S$.
Next, if $x = b$, then $S = [a,b]$, contradiction. Thus, $x < b$, and by \textcircled{2}, $[x,x+\epsilon] \subseteq S$ for some $\epsilon > 0$. However, this implies that $x + \epsilon$ is a greater lower bound of $\scomplement{S}$ than $x$, contradiction.
\end{proof}
We now restate and prove a generalized, concrete version of the real induction axiom given in \rref{lem:realindODE}. This strengthened version includes the evolution domain constraint.

\begin{lemma}[Real induction for \rref{lem:realindODE}]
\label{lem:realindODEin}
The following real induction axiom is sound, where $y$ is fresh in \m{\dbox{\pevolvein{\D{x}=\usDE}{\usQ\argx}}{\usP\argx}}.
\begin{align*}
\cinferenceRule[RealIndInAx|RI{$\&$}]{}
{
&\axkey{\dbox{\pevolvein{\D{x}=\usDE}{\usQ\argx}}{\usP\argx}} \lbisubjunct\\
&\lforall{y}{\dbox{\pevolvein{\D{x}=\usDE}{\usQ\argx \land (\usP\argx \lor x=y)}}{\Big(\initassum \limply \\
&\tag{\textcircled{a}} \usP\argx \land\\
&\tag{\textcircled{b}} \big(\ddiamond{\pevolvein{\D{x}=\usDE}{\usQ\argx}}{x\neq y} \limply
       \ddiamond{\pevolvein{\D{x}=\usDE}{\usP\argx}}{x\neq y}\big)\Big)
}}
}{}
\end{align*}

\end{lemma}
\begin{proof}
We label the two conjuncts on the RHS of \irref{RealIndInAx} as \textcircled{a} and \textcircled{b} respectively, as shown above. Consider an initial state $\omega$, we prove both directions of the axiom separately.
\begin{enumerate}
\item[``$\limply$'']
Assume that \textcircled{$\star$} $\omega \in \ifsem {\dbox{\pevolvein{\D{x}=\usDE}{\usQ\argx}}{\usP\argx}}$.
Unfolding the quantification and box modality on the RHS, let $\omega_y$ be identical to $\omega$ except where the values for $y$ are replaced with any initial values $d \in \reals^n$. Consider any solution $\solvar_y : [0,T] \to \States$
of \(\pevolvein{\D{x}=\usDE}{\usQ\argx \land (\usP\argx \lor x=y)}\)
where $\solvar_y(0) = \omega_y$ on $\scomplement{\{\D{x}\}}$, $\solvar_y(T) \in \ifsem{\initassum}$, and
\[\solmodels{\solvar_y}{\pevolvein{\D{x}=\usDE}{\usQ\argx \land (\usP\argx \lor x=y)}}\]

We construct a similar solution $\solvar : [0,T] \to \States$ that keeps $y$ constant at their initial values in $\omega$:
\[ \solvar(t)(z) \mdefeq
  \begin{cases}
  \solvar_y(t)(z)  & z \in  \scomplement{\{y\}}\\
  \omega(z)   & z \in \{y\}
  \end{cases} \]
By construction, $\solvar(0)$ is identical to $\omega$ on $\scomplement{\{\D{x}\}}$. Since $y$ is fresh in \(\pevolvein{\D{x}=\usDE}{\usQ\argx}\), by coincidence~(\rref{lem:coincide}), we must have $\solmodels{\solvar}{\pevolvein{\D{x}=\usDE}{\usQ\argx}}$. By assumption \textcircled{$\star$}, $\solvar(T) \in \ifsem{\usP\argx}$, which implies that $\solvar_y(T) \in \ifsem{\usP\argx}$ by coincidence since $y$ is fresh in $\usP\argx$. This proves conjunct \textcircled{a}.

Unfolding the implication and diamond modality of conjunct \textcircled{b}, we may assume that there is another solution $\psi_y : [0,\tau] \to \States$ starting from $\solvar_y(T)$ with $\psi_y(\tau) \in \ifsem{x \neq y}$ and $\solmodels{\psi_y}{\pevolvein{\D{x}=\usDE}{\usQ\argx}}$. Note that $\psi_y(0) = \solvar_y(T)$ \emph{exactly} rather than just on $\scomplement{\{\D{x}\}}$, because both of these states already have the same values for the differential variables.
We need to show:
\[\solvar_y(T) \in \ifsem{\ddiamond{\pevolvein{\D{x}=\usDE}{\usP\argx}}{x\neq y}}\]
We shall directly show:
\[\solmodels{\psi_y}{\pevolvein{\D{x}=\usDE}{\usP\argx}}\]
In particular, since $\psi_y$ already satisfies the requisite differential equations and $\psi_y(\tau)\in \ifsem{x \neq y}$, it is sufficient to show that it stays in the evolution domain for its entire duration, \ie, $\psi_y([0,\tau]) \in \ifsem{\usP\argx}$.
Let $0 \leq \zeta \leq \tau$ and consider the concatenated solution $\Phi : [0,T+\zeta] \to \States$ defined by:
\[\Phi(t)(z) \mdefeq
  \begin{cases}
  \solvar_y(t)(z)  & t \leq T, z \in \scomplement{\{y\}} \\
  \psi_y(t-T)(z)     & t > T, z \in \scomplement{\{y\}} \\
  \omega(z)        & z \in \{y\}
  \end{cases}\]
As with $\solvar$, the solution $\Phi$ is constructed to keep $y$ constant at their initial values in $\omega$. Since $\psi_y$ must uniquely extend $\solvar_y$~\cite[\S10.IX]{Walter1998}, the concatenated solution $\Phi$ is a solution starting from $\omega$, it solves the system $\D{x}=\usDE$, and it stays in $\usQ\argx$ for its entire duration by coincidence~(\rref{lem:coincide}). Hence, by \textcircled{$\star$}, $\Phi(T+\zeta) \in \ifsem{\usP\argx}$, which implies $\psi_y(\zeta) \in \ifsem{\usP\argx}$ by coincidence~(\rref{lem:coincide}), as required.

\item[``$\lylpmi$'']
We assume the RHS and prove the LHS in initial state $\omega$. If $\omega \notin \ifsem{\usQ\argx}$, then there is nothing to show, because there are no solutions that stay in $\usQ\argx$. Otherwise, consider an arbitrary solution $\solvar : [0,T] \to \States$ starting from $\omega$ such that $\solmodels{\solvar}{\pevolvein{\D{x}=\usDE}{\usQ\argx}}$. We prove $\solvar([0,T]) \in \ifsem{\usP\argx}$ by showing that the subset $S \mdefeq \{\zeta : \solvar(\zeta) \in \ifsem{\usP\argx}\}$ is an inductive subset of $[0,T]$ i.e., satisfies properties \textcircled{1} and \textcircled{2} in \rref{def:indsubset}.
So, assume that $[0,\zeta) \subseteq S$ for some time $0 \leq \zeta \leq T$.

Consider the state $\omega_y$ identical to $\omega$, except where the values for variables $y$ are replaced with the corresponding values of $x$ in $\solvar(\zeta)$:
\[\omega_y(z) \mdefeq
  \begin{cases}
  \solvar(\zeta)(x_i)  & z = y_i \\
  \omega(z)          & \text{otherwise}
  \end{cases}\]

Correspondingly, consider the solution $\solvar_y : [0,\zeta] \to \States$ identical to $\solvar$ but which keeps $y$ constant at initial values in $\omega_y$ rather than in $\omega$:
\[ \solvar_y(t)(z) \mdefeq
  \begin{cases}
  \solvar(t)(z)  & z \in  \scomplement{\{y\}}\\
  \omega_y(z)    & z \in \{y\}
  \end{cases} \]

By coincidence (\rref{lem:coincide}), $\solvar_y$ solves \(\pevolvein{\D{x}=\usDE}{\usQ\argx}\) from initial state $\omega_y$. We still know $\solvar_y([0,\zeta)) \in \ifsem{\usP\argx}$ by coincidence. Additionally, note that $\solvar_y(\zeta) \in \ifsem{x=y}$ by construction. Therefore, $\solvar_y([0,\zeta]) \in \ifsem{\usQ\argx \land (\usP\argx \lor x=y)}$. We now unfold the quantification, box modality and implication on the RHS to obtain:
\[\solvar_y(\zeta) \in \ifsem{\textcircled{a} \land \textcircled{b}} \]

\begin{itemize}
\item[\textcircled{1}] We need to show $\solvar(\zeta) \in \ifsem{\usP\argx}$, but by \textcircled{a}, we have $\solvar_y(\zeta) \in \ifsem{\usP\argx}$. By coincidence~(\rref{lem:coincide}), this implies $\solvar(\zeta) \in \ifsem{\usP\argx}$.

\item[\textcircled{2}] We further assume that $\zeta < T$, and we need to show $\solvar([\zeta,\zeta+\epsilon]) \in \ifsem{\usP\argx}$ for some $\epsilon > 0$.
We shall first discharge the implication in \textcircled{b}, i.e.~we show:
\[\solvar_y(\zeta) \in \ifsem{\ddiamond{\pevolvein{\D{x}=\usDE}{\usQ\argx}}{x \neq y}}\]

Observe that since $\zeta < T$, we may consider the solution that extends from state $\solvar(\zeta)$, \ie, $\psi : [0,T-\zeta] \to \States$, where $\psi(\tau) \mdefeq \solvar(\tau+\zeta)$, and we have $\solmodels{\psi}{\pevolvein{\D{x}=\usDE}{\usQ\argx}}$.

We correspondingly construct the solution that extends from state $\solvar_y(\zeta)$, $\psi_y : [0,T-\zeta] \to \States$ that keeps $y$ constant instead:
\[ \psi_y(t)(z) \mdefeq
  \begin{cases}
  \psi(t)(z)  & z \in  \scomplement{\{y\}}\\
  \solvar_y(\zeta)(z)    & z \in \{y\}
  \end{cases} \]

We already know \(\solvar_y(\zeta) \in \ifsem{x=y}\). We also have $T-\zeta > 0$, and therefore, since the differential equation  is assumed to always locally evolve (for example \(\D{x_1}=1\)), there must be some duration $0<\epsilon<T-\zeta$ after which the value of $x$ has changed from its initial value which is held constant in $y$, i.e., $\psi_y(\epsilon) \in \ifsem{x\neq y}$. In other words, the truncation \(\truncafter{\psi_y}{\epsilon}\) witnesses:
\[\solvar_y(\zeta) \in \ifsem{\ddiamond{\pevolvein{\D{x}=\usDE}{\usQ\argx}}{x\neq y}}\]

Discharging the implication in \textcircled{b}, we obtain:
\[\solvar_y(\zeta) \in \ifsem{\ddiamond{\pevolvein{\D{x}=\usDE}{\usP\argx}}{x\neq y}}\]

Unfolding the semantics gives us a solution, which by uniqueness, yields a truncation $\truncafter{\psi_y}{\epsilon}$ of $\psi_y$, for some $\epsilon > 0$ which starts from $\solvar_y(\zeta)$ and satisfies $\psi_y|_\epsilon([0,\epsilon]) \in \ifsem{\usP\argx}$.

By construction, $\truncafter{\psi_y}{\epsilon}(\tau)$ coincides with $\solvar(\tau+\zeta)$ on $x$ for all $0 \leq \tau \leq \epsilon$, which implies $\solvar([\zeta,\zeta+\epsilon]) \in \ifsem{\usP\argx}$ by \rref{lem:coincide}. \qedhere
\end{itemize}
\end{enumerate}
\end{proof}

\subsection{Derived Rules and Axioms}
\label{app:diaaxioms}

We now derive several useful rules and axioms that we will use in subsequent derivations. Some of which were already proved in~\cite{DBLP:conf/lics/Platzer12b,DBLP:journals/jar/Platzer17} so their proofs are omitted.

\subsubsection{Basic Derived Rules and Axioms}

We start with basic derived rules and axioms of \dL. The axiom \irref{Kd} derives from \irref{K} by dualizing its inner implication with \irref{diamond} \cite{DBLP:conf/lics/Platzer12b}, and the rule \irref{Md} derives by \irref{G} on the outer assumption of \irref{Kd} \cite{DBLP:journals/jar/Platzer17}.
\[
\dinferenceRule[Kd|K${\didia{\cdot}}$]{}
{\linferenceRule[impll]
  {\dbox{\alpha}{(\fvar_2 \limply \fvar_1)}}
  {(\ddiamond{\alpha}{\fvar_2} \limply \axkey{\ddiamond{\alpha}{\fvar_1}})}
}{} \qquad
\dinferenceRule[Md|M${\didia{\cdot}}$]{}
{\linferenceRule
  {\lsequent{\fvar_2}{\fvar_1} \qquad \lsequent{\Gamma}{\ddiamond{\alpha}{\fvar_2}}}
  {\lsequent{\Gamma}{\ddiamond{\alpha}{\fvar_1}}}
}{}
\]

If $\rrfvar(y)$ is true in an initial state and $y$ has no differential equation in $\D{x}=\genDE{x}$, then it trivially continues to hold along solutions to the differential equations from that state because $y$ remains constant along these solutions. Axiom \irref{V} proves this for box modalities (and for diamond modalities in the antecedents):

{\footnotesize
\begin{sequentdeduction}[array]
\linfer[dC]{
  \linfer[V]{
    \lclose
  }
  {\lsequent{\rrfvar(y)}{\dbox{\pevolvein{\D{x}=\genDE{x}}{\ivr}}{\rrfvar(y)}}} !
  \lsequent{\Gamma}{\dbox{\pevolvein{\D{x}=\genDE{x}}{\ivr \land \rrfvar(y)}}{\rfvar}}
}
  {\lsequent{\Gamma,\rrfvar(y)}{\dbox{\pevolvein{\D{x}=\genDE{x}}{\ivr}}{\rfvar}}}
\end{sequentdeduction}
}%

Conversely, if $\rrfvar(y)$ is true in a final state reachable by an ODE $\D{x}=\genDE{x}$, then it must have trivially been true initially. In the derivation below, the open premise labelled \textcircled{1} closes because it leads to a domain constraint that contradicts the postcondition of the diamond modality.

{\footnotesize
\begin{sequentdeduction}[array]
\linfer[cut]{
\linfer[orl]{
    \lsequent{\Gamma, \rrfvar(y), \ddiamond{\pevolvein{\D{x}=\genDE{x}}{\ivr}}{(\rfvar \land \rrfvar(y))}}{\phi}
    !
    \textcircled{1}
}
  {\lsequent{\Gamma, \rrfvar(y) \lor \lnot{\rrfvar(y)}, \ddiamond{\pevolvein{\D{x}=\genDE{x}}{\ivr}}{(\rfvar \land \rrfvar(y))}}{\phi}}
}
  {\lsequent{\Gamma,\ddiamond{\pevolvein{\D{x}=\genDE{x}}{\ivr}}{(\rfvar \land \rrfvar(y))}}{\phi}}
\end{sequentdeduction}
}%

To prove premise \textcircled{1}, we use \irref{V+dC+diamond}
{\footnotesize
\begin{sequentdeduction}[array]
\linfer[diamond+notl]{
\linfer[dC+V]{
\linfer[dW]{
  \lclose
}
  {\lsequent{}{\dbox{\pevolvein{\D{x}=\genDE{x}}{\ivr\land \lnot{\rrfvar(y)}}}{\lnot{(\rfvar \land \rrfvar(y))}}}}
}
  {\lsequent{\lnot{\rrfvar(y)}}{\dbox{\pevolvein{\D{x}=\genDE{x}}{\ivr}}{\lnot{(\rfvar \land \rrfvar(y))}}}}
}
  {\lsequent{\lnot{\rrfvar(y)}, \ddiamond{\pevolvein{\D{x}=\genDE{x}}{\ivr}}{(\rfvar \land \rrfvar(y))}}{\lfalse}}
\end{sequentdeduction}
}%

In the sequel, we omit these routine steps and label proof steps that manipulate constant context assumptions with \irref{V} directly.

We now prove \rref{cor:diadiffeqax}, which provides tools for working with diamond modalities involving ODEs.
\begin{proof}[Proof of \rref{cor:diadiffeqax}]
Axiom \irref{dDR} derives from \irref{DMP} by dualizing with the \irref{diamond} axiom.

{\footnotesize
\begin{sequentdeduction}[array]
\linfer[diamond+notr+notl]{
\linfer[DMP+DW]{
  \lclose
}
  {\lsequent{\dbox{\pevolvein{\D{x}=\genDE{x}}{\rrfvar}}{\ivr},\dbox{\pevolvein{\D{x}=\genDE{x}}{\ivr}}{\lnot{\rfvar}}}{\dbox{\pevolvein{\D{x}=\genDE{x}}{\rrfvar}}{\lnot{\rfvar}}}}
}
  {\lsequent{\dbox{\pevolvein{\D{x}=\genDE{x}}{\rrfvar}}{\ivr},\ddiamond{\pevolvein{\D{x}=\genDE{x}}{\rrfvar}}{\rfvar}}{\ddiamond{\pevolvein{\D{x}=\genDE{x}}{\ivr}}{\rfvar}}}
\end{sequentdeduction}
}%

Rule \irref{gddR} derives from \irref{dDR} by simplifying its left premise with rule \irref{dW}.
\end{proof}

\subsubsection{Extended Derived Rules and Axioms}

We derive additional rules and axioms that make use of our axiomatic extensions.

\begin{corollary}[Extended diamond modality rules and axioms] \label{cor:extended-diamond-deriveds}
The following are derived axioms in \dL extended with \irref{UniqAx+DadjointAx}.

\begin{calculus}
\dinferenceRule[decompand|$\didia{\&\land}$]{}
{\linferenceRule[equivl]
  {\ddiamond{\pevolvein{\D{x}=\genDE{x}}{\ivr}}{\rfvar} \land \ddiamond{\pevolvein{\D{x}=\genDE{x}}{\rrfvar}}{\rfvar}}
  {\axkey{\ddiamond{\pevolvein{\D{x}=\genDE{x}}{\ivr \land \rrfvar}}{\rfvar}}}
}{}
\dinferenceRule[diareflect|reflect$_{\didia{\cdot}}$]{}
{\linferenceRule[equivl]
  {\lexists{x}{(\rrfvar(x) \land \ddiamond{\pevolvein{\D{x}=-\genDE{x}}{\ivr(x)}}{\rfvar(x)})}}
  {\axkey{\lexists{x}{(\rfvar(x) \land \ddiamond{\pevolvein{\D{x}=\genDE{x}}{\ivr(x)}}{\rrfvar(x)})}}}
}{}
\end{calculus}
\end{corollary}
\begin{proof}
The equivalence \irref{decompand} derives from \irref{gddR} for the ``$\limply$" direction, because of the propositional tautologies $\ivr \land \rrfvar \limply \ivr$ and $\ivr \land \rrfvar \limply \rrfvar$. The ``$\lylpmi$" direction is an instance of \irref{UniqAx} by setting $\rfvar_1,\rfvar_2$ to $\rfvar$, and $\ivr_1,\ivr_2$ to $\ivr,\rrfvar$ respectively.

We prove \irref{diareflect} from \irref{DadjointAx}. Both implications are proved separately and the ``$\lylpmi$'' direction follows by instantiating the proof of the ``$\limply$'' direction, since $-(-\genDE{x}) = \genDE{x}$.

In the derivation below, the succedent is abbreviated with $\phi \mdefequiv \lexists{z}{(\rrfvar(z) \land \ddiamond{\pevolvein{\D{z}=-\genDE{z}}{\ivr(z)}}{\rfvar(z)})}$, where we have renamed the variables for clarity. The first \irref{cut+Md} step introduces an existentially quantified $z$ under the diamond modality using the provable first-order formula $\rrfvar(y) \limply \lexists{z}{(z=y \land \rrfvar(z))}$. Next, Barcan \irref{dBarcan} moves the existentially quantified $z$ out of the diamond modality.
{\footnotesize
\begin{sequentdeduction}[array]
\linfer[existsl]{
\linfer[cut+Md]{
\linfer[dBarcan]{
\linfer[existsl]{
  \lsequent{\rfvar(y), \ddiamond{\pevolvein{\D{y}=\genDE{y}}{\ivr(y)}}{(z=y \land \rrfvar(z))}}{\phi}
}
  {\lsequent{\rfvar(y), \lexists{z}{\ddiamond{\pevolvein{\D{y}=\genDE{y}}{\ivr(y)}}{(z=y \land \rrfvar(z))}}}{\phi}}
}
  {\lsequent{\rfvar(y), \ddiamond{\pevolvein{\D{y}=\genDE{y}}{\ivr(y)}}{\lexists{z}{(z=y \land \rrfvar(z))}}} {\phi}}
}
  {\lsequent{\rfvar(y), \ddiamond{\pevolvein{\D{y}=\genDE{y}}{\ivr(y)}}{\rrfvar(y)}} {\phi}}
}
  {\lsequent{\lexists{x}{(\rfvar(x) \land \ddiamond{\pevolvein{\D{x}=\genDE{x}}{\ivr(x)}}{\rrfvar(x)})}} {\phi}}
\end{sequentdeduction}
}%

Continuing, since $z$ is not bound in $\D{y}=\genDE{y}$, a \irref{V} step allows us to move $\rrfvar(z)$ out from under the diamond modality in the antecedents. We then use \irref{DadjointAx} to flip the differential equations from evolving $y$ forwards to evolving $z$ backwards. The \irref{V+Kd} step uses the fact that the (new) ODE does not modify $y$ so that $\rfvar(y)$ remains true along the ODE, which allows its postcondition to be strengthened to $\rfvar(z)$, yielding a witness for the succedent.
{\footnotesize
\begin{sequentdeduction}[array]
\linfer[V]{
\linfer[DadjointAx]{
\linfer[V+Kd]{
\linfer[existsr]{
  \lclose
}
  {\lsequent{\rrfvar(z), \ddiamond{\pevolvein{\D{z}=-\genDE{z}}{\ivr(z)}}{\rfvar(z)}} {\phi}}
}
  {\lsequent{\rfvar(y),\rrfvar(z), \ddiamond{\pevolvein{\D{z}=-\genDE{z}}{\ivr(z)}}{z=y}} {\phi}}
}
  {\lsequent{\rfvar(y),\rrfvar(z), \ddiamond{\pevolvein{\D{y}=\genDE{y}}{\ivr(y)}}{z=y}} {\phi}}
}
  {\lsequent{\rfvar(y), \ddiamond{\pevolvein{\D{y}=\genDE{y}}{\ivr(y)}}{(z=y \land \rrfvar(z))}}{\phi}}
\end{sequentdeduction}
}%
\end{proof}

An invariant reflection principle derives from \irref{diareflect}: the negation of invariants $\rfvar(x)$ of the forwards differential equations $\D{x}=\genDE{x}$ are invariants of the backwards differential equations $\D{x}=-\genDE{x}$.

\begin{corollary}[Reflection]
The invariant reflection axiom derives from axiom \irref{DadjointAx}:
\[
\dinferenceRule[reflect|reflect]{}
{\linferenceRule[equivl]
  {\lforall{x}{\big(\lnot{\rfvar(x)} \limply \dbox{\pevolvein{\D{x}=-\genDE{x}}{\ivr(x)}}{\lnot{\rfvar(x)}}\big)}}
  {\axkey{\lforall{x}{\big(\rfvar(x) \limply \dbox{\pevolvein{\D{x}=\genDE{x}}{\ivr(x)}}{\rfvar(x)}\big)}}}
}
{}
\]
\end{corollary}
\begin{proof}
The axiom derives from \irref{diareflect} by instantiating it with $R\mdefequiv \lnot P$ and negating both sides of the equivalence with \irref{diamond}.
\end{proof}

Finally, we derive the real induction rule corresponding to axiom \irref{RealIndInAx}. We will use the $\ddnextsmall$ abbreviation from~\rref{sec:extaxioms} in the statement of the rule but explicitly include $\initassum$ which was elided for brevity in \rref{sec:semialg}.

\begin{corollary}[Real induction rule with domain constraints for \rref{cor:realind}]
\label{cor:realindin}
This rule (with two stacked premises) derives from \irref{RealIndInAx+DadjointAx+UniqAx}.
\[\dinferenceRule[realindin|rI{$\&$}]{}
{\linferenceRule
  {
  \begin{aligned}
  \lsequent{\initassum,\rfvar,\ivr,\dprogressinsmall{\D{x}=\genDE{x}}{\ivr}}{\dprogressinsmall{\D{x}=\genDE{x}}{\rfvar}}\\
  \lsequent{\initassum,\lnot{\rfvar},\ivr,\dprogressinsmall{\D{x}=-\genDE{x}}{\ivr}}{\dprogressinsmall{\D{x}=-\genDE{x}}{\lnot{\rfvar}}}
  \end{aligned}
  }
  {\lsequent{\rfvar}{\dbox{\pevolvein{\D{x}=\genDE{x}}{\ivr}}{\rfvar}}}
}{}\]
\end{corollary}
\begin{proof}
We label the premises of \irref{realindin} with \textcircled{a} for the top premise and \textcircled{b} for the bottom premise.
The derivation starts by rewriting the succedent with \irref{RealIndInAx}. We have abbreviated the second conjunct from this step with $\rrfvar \mdefequiv \dprogressinsmall{\D{x}=\genDE{x}}{\ivr} \limply \dprogressinsmall{\D{x}=\genDE{x}}{\rfvar}$. The \irref{Mb} step rewrites the postcondition with the propositional tautology $(\initassum \limply \rfvar \land \rrfvar) \lbisubjunct (\initassum \limply \rfvar) \land (\initassum \limply \rfvar \limply \rrfvar)$. We label the two premises after \irref{band+andr} with \textcircled{1} and \textcircled{2} respectively.
{\footnotesize
\begin{sequentdeduction}[array]
\linfer[RealIndInAx]{
\linfer[allr]{
\linfer[Mb]{
\linfer[band+andr]{
  \textcircled{1} ! \textcircled{2}
}
  {\lsequent{\rfvar}{\dbox{\pevolvein{\D{x}=\genDE{x}}{\ivr \land (\rfvar \lor \initassum)}}{\big( (\initassum \limply \rfvar) \land (\initassum \limply \rfvar \limply \rrfvar)\big)}}}
}
  {\lsequent{\rfvar}{\dbox{\pevolvein{\D{x}=\genDE{x}}{\ivr \land (\rfvar \lor \initassum)}}{(\initassum \limply \rfvar \land \rrfvar)}}}
}
  {\lsequent{\rfvar}{\lforall{y}{\dbox{\pevolvein{\D{x}=\genDE{x}}{\ivr \land (\rfvar \lor \initassum)}}{(\initassum \limply \rfvar \land \rrfvar)}}}}
}
  {\lsequent{\rfvar}{\dbox{\pevolvein{\D{x}=\genDE{x}}{\ivr}}{\rfvar}}}
\end{sequentdeduction}
}%
We continue from open premise \textcircled{2} with a \irref{dW} step, which yields the premise \textcircled{a} of \irref{realindin} (by unfolding our abbreviation for $\rrfvar$):
{\footnotesize
\begin{sequentdeduction}[array]
\linfer[dW]{
\linfer[implyr]{
  \textcircled{a}
}
  {\lsequent{\ivr}{(\initassum \limply \rfvar \limply \rrfvar)}}
}
  {\lsequent{\rfvar}{\dbox{\pevolvein{\D{x}=\genDE{x}}{\ivr \land (\rfvar \lor \initassum)}}{(\initassum \limply \rfvar \limply \rrfvar)}}}
\end{sequentdeduction}
}%
We continue from the open premise \textcircled{1} by case splitting on the left with the provable real arithmetic formula $\initassum \lor x\neq y$. This yields two further cases labelled \textcircled{3} and \textcircled{4}.
{\footnotesize
\begin{sequentdeduction}[array]
\linfer[cut+qear]{
\linfer[orl]{
  \textcircled{3} ! \textcircled{4}
}
  {\lsequent{\initassum \lor x\neq y,\rfvar}{\dbox{\pevolvein{\D{x}=\genDE{x}}{\ivr \land (\rfvar \lor \initassum)}}{(\initassum \limply \rfvar)}}}
}
  {\lsequent{\rfvar}{\dbox{\pevolvein{\D{x}=\genDE{x}}{\ivr \land (\rfvar \lor \initassum)}}{(\initassum \limply \rfvar)}}}
\end{sequentdeduction}
}%

For \textcircled{3}, since $\initassum$ initially, we are trivially done, because $\rfvar(y)$ is true initially, and $y$ is held constant by $\D{x}=\genDE{x}$. This is proved with an \irref{Mb} step followed by \irref{V}.
{\footnotesize
\begin{sequentdeduction}[array]
  \linfer[cut+Mb]{
  \linfer[V]{
    \lclose
  }
    {\lsequent{\rfvar(y)}{\dbox{\pevolvein{\D{x}=\genDE{x}}{\ivr \land (\rfvar \lor \initassum)}}{\rfvar(y)}}}
  }
  {\lsequent{\initassum,\rfvar}{\dbox{\pevolvein{\D{x}=\genDE{x}}{\ivr \land (\rfvar \lor \initassum)}}{(\initassum \limply \rfvar)}}}
\end{sequentdeduction}
}%
For \textcircled{4}, where $x\neq y$, we first use \irref{DW} to assume $\ivr$ in the postcondition. Abbreviate $\rtfvar \mdefequiv \rfvar \lor \initassum$. We then move into the diamond modality, and use \irref{diareflect}. We cut the succedent of premise \textcircled{b}. The resulting two open premises are labelled \textcircled{5} and \textcircled{6}.
{\footnotesize
\begin{sequentdeduction}
  \linfer[Mb+DW]{
  \linfer[diamond+notr]{
  \linfer[diareflect]{
  \linfer[cut]{
    \textcircled{5} \qquad \textcircled{6}
  }
    {\lsequent{\initassum \land \ivr \land \lnot{\rfvar},\ddiamond{\pevolvein{\D{x}=-\genDE{x}}{\ivr \land \rtfvar}}{(x \neq y \land \rfvar)}}{\lfalse}}
  }
    {\lsequent{x\neq y \land \rfvar,\ddiamond{\pevolvein{\D{x}=\genDE{x}}{\ivr \land \rtfvar}}{(\ivr \land \initassum \land \lnot{\rfvar})}}{\lfalse}}
  }
  {\lsequent{x\neq y,\rfvar}{\dbox{\pevolvein{\D{x}=\genDE{x}}{\ivr \land \rtfvar}}{(\ivr \limply \initassum \limply \rfvar)}}}
  }
  {\lsequent{x\neq y,\rfvar}{\dbox{\pevolvein{\D{x}=\genDE{x}}{\ivr \land \rtfvar}}{(\initassum \limply \rfvar)}}}
\end{sequentdeduction}
}%
The premise \textcircled{5} reduces to premise \textcircled{b}, after we use \irref{Md+gddR} to simplify the diamond modality assumption in the antecedents. The \irref{gddR} step proves with the propositional tautology $\ivr \land \rtfvar \limply \ivr$.
{\footnotesize
\begin{sequentdeduction}[array]
  \linfer[Md]{
  \linfer[gddR]{
  \linfer[]{
    \textcircled{b}
  }
    {\lsequent{\initassum, \ivr, \lnot{\rfvar},\ddiamond{\pevolvein{\D{x}=-\genDE{x}}{\ivr}}{x\neq y }}{\dprogressinsmall{\D{x}=-\genDE{x}}{\lnot{\rfvar}}}}
  }
  {\lsequent{\initassum, \ivr, \lnot{\rfvar},\ddiamond{\pevolvein{\D{x}=-\genDE{x}}{\ivr \land \rtfvar}}{x\neq y }}{\dprogressinsmall{\D{x}=-\genDE{x}}{\lnot{\rfvar}}}}
  }
  {\lsequent{\initassum, \ivr, \lnot{\rfvar},\ddiamond{\pevolvein{\D{x}=-\genDE{x}}{\ivr \land \rtfvar}}{(x \neq y \land \rfvar)}}{\dprogressinsmall{\D{x}=-\genDE{x}}{\lnot{\rfvar}}}}
\end{sequentdeduction}
}%

Premise \textcircled{6} uses the cut. The first \irref{gddR} step unfolds the syntactic abbreviation in $\ddnextsmall$, and drops $\ivr$ from the domain constraint with the tautology $\ivr \land \rtfvar \limply \rtfvar$. We then combine the two diamond modalities in the antecedents with \irref{UniqAx} and simplify its resulting domain constraint and postcondition with \irref{Md+gddR}, which respectively use the tautologies $x{\neq} y \lor x {\neq} y \land \rfvar \lbisubjunct x {\neq} y$ and $\lnot{\rfvar} \land \rtfvar \limply \initassum$.
{\footnotesize
\begin{sequentdeduction}[array]
  \linfer[gddR]{
  \linfer[UniqAx]{
  \linfer[Md]{
  \linfer[gddR]{
      \lsequent{\ddiamond{\pevolvein{\D{x}=-\genDE{x}}{\initassum}}{\,x {\neq} y}}{\lfalse}
      }
      {\lsequent{\ddiamond{\pevolvein{\D{x}=-\genDE{x}}{\lnot{\rfvar} \land \rtfvar}}{\,x{\neq} y}}{\lfalse}}
    }
    {\lsequent{\ddiamond{\pevolvein{\D{x}=-\genDE{x}}{\lnot{\rfvar} \land \rtfvar}}{(x{\neq} y \lor x {\neq} y \land \rfvar)}}{\lfalse}}
  }
  {\lsequent{\ddiamond{\pevolvein{\D{x}=-\genDE{x}}{\lnot{\rfvar}}}{x{\neq} y},\ddiamond{\pevolvein{\D{x}=-\genDE{x}}{\rtfvar}}{(x {\neq} y \land \rfvar)}}{\lfalse}}
  }
  {\lsequent{\dprogressinsmall{\D{x}=-\genDE{x}}{\lnot{\rfvar}},\ddiamond{\pevolvein{\D{x}=-\genDE{x}}{\ivr \land \rtfvar}}{(x {\neq} y \land \rfvar)}}{\lfalse}}
\end{sequentdeduction}
}%
We complete the proof of \textcircled{6} by dualizing and \irref{dW}, since the domain constraint and postcondition of the diamond modality in the antecedents is contradictory.
{\footnotesize
\begin{sequentdeduction}[array]
  \linfer[diamond+notl]{
  \linfer[dW]{
    \lclose
  }
  {\lsequent{}{\dbox{\pevolvein{\D{x}=-\genDE{x}}{\initassum}}{\initassum}}}
  }
  {\lsequent{\ddiamond{\pevolvein{\D{x}=-\genDE{x}}{\initassum}}{x {\neq} y}}{\lfalse}}
\end{sequentdeduction}
}%
\end{proof}

The rule \irref{realindin} discards any additional context in the antecedents of its premises.
Intuitively, this is due to the use of \irref{RealIndInAx} which focuses on particular states along trajectories of the ODE $\D{x}=\genDE{x}$; it would be unsound to keep any assumptions about the initial state that depend on $x$ because we may not be at the initial state!
On the other hand, assumptions that do not depend on $x$ remain true along the ODE. They can be kept with uses of \irref{V} throughout the derivation above or added into $\ivr$ before using \irref{realindin} by a \irref{dC} that proves with \irref{V}.
We elide these additional steps, and directly use rule \irref{realindin} while keeping these constant context assumptions around.

\irref{realindin} is derived with the $\ddnextsmall$ modality. However, it is easy to convert between the two modalities with the following derived axiom.
\begin{corollary}[Initial state inclusion]
\label{cor:bigsmallequiv}
This is a derived axiom.
\[\dinferenceRule[bigsmallequiv|Init]{}
{
\linferenceRule[impl]
  {\initassum \land \rfvar}
  {\big(\axkey{\dprogressinsmall{\D{x}=\genDE{x}}{\rfvar}} \lbisubjunct \dprogressin{\D{x}=\genDE{x}}{\rfvar}\big)}
}
{}\]
\end{corollary}
\begin{proof}
We derive both directions of the equivalence by unfolding the syntactic abbreviations.
In the ``$\limply$'' direction, a \irref{gddR} step is sufficient, because $\rfvar \limply \rfvar \lor \initassum$ is a propositional tautology:
{\footnotesize
\begin{sequentdeduction}[array]
  \linfer[gddR]{
    \lclose
  }
  {\lsequent{\initassum,\rfvar,\ddiamond{\pevolvein{\D{x}=\genDE{x}}{\rfvar}}{x\neq y}}{\ddiamond{\pevolvein{\D{x}=\genDE{x}}{\rfvar \lor \initassum}}{x \neq y}}}
\end{sequentdeduction}
}%
In the ``$\lylpmi$'' direction, we start with a \irref{dDR} step to reduce to the box modality. We then cut $P(y)$, which proves from antecedents $x=y,\rfvar$. This is then introduced in the domain constraint by \irref{V}, which allows us to close the proof by \irref{dW}.
{\footnotesize
\begin{sequentdeduction}[array]
  \linfer[dDR]{
  \linfer[cut]{
  \linfer[V]{
  \linfer[dW]{
    \lclose
  }
    \lsequent{}{\dbox{\pevolvein{\D{x}=\genDE{x}}{(\rfvar \lor \initassum) \land \rfvar(y)}}{\rfvar}}
  }
    {\lsequent{\rfvar(y)}{\dbox{\pevolvein{\D{x}=\genDE{x}}{\rfvar \lor \initassum}}{\rfvar}}}
  }
    {\lsequent{\initassum,\rfvar}{\dbox{\pevolvein{\D{x}=\genDE{x}}{\rfvar \lor \initassum}}{\rfvar}}}
  }
  {\lsequent{\initassum,\rfvar,\ddiamond{\pevolvein{\D{x}=\genDE{x}}{\rfvar \lor \initassum}}{x \neq y}}{\ddiamond{\pevolvein{\D{x}=\genDE{x}}{\rfvar}}{x\neq y}}}
\end{sequentdeduction}
}%
\end{proof}

With \irref{bigsmallequiv}, we may equivalently rewrite the premises of \irref{realindin} with the $\ddnext$ modality, and so we will directly use it with $\ddnext$ instead of $\ddnextsmall$. Similarly, rule \irref{realind} from \rref{cor:realind} derives from \irref{realindin} with $\ivr \equiv \ltrue$.

Continuity (and local progress characterizations) generalize to open semialgebraic $\ivr$. Intuitively, the soundness of \irref{ContAx} only required that $p>0$ characterized an open set. The derived axiom \irref{ContOpen} builds on \irref{ContAx} to prove a stronger analogue for any formula characterizing an open semialgebraic set.

\begin{corollary}[Open continuity] \label{cor:ContOpen}
Let $\ivr$ be a formula characterizing an open semialgebraic set, this axiom derives from \irref{ContAx+UniqAx}.
\[
\dinferenceRule[ContOpen|Cont$^O$]{open continuous existence}
{
 \linferenceRule[impl]
  {\initassum}
  {\big(\ivr \limply \axkey{\dprogressinsmall{\D{x}=\genDE{x}}{\ivr}{}}\big)}
}{}
\]
\end{corollary}
\begin{proof}
Since $\ivr$ characterizes an open, semialgebraic set, by the finiteness theorem~\cite[Theorem 2.7.2]{Bochnak1998} for open semialgebraic sets, $\ivr$ may be written as follows ($q_{ij}$ are polynomials):
\[\ivr \equiv \lorfold_{i=0}^M \landfold_{j=0}^{m(i)} q_{ij}> 0\]
We may assume that $\ivr$ is written in this form by an application of \irref{qear} (and congruence or \irref{dDR}). Throughout this proof, we collapse similar premises in derivations and index them by $i,j$.
We abbreviate the $i$-th disjunct of $\ivr$ with $\ivr_i \mdefequiv \landfold_{j=0}^{m(i)} q_{ij}> 0$.

We start by splitting on the outermost disjunction of $\ivr$ with \irref{orl}. For each resulting premise (indexed by $i$), we select the corresponding disjunct of $\ivr$ to prove local progress. The domain change with \irref{gddR} proves since \(\ivr_i {\limply} \ivr\) is a propositional tautology for each $i$.
{\footnotesize
\begin{sequentdeduction}[array]
\linfer[orl]{
\linfer[gddR]{
  \lsequent{\initassum,\ivr_i}{\dprogressinsmall{\D{x}=\genDE{x}}{\ivr_i}}
}
  {\lsequent{\initassum,\ivr_i}{\dprogressinsmall{\D{x}=\genDE{x}}{\ivr}}}
}
  {\lsequent{\initassum,\ivr}{\dprogressinsmall{\D{x}=\genDE{x}}{\ivr}}}
\end{sequentdeduction}
}%

Now, we only need to prove local progress in $\ivr_i$. We make use of \irref{decompand} to split up the conjunction in $\ivr_i$. This leaves premises (indexed by $j$) which are all closed by \irref{ContAx}.
{\footnotesize
\begin{sequentdeduction}[array]
\linfer[decompand+andr]{
\linfer[ContAx]{
  \lclose
}
  {\lsequent{\initassum,q_{ij}> 0}{\dprogressinsmall{\D{x}=\genDE{x}}{q_{ij}> 0}}}
}
  {\lsequent{\initassum,\landfold_{j=0}^{m(i)} q_{ij}> 0}{\dprogressinsmall{\D{x}=\genDE{x}}{\ivr_i}}}
\end{sequentdeduction}
}%
\end{proof}

\section{Completeness}
\label{app:completeness}

This appendix gives the full completeness arguments for the derived rules \irref{dRI} and \irref{sAI} (and the local progress conditions \irref{Lpiff}). We prove the completeness of \irref{dRI} by showing that \irref{DRI} is a derived axiom. We take a similar approach for \irref{sAI}, although the precise form of the resulting derived axiom is more involved. We take syntactic approaches to proving completeness of \irref{dRI} and \irref{sAI} to demonstrate the versatility of the \dL calculus and make it possible to disprove invariance properties (as opposed to just failing to apply a complete proof rule). We refer the readers to other presentations~\cite{DBLP:conf/tacas/GhorbalP14,DBLP:conf/emsoft/LiuZZ11} for purely semantical completeness arguments for invariants. Recall from~\rref{app:extaxiomatization}, that axioms \irref{Cont+RealIndInAx} have an additional syntactic requirement, e.g.~$\D{x_1}=1$. We assume that the syntactic requirement is met throughout this appendix, using \irref{DG} if necessary, but elide the explicit proof steps.

The $\ddnext$ and $\ddnextsmall$ modalities have their corresponding semantic readings only when the assumption $\initassum$ is true in the initial state~\rref{sec:extaxioms}. This additional assumption was elided in \rref{sec:semialg} for brevity, but is expanded in full in this appendix. For clarity, we re-state the derived axioms from \rref{sec:semialg} with this additional assumption where necessary. The ideas for these proofs are in the main paper.

\subsection{Progress Formulas}
\label{app:progressformulas}

We start with following useful observation on rearrangements of the progress formulas for polynomials:
\begin{proposition}
\label{prop:rearrangement}
Let $N$ be the rank of $p$. The following are provable equivalences on the progress and differential radical formulas.
\begin{align}
\sigliedgt{\genDE{x}}{p} \lbisubjunct&~p > 0 \lor (p = 0 \land \lied[]{\genDE{x}}{p} > 0) \label{eq:sigliedgtrearrangement}\\
&\lor \dots\nonumber\\
&\lor \big(p=0 \land \lied[]{\genDE{x}}{p} = 0 \land \dots \land \lied[N-2]{\genDE{x}}{p} = 0 \land \lied[N-1]{\genDE{x}}{p} > 0\big)\nonumber\\
\sigliedgeq{\genDE{x}}{p} \lbisubjunct&~p\geq 0 \land \big(p=0 \limply \lied[]{\genDE{x}}{p} \geq 0\big) \label{eq:sigliedgeqrearrangement}\\
&\land \dots\nonumber\\
&\land \big(p=0 \land \lied[]{\genDE{x}}{p} = 0 \land \dots \land \lied[N-2]{\genDE{x}}{p} = 0 \limply \lied[N-1]{\genDE{x}}{p} \geq 0\big)\nonumber\\
\lnot{(\sigliedgt{\genDE{x}}{p})} \lbisubjunct &~\sigliedgeq{\genDE{x}}{(-p)} \label{eq:sigliedfullrearrangement} \\
\lnot{(\sigliedzero{\genDE{x}}{p})} \lbisubjunct&~\sigliedgt{\genDE{x}}{p} \lor \sigliedgt{\genDE{x}}{(-p)} \label{eq:sigliedzerorearrangement}
\end{align}
\end{proposition}
\begin{proof}
We prove the equivalences case by case, in order. We will use the following real arithmetic equivalences:
\begin{align*}
p \geq 0 &\lbisubjunct p=0 \lor p > 0\\
-p \geq 0 \land p \geq 0 &\lbisubjunct p=0\\
\lnot{(p > 0)} &\lbisubjunct -p \geq 0
\end{align*}
Note, also that Lie derivation is linear i.e.~$\lied[i]{\genDE{x}}{(-p)} = -(\lied[i]{\genDE{x}}{p})$ is provable in real arithmetic for any $i$.
\begin{itemize}
\item[\rref{eq:sigliedgtrearrangement}] This equivalence follows by real arithmetic, and simplifying with propositional rearrangement as follows (here, the remaining conjuncts of $\sigliedgt{\genDE{x}}{p}$ are abbreviated to $\dots$):
\begin{align*}
&p \geq 0 \land \Big((p = 0 \limply \lied[]{\genDE{x}}{p} \geq 0) \land \dots \Big)\lbisubjunct \\
& p > 0 \land \Big((p = 0 \limply \lied[]{\genDE{x}}{p} \geq 0) \land \dots \Big) \lor \\
& p=0 \land \Big((p = 0 \limply \lied[]{\genDE{x}}{p} \geq 0) \land \dots \Big)
\end{align*}

The first disjunct on the RHS simplifies by real arithmetic to $p > 0$ since all of the implicational conjuncts contain $p=0$ on the left of an implication. The latter simplifies to $p=0 \land \Big(\lied[]{\genDE{x}}{p} \geq 0 \land \dots\Big)$, yielding the provable equivalence:
\[\sigliedgt{\genDE{x}}{p} \lbisubjunct p > 0 \lor p=0 \land \Big(\lied[]{\genDE{x}}{p} \geq 0 \land \dots\Big)\]
The equivalence $\rref{eq:sigliedgtrearrangement}$ follows by iterating this expansion for the conjuncts corresponding to higher Lie derivatives.

\item[\rref{eq:sigliedgeqrearrangement}] This equivalence proves by expanding the formula $\sigliedgeq{\genDE{x}}{p}$ which yields a disjunction between $\sigliedgt{\genDE{x}}{p}$ and $\sigliedzero{\genDE{x}}{p}$. The latter formula is used to relax the strict inequality in the last conjunct of $\sigliedgt{\genDE{x}}{p}$ to a non-strict inequality.
\item[\rref{eq:sigliedfullrearrangement}] This equivalence follows by negating both sides of the equivalence \rref{eq:sigliedgtrearrangement} and moving negations on the RHS inwards with propositional tautologies, yielding the provable equivalence:
\begin{align*}
\lnot{(\sigliedgt{\genDE{x}}{p})} \lbisubjunct &\Big( \lnot{(p > 0)} \land (p = 0 \limply \lnot{(\lied[]{\genDE{x}}{p} > 0)})\\
&\land \dots\\
&\land \big(p{=}0 \land \lied[]{\genDE{x}}{p} {=} 0 \land \dots \land \lied[N-2]{\genDE{x}}{p} {=} 0 \limply \lnot{(\lied[N-1]{\genDE{x}}{p}{>}0)}\big)\Big)
\end{align*}
The desired equivalence proves by further rewriting the above RHS with real arithmetic and equivalence \rref{eq:sigliedgeqrearrangement}.
\item[\rref{eq:sigliedzerorearrangement}] By \rref{eq:sigliedfullrearrangement}, we have the provable equivalence:
\[\lnot{(\sigliedgt{\genDE{x}}{p})} \land \lnot{(\sigliedgt{\genDE{x}}{(-p)})} \lbisubjunct (\sigliedgeq{\genDE{x}}{(-p)}) \land (\sigliedgeq{\genDE{x}}{p}) \]
By rewriting with \rref{eq:sigliedgeqrearrangement}, the RHS of this equivalence is equivalent to the formula $\sigliedzero{\genDE{x}}{p}$ by real arithmetic. Negating both sides yields the provable equivalence \rref{eq:sigliedzerorearrangement}. \qedhere
\end{itemize}
\end{proof}

The equivalence \rref{eq:sigliedfullrearrangement} is particularly important, because it underlies the next proposition, from which all results about local progress will follow.

\begin{proposition} \label{prop:negationrearrangement}
Let $\rfvar$ be in normal form:
\[\rfvar \equiv \lorfold_{i=0}^{M} \Big(\landfold_{j=0}^{m(i)} p_{ij} \geq 0 \land \landfold_{j=0}^{n(i)} q_{ij}> 0\Big) \]

$\lnot{\rfvar}$ can be put in a normal form:
\[\lnot{\rfvar} \equiv \lorfold_{i=0}^{N} \Big(\landfold_{j=0}^{a(i)} r_{ij} \geq 0 \land \landfold_{j=0}^{b(i)} s_{ij}> 0\Big) \]
for which we additionally have the provable equivalence:
\begin{align*}
\lnot{(\sigliedsai{\genDE{x}}{\rfvar})} \lbisubjunct \sigliedsai{\genDE{x}}{(\lnot{\rfvar})}
\end{align*}
\end{proposition}
\begin{proof}
Throughout this proof, we will make use of the standard propositional tautologies:
\begin{align*}
\lnot{(A \land B)} &\lbisubjunct \lnot{A} \lor \lnot{B}\\
\lnot{(A \lor B)} &\lbisubjunct \lnot{A} \land \lnot{B}
\end{align*}

We start by negating $\rfvar$ (in normal form), and negating polynomials so that all inequalities have $0$ on the RHS. We write $\phi$ for the resulting RHS:
\[\lnot{\rfvar} \lbisubjunct \underbrace{\landfold_{i=0}^{M} \Big(\lorfold_{j=0}^{m(i)} -p_{ij} > 0 \lor \lorfold_{j=0}^{n(i)} -q_{ij} \geq 0 \Big)}_{\phi} \]
The progress formula $\sigliedsai{\genDE{x}}{\rfvar}$ for the normal form of $\rfvar$ is:
\[ \lorfold_{i=0}^{M} \Big(\landfold_{j=0}^{m(i)} \sigliedgeq{\genDE{x}}{p_{ij}} \land \landfold_{j=0}^{n(i)} \sigliedgt{\genDE{x}}{q_{ij}}\Big)\]

Negating both sides of this progress formula for $\rfvar$ proves:
\[\lnot{(\sigliedsai{\genDE{x}}{P})} \lbisubjunct
\landfold_{i=0}^{M} \Big(\lorfold_{j=0}^{m(i)} \lnot{(\sigliedgeq{\genDE{x}}{p_{ij}})} \lor \lorfold_{j=0}^{n(i)} \lnot{(\sigliedgt{\genDE{x}}{q_{ij}})}\Big)\]
Rewriting the RHS with equivalence \rref{eq:sigliedfullrearrangement} from \rref{prop:rearrangement} yields the following provable equivalence. We write $\psi$ for the resulting RHS.
\[\lnot{(\sigliedsai{\genDE{x}}{P})} \lbisubjunct \underbrace{\landfold_{i=0}^{M} \Big(\lorfold_{j=0}^{m(i)} \sigliedgt{\genDE{x}}{(-p_{ij})} \lor \lorfold_{j=0}^{n(i)} \sigliedgeq{\genDE{x}}{(-q_{ij})}\Big)}_\psi \]

Observe that $\phi,\psi$ have the same conjunctive normal form shape. We distribute the outer conjunction over the inner djsjunction in $\phi$ to obtain the following provable equivalence, whose RHS is a normal form for $\lnot{P}$ (for some indices $N,a(i),b(i)$ and polynomials $r_{ij},s_{ij}$):
\[\lnot{\rfvar} \lbisubjunct \lorfold_{i=0}^{N} \Big(\landfold_{j=0}^{a(i)} r_{ij} \geq 0 \lor \landfold_{j=0}^{b(i)} s_{ij} > 0\Big)\]
We distribute the disjunction in $\psi$ following the same syntactic steps taken in $\phi$ to obtain the following provable equivalence:
\[ \psi \lbisubjunct \lorfold_{i=0}^{N} \Big(\landfold_{j=0}^{a(i)} \sigliedgeq{\genDE{x}}{r_{ij}} \lor \landfold_{j=0}^{b(i)} \sigliedgt{\genDE{x}}{s_{ij}} \Big)\]
Rewriting with the equivalences derived so far, and using the above normal form for $\lnot{\rfvar}$, yields the required, provable equivalence:
\[
\lnot{(\sigliedsai{\genDE{x}}{\rfvar})} \lbisubjunct \sigliedsai{\genDE{x}}{(\lnot{\rfvar})}
\qedhere
\]
\end{proof}

\subsection{Local Progress}
\label{app:localprogress}

We first derive the properties about local progress stated in~\rref{subsec:localprogress}. These properties will be used in the completeness arguments for both algebraic and semialgebraic invariants.

\subsubsection{Atomic Non-strict Inequalities}

The \irref{Lpgeq} axiom derived in \rref{subsec:localprogress} has an implicit initial state assumption on the left of the implication for the $\ddnextsmall$ modality:
\[\dinferenceRule[Lpgeqxeqy|LPi$_\geq$]{}
{
\linferenceRule[impll]
  {\initassum \land p \geq 0 \land \big(p=0 \limply \dprogressinsmall{\D{x}=\genDE{x}}{\lied[]{\genDE{x}}{p} \geq 0}\big)}
  {\axkey{\dprogressinsmall{\D{x}=\genDE{x}}{p \geq 0}}}
}{}\]

We start by completing the proof of \rref{lem:localprogressgeq} that was outlined in \rref{subsec:localprogress}: we either iterate \irref{Lpgeqxeqy} until the first significant Lie derivative, or prove local progress in $p\geq 0$ using \irref{dRI} immediately.

\begin{proof}[Proof of \rref{lem:localprogressgeq}]
We derive the following axiom:
\[\dinferenceRule[Lpgeqfullxeqy|LP$_{\geq^*}$]{Progress Conditions}
{
\linferenceRule[impl]
  {\initassum \land \sigliedgeq{\genDE{x}}{p}}
  {\dprogressinsmall{\D{x}=\genDE{x}}{p \geq 0}}
}{}
\]

Let $N$ be the rank of $p$ with respect to $\D{x}=\genDE{x}$.
We unfold the definition of $\sigliedgeq{\genDE{x}}{p}$ and handle both cases separately.
{\footnotesize\renewcommand{\linferPremissSeparation}{~~~}%
\begin{sequentdeduction}[array]
\linfer[]{
  \linfer[orl]{
  \lsequent{\initassum,\sigliedgt{\genDE{x}}{p}}{\dprogressinsmall{\D{x}=\genDE{x}}{p \geq 0}} !
  \lsequent{\initassum,\sigliedzero{\genDE{x}}{p}}{\dprogressinsmall{\D{x}=\genDE{x}}{p \geq 0}}
  }
  {\lsequent{\initassum,\sigliedgt{\genDE{x}}{p} \lor \sigliedzero{\genDE{x}}{p}}{\dprogressinsmall{\D{x}=\genDE{x}}{p \geq 0}}}
}
  {\lsequent{\initassum,\sigliedgeq{\genDE{x}}{p}}{\dprogressinsmall{\D{x}=\genDE{x}}{p \geq 0}}}
\end{sequentdeduction}}%
The right premise by \irref{dDR}, because by \irref{dRI}, $p=0$ is invariant. The proof is completed with \irref{ContAx} using the trivial arithmetic fact $1 > 0$:
{\footnotesize
\begin{sequentdeduction}[array]
\linfer[dDR]{
  \linfer[dRI]{
    \lclose
  }
  {\lsequent{\sigliedzero{\genDE{x}}{p}}{\dbox{\pevolvein{\D{x}=\genDE{x}}{1 > 0}}{p = 0}}} !
  \linfer[qear+ContAx]{
    \lclose
  }
  {\lsequent{\initassum}{\dprogressinsmall{\D{x}=\genDE{x}}{1 \geq 0}}}
}
  {\lsequent{\initassum,\sigliedzero{\genDE{x}}{p}}{\dprogressinsmall{\D{x}=\genDE{x}}{p \geq 0}}}
\end{sequentdeduction}}%
The left premise also closes, because it gathers all of the open premises obtained by iterating \irref{Lpgeqxeqy} for higher Lie derivatives. In this way, the derivation continues until we are left with the final open premise which is abbreviated here, and continued below.
{\footnotesize\renewcommand{\linferPremissSeparation}{~~~}%
\begin{sequentdeduction}[array]
\linfer[Lpgeqxeqy]{
  \linfer[qear]{\lclose}{\lsequent{\sigliedgt{\genDE{x}}{p}}{p \geq 0}} !
  \linfer[Lpgeqxeqy]{
  \linfer[qear]{\lclose}{\lsequent{\sigliedgt{\genDE{x}}{p},p=0}{\lied[]{\genDE{x}}{p} \geq 0}} !
  \linfer[Lpgeqxeqy]{
    \lsequent{\initassum,\sigliedgt{\genDE{x}}{p},\dots}{\dots}
  }
  {\dots}
}
  {\lsequent{\initassum,\sigliedgt{\genDE{x}}{p},p=0}{\dprogressinsmall{\D{x}=\genDE{x}}{\lied[]{\genDE{x}}{p} \geq 0}}}
}
  {\lsequent{\initassum,\sigliedgt{\genDE{x}}{p}}{\dprogressinsmall{\D{x}=\genDE{x}}{p \geq 0}}}
\end{sequentdeduction}
}%
The open premise corresponds to the last conjunct of $\sigliedgt{\genDE{x}}{p}$. The implication in the conjunct is discharged with the gathered antecedents $p=0,\dots,\lied[N-2]{\genDE{x}}{p}= 0$.
{\footnotesize\renewcommand{\linferPremissSeparation}{~~~}%
\begin{sequentdeduction}[array]
\linfer[cut]{
  \linfer[gddR]{
  \linfer[ContAx]{
    \lclose
  }
    {\lsequent{\initassum,\lied[N-1]{\genDE{x}}{p} > 0}{\dprogressinsmall{\D{x}=\genDE{x}}{\lied[N-1]{\genDE{x}}{p} > 0}}}
  }
  {\lsequent{\initassum,\lied[N-1]{\genDE{x}}{p} > 0}{\dprogressinsmall{\D{x}=\genDE{x}}{\lied[N-1]{\genDE{x}}{p} \geq 0}}}
}
  \lsequent{\initassum,\sigliedgt{\genDE{x}}{p},p=0,\dots,\lied[N-2]{\genDE{x}}{p}= 0}{\dprogressinsmall{\D{x}=\genDE{x}}{\lied[N-1]{\genDE{x}}{p} \geq 0}}
\end{sequentdeduction}
}%
\end{proof}

\subsubsection{Atomic Strict Inequalities}

Next, we prove~\rref{lem:localprogressgt}. The essential idea is to reduce back to the non-strict case. We do so with the aid of the following proposition.
\begin{proposition}
\label{prop:leibnizpowers}
Let \(r = p^{k}\) for some $k \geq 1$, then \(\lied[i]{\genDE{x}}{r} \in \ideal{p}\) for all $0 \leq i \leq k-1$.
\end{proposition}
\begin{proof}
We proceed by induction on $k$.
\begin{itemize}
\item For $k=1$, we have $r=p^1$ so $\lied[0]{\genDE{x}}{r}=r \in \ideal{p}$ trivially.

\item For $r=p^{k+1}$, we obtain an expression for the $j$-th Lie derivative of $r$ by Leibniz's rule, where $0 \leq j \leq k$:
\[\lied[j]{\genDE{x}}{r} = \lie[j]{\genDE{x}}{p^{k}p} = \sum_{i=0}^j {j \choose i} \lied[j-i]{\genDE{x}}{(p^k)} \lied[i]{\genDE{x}}{p}\]
The induction hypothesis implies $\lied[j-i]{\genDE{x}}{(p^k)} \in \ideal{p}$ for $1 {\leq} i {\leq} j$, and thus, each summand ${j \choose i} \lied[j-i]{\genDE{x}}{(p^k)} \lied[i]{\genDE{x}}{p} \in \ideal{p}$ by \rref{def:ideal}.

The final summand for $i=0$ is:
\[{j \choose 0} \lied[j]{\genDE{x}}{(p^k)} \lied[0]{\genDE{x}}{p} = \lied[j]{\genDE{x}}{(p^k)} p \in \ideal{p}\]

Hence, $\lied[j]{\genDE{x}}{r} \in \ideal{p}$ as required. \qedhere
\end{itemize}
\end{proof}

For $r=p^{k},k \geq 1$, the formula \(p=0 \limply \landfold_{i=0}^{k-1} \lied[i]{\genDE{x}}{r} = 0\), thus, is provable in real arithmetic. This enables a proof of \rref{lem:localprogressgt}.

\begin{proof}[Proof of \rref{lem:localprogressgt}]
We derive the following axiom:
\[\dinferenceRule[Lpgtfullxeqy|LP$_{>^*}$]{Progress Conditions}
{
\linferenceRule[impl]
  {\initassum \land \sigliedgt{\genDE{x}}{p}}
  {\dprogressin{\D{x}=\genDE{x}}{p > 0}}
}{}
\]

Let $N$ be the rank of $p$ with respect to $\D{x}=\genDE{x}$. The rank bounds the number of higher Lie derivatives that we will need to consider.
Recall $N\geq1$ by~\rref{eq:differential-rank}.

We start by unfolding the syntactic abbreviation of the $\ddnext$ modality, and observe that we can reduce to the non-strict case with \irref{gddR} and the real arithmetic fact \(p{-}r\geq 0 \limply p > 0 \lor x=y\) for the abbreviation \(r\mdefeq |x-y|^{2N}\), which is a polynomial term: \(\big((x_1-y_1)^2 + \dots + (x_n-y_n)^2\big)^N\).

{\footnotesize
\begin{sequentdeduction}[array]
\linfer[]{
  \linfer[gddR]{
    \linfer[qear]{ \lclose }
    {\lsequent{p{-}r \geq 0}{p > 0 \lor x=y}}!
    \lsequent{\initassum,\sigliedgt{\genDE{x}}{p}}{\dprogressinsmall{\D{x}=\genDE{x}}{p{-}r \geq 0}}
  }
  {\lsequent{\initassum,\sigliedgt{\genDE{x}}{p}}{\dprogressinsmall{\D{x}=\genDE{x}}{p > 0 \lor x=y}}}
}
  {\lsequent{\initassum,\sigliedgt{\genDE{x}}{p}}{\dprogressin{\D{x}=\genDE{x}}{p > 0}}}
\end{sequentdeduction}
}%
Next, we make use of $\initassum$ in the antecedents. The first cut proves because $\initassum \limply |x-y|^2=0$ is a provable formula of real arithmetic. As remarked, with $|x-y|^2=0$ and $N \geq 1$, by \rref{prop:leibnizpowers}, $|x-y|^2=0 \limply \landfold_{i=0}^{N-1} \lied[i]{\genDE{x}}{r} = 0$ is a provable real arithmetic formula. The second cut proves using this fact. We call the resulting open premise \textcircled{1}.
{\footnotesize
\begin{sequentdeduction}[array]
\linfer[cut+qear]{
\linfer[cut+qear]{
  \lsequent{\initassum,\landfold_{i=0}^{N-1} \lied[i]{\genDE{x}}{r} = 0, \sigliedgt{\genDE{x}}{p}}{\dprogressinsmall{\D{x}=\genDE{x}}{p{-}r \geq 0}}
}
  {\lsequent{\initassum,|x-y|^2=0,\sigliedgt{\genDE{x}}{p}}{\dprogressinsmall{\D{x}=\genDE{x}}{p{-}r \geq 0}}}
}
  {\lsequent{\initassum,\sigliedgt{\genDE{x}}{p}}{\dprogressinsmall{\D{x}=\genDE{x}}{p{-}r \geq 0}}}
\end{sequentdeduction}
}%
To continue from \textcircled{1}, we observe for $0\leq i \leq N-1$,
\[\lie[i]{\genDE{x}}{p{-}r} = \lied[i]{\genDE{x}}{p} - \lied[i]{\genDE{x}}{r}\]

Using the conjunction $\landfold_{i=0}^{N-1} \lied[i]{\genDE{x}}{r} = 0$ in the antecedents, the formula $\lied[i]{\genDE{x}}{(p{-}r)} = \lied[i]{\genDE{x}}{p}$ proves by a cut and real arithmetic for $0\leq i \leq N-1$.
This justifies the next real arithmetic step from \textcircled{1}, where we abbreviate \(\Gamma_r \mdefequiv \landfold_{i=0}^{N-1} \lied[i]{\genDE{x}}{(p{-}r)} = \lied[i]{\genDE{x}}{p}\). Intuitively, $\Gamma_r$ will allow us to locally consider higher Lie derivatives of $p$ instead of higher Lie derivatives of $p-r$ in subsequent derivation steps.
{\footnotesize
\begin{sequentdeduction}[array]
\linfer[cut+qear]{
  \lsequent{\Gamma_r, \initassum, \sigliedgt{\genDE{x}}{p}}{\dprogressinsmall{\D{x}=\genDE{x}}{p{-}r \geq 0}}
}
  {\lsequent{\initassum,\landfold_{i=0}^{N-1} \lied[i]{\genDE{x}}{r} = 0, \sigliedgt{\genDE{x}}{p}}{\dprogressinsmall{\D{x}=\genDE{x}}{p{-}r \geq 0}}}
\end{sequentdeduction}
}%

It remains for us to use the same technique of iterating \irref{Lpgeqxeqy}, as shown in the proof of~\rref{lem:localprogressgeq}.
The following derivation starts with a single \irref{Lpgeqxeqy} step. The left premise closes by real arithmetic because $\sigliedgt{\genDE{x}}{p}$ has the conjunct $p\geq 0$, and $\Gamma_r$ provides $p-r = r$, which imply \(p-r\geq0\). The right premise is labeled \textcircled{2}:
{\footnotesize
\begin{sequentdeduction}[array]
\linfer[Lpgeqxeqy]{
  \linfer[qear]{
    \lclose
  }{\lsequent{\Gamma_r, \sigliedgt{\genDE{x}}{p}}{p{-}r \geq 0}} !
  \textcircled{2}
}
  {\lsequent{\Gamma_r, \initassum, \sigliedgt{\genDE{x}}{p}}{\dprogressinsmall{\D{x}=\genDE{x}}{p{-}r\geq 0}}}
\end{sequentdeduction}
}%
Continuing from \textcircled{2}, we now need to show local progress for the first Lie derivative of $p{-}r$. The first step simplifies formula $p{-}r=0$ in the antecedents using $\Gamma_r$. We use \irref{Lpgeqxeqy} again, use $\Gamma_r$ to simplify and prove the left premise, abbreviating the right premise with \textcircled{3}.
{\footnotesize
\begin{sequentdeduction}[array]
\linfer[qear]{
\linfer[Lpgeqxeqy]{
  \linfer[qear]{
  \linfer[qear]{
    \lclose
  }
    {\lsequent{p=0 \limply \lied[]{\genDE{x}}{p} \geq 0, p=0}{\lied[1]{\genDE{x}}{p}\geq 0}}
  }
  {\lsequent{\Gamma_r, \sigliedgt{\genDE{x}}{p}, p=0}{\lied[1]{\genDE{x}}{(p{-}r)}{\geq} 0}} !
  \textcircled{3}
}
  {\lsequent{\Gamma_r, \initassum, \sigliedgt{\genDE{x}}{p}, p=0}{\dprogressinsmall{\D{x}=\genDE{x}}{\lied[1]{\genDE{x}}{(p{-}r)}{\geq} 0}}}
}
  {\lsequent{\Gamma_r, \initassum, \sigliedgt{\genDE{x}}{p}, p{-}r=0}{\dprogressinsmall{\D{x}=\genDE{x}}{\lied[1]{\genDE{x}}{(p{-}r)}{\geq} 0}}}
\end{sequentdeduction}
}%
We continue similarly for the higher Lie derivatives from \textcircled{3}, using $\Gamma_r$ to replace $\lied[i]{\genDE{x}}{(p{-}r)}$ with $\lied[i]{\genDE{x}}{p}$, and then using the corresponding conjunct of $\sigliedgt{\genDE{x}}{p}$. The final open premise obtained from \textcircled{3} by iterating \irref{Lpgeqxeqy} corresponds to the last conjunct of $\sigliedgt{\genDE{x}}{p}$:
{\footnotesize\renewcommand{\sigliedgt}[3][]{\siglied[#1]{#2}{#3}{>}0}%
\begin{sequentdeduction}[array]
\linfer[qear]{
\linfer[Lpgeqxeqy]{
\linfer[Lpgeqxeqy]{
  \lsequent{\Gamma_r, \initassum,\sigliedgt{\genDE{x}}{p},p{=}0,\dots,\lied[N-2]{\genDE{x}}{p}{=}0}{\dprogressinsmall{\D{x}=\genDE{x}}{\lied[N-1]{\genDE{x}}{(p{-}r)} {\geq} 0}}
}
  {\dots}
}
  {\lsequent{\Gamma_r, \initassum, \sigliedgt{\genDE{x}}{p}, p=0,\lied[1]{\genDE{x}}{p}\geq 0}{\dprogressinsmall{\D{x}=\genDE{x}}{\lied[2]{\genDE{x}}{(p{-}r)}{\geq} 0}}}
}
  {\lsequent{\Gamma_r, \initassum, \sigliedgt{\genDE{x}}{p}, p{=}0,\lied[1]{\genDE{x}}{(p{-}r)}{\geq} 0}{\dprogressinsmall{\D{x}=\genDE{x}}{\lied[2]{\genDE{x}}{(p{-}r)}{\geq} 0}}}
\end{sequentdeduction}}%

The gathered antecedents $p=0,\dots,\lied[N-2]{\genDE{x}}{p}= 0$ are respectively obtained from $\Gamma_r$ by real arithmetic. The proof is closed with \irref{gddR+ContAx}, similarly to~\rref{lem:localprogressgeq}.
{\footnotesize
\begin{sequentdeduction}[array]
\linfer[cut]{
  \linfer[gddR]{
  \linfer[cut+qear]{
  \linfer[ContAx]{
    \lclose
  }
    {\lsequent{\initassum,\lied[N-1]{\genDE{x}}{(p{-}r)}{>} 0}{\dprogressinsmall{\D{x}=\genDE{x}}{\lied[N-1]{\genDE{x}}{(p{-}r)} {>} 0}}}
  }
    {\lsequent{\Gamma_r,\initassum,\lied[N-1]{\genDE{x}}{p} {>} 0}{\dprogressinsmall{\D{x}=\genDE{x}}{\lied[N-1]{\genDE{x}}{(p{-}r)} {>} 0}}}
  }
  {\lsequent{\Gamma_r,\initassum,\lied[N-1]{\genDE{x}}{p} {>} 0}{\dprogressinsmall{\D{x}=\genDE{x}}{\lied[N-1]{\genDE{x}}{(p{-}r)} {\geq} 0}}}
}
  {\lsequent{\Gamma_r, \initassum, \sigliedgt{\genDE{x}}{p},p{=}0,..,\lied[N-2]{\genDE{x}}{p}{=}0}{\dprogressinsmall{\D{x}=\genDE{x}}{\lied[N-1]{\genDE{x}}{(p{-}r)} {\geq} 0}}}
\end{sequentdeduction}}%
\end{proof}

\subsubsection{Semialgebraic Case}

We now prove the main lemma for local progress.

\begin{proof}[Proof of \rref{lem:localprogresssemialg}]
We derive the following axiom:
\[
\dinferenceRule[LpRfullxeqy|LP\usebox{\Rval}]{Progress Condition}
{\linferenceRule[impl]
  {\initassum \land \sigliedsai{\genDE{x}}{\rfvar}}
  {\dprogressin{\D{x}=\genDE{x}}{\rfvar}}
}{}
\]

We assume that $\rfvar$ is written in normal form~\rref{eq:normalform}. Throughout this proof, we will collapse similar premises in derivations and index them by $i,j$. We abbreviate the $i$-th disjunct of $\rfvar$ with $\rfvar_i \mdefequiv \landfold_{j=0}^{m(i)} p_{ij} \geq 0 \land \landfold_{j=0}^{n(i)} q_{ij}> 0$.

We start by splitting the outermost disjunction in $\sigliedsai{\genDE{x}}{\rfvar}$ with \irref{orl}. For each resulting premise (indexed by $i$), we select the corresponding disjunct of $\rfvar$ to prove local progress. The domain change with \irref{gddR} proves because $\rfvar_i \limply \rfvar$ is a propositional tautology for each $i$.
{\footnotesize
\begin{sequentdeduction}[array]
\linfer[orl]{
\linfer[gddR]{
  \lsequent{\initassum,\landfold_{j=0}^{m(i)} \sigliedgeq{\genDE{x}}{p_{ij}} \land \landfold_{j=0}^{n(i)} \sigliedgt{\genDE{x}}{q_{ij}}}{\dprogressin{\D{x}=\genDE{x}}{\rfvar_i}}
}
  {\lsequent{\initassum,\landfold_{j=0}^{m(i)} \sigliedgeq{\genDE{x}}{p_{ij}} \land \landfold_{j=0}^{n(i)} \sigliedgt{\genDE{x}}{q_{ij}}}{\dprogressin{\D{x}=\genDE{x}}{\rfvar}}}
}
  {\lsequent{\initassum,\sigliedsai{\genDE{x}}{\rfvar}}{\dprogressin{\D{x}=\genDE{x}}{\rfvar}}}
\end{sequentdeduction}
}%

Now, we only need to prove local progress in $\rfvar_i$. We make use of \irref{decompand} to split up the conjunct in $\rfvar_i$. This leaves premises (indexed by $j$) for the non-strict and strict inequalities of $\rfvar_i$ respectively. These premises are abbreviated with \textcircled{1} and \textcircled{2} respectively.

{\footnotesize
\begin{sequentdeduction}[default]
\linfer[decompand+andr]{
  \textcircled{1} \qquad\qquad \textcircled{2}
}
  {\lsequent{\initassum,\landfold_{j=0}^{m(i)} \sigliedgeq{\genDE{x}}{p_{ij}} \land \landfold_{j=0}^{n(i)} \sigliedgt{\genDE{x}}{q_{ij}}}{\dprogressin{\D{x}=\genDE{x}}{\rfvar_i}}}
\end{sequentdeduction}
}%
For the non-strict inequalities (\textcircled{1}), we use \irref{Lpgeqfullxeqy}, after unfolding $\ddnext$ and using \irref{gddR}, because $p_{ij} \geq 0 \limply p_{ij} \geq 0 \lor x=y$ is a propositional tautology:
{\footnotesize
\begin{sequentdeduction}[array]
  \linfer[]{
  \linfer[gddR]{
  \linfer[Lpgeqfullxeqy]{
    \lclose
  }
    {\lsequent{\initassum,\sigliedgeq{\genDE{x}}{p_{ij}}}{\dprogressinsmall{\D{x}=\genDE{x}}{p_{ij} \geq 0}}}
  }
    {\lsequent{\initassum,\sigliedgeq{\genDE{x}}{p_{ij}}}{\dprogressinsmall{\D{x}=\genDE{x}}{p_{ij} \geq 0 \lor x=y}}}
  }
  {\lsequent{\initassum,\sigliedgeq{\genDE{x}}{p_{ij}}}{\dprogressin{\D{x}=\genDE{x}}{p_{ij} \geq 0}}}
\end{sequentdeduction}
}%
For the strict inequalities (\textcircled{2}), we use \irref{Lpgtfullxeqy} directly:
{\footnotesize
\begin{sequentdeduction}[array]
  \linfer[Lpgtfullxeqy]{
    \lclose
  }
  {\lsequent{\initassum,\sigliedgt{\genDE{x}}{p_{ij}}}{\dprogressin{\D{x}=\genDE{x}}{p_{ij} > 0}}}
\end{sequentdeduction}
}%
Note that the $\ddnext$ modality is not required for non-strict inequalities, but they are crucially used for the strict inequalities.
\end{proof}

Finally, we give a characterization of semialgebraic local progress.
\begin{proof}[Proof of \rref{cor:localprogresscomplete}]
We derive the following axioms:

\begin{calculus}
\dinferenceRule[Lpiffxeqy|LP]{Iff Progress Condition}
{\linferenceRule[impl]
  {\initassum}
  {\big(\dprogressin{\D{x}=\genDE{x}}{\rfvar} \lbisubjunct \sigliedsai{\genDE{x}}{\rfvar}\big)}
}{}

\dinferenceRule[dualityxeqy|$\lnot{\ddnext}$]{Duality}
{\linferenceRule[impl]
  {\initassum}
  {\big(\dprogressin{\D{x}=\genDE{x}}{\rfvar} \lbisubjunct \lnot{\dprogressin{\D{x}=\genDE{x}}{\lnot{\rfvar}}}\big)}
}{}
\end{calculus}

We assume that $\rfvar$ is written in normal form~\rref{eq:normalform}.
By \rref{prop:negationrearrangement}, there is a normal form for $\lnot{P}$, \ie,
\[\lnot{\rfvar} \equiv \lorfold_{i=0}^{N} \big(\landfold_{j=0}^{a(i)} r_{ij} = 0 \land \landfold_{j=0}^{b(i)} s_{ij}> 0\big) \]
where we additionally have the provable equivalence:
\[\lnot{(\sigliedsai{\genDE{x}}{\rfvar})} \lbisubjunct \sigliedsai{\genDE{x}}{(\lnot{\rfvar})}\]

We first derive \irref{Lpiffxeqy}. The ``$\lylpmi$'' direction is \irref{LpRfullxeqy}. The proof for the ``$\limply$'' direction (of the inner equivalence) starts by reducing to the contrapositive statement by logical manipulation. We then use the above normal form to rewrite the negation in the antecedents. By \irref{LpRfullxeqy}, we cut in the local progress formula for $\lnot{\rfvar}$. We then move the negated succedent into the antecedents, and combine the two local progress antecedents with \irref{UniqAx}. This combines their respective domain constraints:
{\footnotesize
\begin{sequentdeduction}[array]
\linfer[cut+notl+notr]{
\linfer[qear]{
\linfer[LpRfullxeqy]{
\linfer[notr]{
\linfer[UniqAx]{
  \lsequent{\dprogressin{\D{x}=\genDE{x}}{\lnot{\rfvar} \land \rfvar}}{\lfalse}
}
  {\lsequent{\dprogressin{\D{x}=\genDE{x}}{\lnot{\rfvar}},\dprogressin{\D{x}=\genDE{x}}{\rfvar}}{\lfalse}}
}
  {\lsequent{\dprogressin{\D{x}=\genDE{x}}{\lnot{\rfvar}}}{\lnot{\dprogressin{\D{x}=\genDE{x}}{\rfvar}}}}
}
  {\lsequent{\initassum,\sigliedsai{\genDE{x}}{(\lnot{\rfvar})}}{\lnot{\dprogressin{\D{x}=\genDE{x}}{\rfvar}}}}
}
  {\lsequent{\initassum,\lnot{(\sigliedsai{\genDE{x}}{\rfvar})}}{\lnot{\dprogressin{\D{x}=\genDE{x}}{\rfvar}}}}
}
  {\lsequent{\initassum,\dprogressin{\D{x}=\genDE{x}}{\rfvar}}{\sigliedsai{\genDE{x}}{\rfvar}}}
\end{sequentdeduction}
}%
Observe that we now have $\lnot{\rfvar} \land \rfvar$ in the domain constraints which is equivalent to $\lfalse$; but we cannot locally progress into an empty set of states. The proof is completed by unfolding the $\ddnext$ syntactic abbreviation, and shifting to the box modality:
{\footnotesize
\begin{sequentdeduction}[array]
\linfer[]{
\linfer[diamond+notl]{
\linfer[dW]{
\linfer[qear]{
  \lclose
}
  {\lsequent{\lnot{\rfvar} \land \rfvar \lor x=y}{x=y}}
}
  {\lsequent{}{\dbox{\pevolvein{\D{x}=\genDE{x}}{\lnot{\rfvar} \land \rfvar \lor x=y}}{x=y}}}
}
  {\lsequent{\ddiamond{\pevolvein{\D{x}=\genDE{x}}{\lnot{\rfvar} \land \rfvar \lor x=y}}{x\not=y}}{\lfalse}}
}
  {\lsequent{\dprogressin{\D{x}=\genDE{x}}{\lnot{\rfvar} \land \rfvar}}{\lfalse}}
\end{sequentdeduction}
}%

The self-duality axiom \irref{dualityxeqy} derives from \irref{Lpiffxeqy} using the equivalence \(\lnot{(\sigliedsai{\genDE{x}}{\rfvar})} \lbisubjunct \sigliedsai{\genDE{x}}{(\lnot{\rfvar})}\) from \rref{prop:negationrearrangement}.
\end{proof}

\subsection{Algebraic Invariants}
\label{app:alginvariants}

This section proves the completeness results for algebraic invariants.
We first prove Liouville's formula which was used in \rref{lem:vdbx} to derive vectorial Darboux from vectorial \irref{DG}. We then complete the proof of \irref{dRI}. This allows us to prove the completeness result and its corollary.

\subsubsection{Liouville's Formula}

We give an arithmetic proof of Liouville's formula, which holds for any derivation operator $(\cdot)'$, i.e., any operator satisfying the usual sum and product rules of differentiation.
Derivation operators include Lie derivatives (which we will use \rref{lem:Liouville} for in the proof of \rref{lem:vdbx}), differentials, and time derivatives.

\begin{lemma}[Liouville] \label{lem:Liouville}
Let the $n\times n$ matrix of polynomials $A$ satisfy the equation \(A' = BA\) for an $n \times n$ matrix $B$ of cofactor polynomials.
Then the following is a provable real arithmetic identity:
\[ (\determinant{A})'  = \trace(B) \determinant{A} \]
\end{lemma}
\begin{proof}
We consider the following expression for the determinant $\determinant{A}$, where $\sigma \in \perm n$ is a permutation on $n$ indices and $\sgn(\sigma)$ denotes its sign:
\[\determinant(A) = \sum_{\sigma \in \perm n} (\sgn(\sigma) A_{1,\sigma_1}A_{2,\sigma_2}\dots A_{n,\sigma_n}) \]
By the sum and product rules of derivation operators:
\begin{align*}
(\determinant(A))' =& \sum_{\sigma \in \perm n} \sgn(\sigma) (A_{1,\sigma_1}A_{2,\sigma_2}\dots A_{n,\sigma_n})' \\
=& \sum_{\sigma \in \perm n} \sgn(\sigma) A'_{1,\sigma_1}A_{2,\sigma_2}\dots A_{n,\sigma_n} \\
+& \sum_{\sigma \in \perm n} \sgn(\sigma) A_{1,\sigma_1}A'_{2,\sigma_2}\dots A_{n,\sigma_n} \\
+& \cdots \\
+&\sum_{\sigma \in \perm n} \sgn(\sigma) A_{1,\sigma_1}A_{2,\sigma_2}\dots A'_{n,\sigma_n}
\end{align*}
Let us write:
{\footnotesize
\[A_{[i]'} = \left(\begin{array}{cccc}
a_{11} & a_{12} & \dots & a_{1n} \\
\vdots & \vdots & \ddots & \vdots \\
a'_{i1} & a'_{i2} & \dots & a'_{in} \\
\vdots & \vdots & \ddots & \vdots \\
a_{n1} & a_{n2} & \dots & a_{nn} \\
\end{array}\right)\]}%
Then
\[(\determinant(A))' = \sum_{i=1}^n \determinant(A_{[i]'})\]
Using $A'=BA$, and by row-reduction properties of determinants:
\begin{align*}
\determinant(A_{[i]'}) &= \determinant {\footnotesize \left(\begin{array}{cccc}
a_{11} & a_{12} & \dots & a_{1n} \\
\vdots & \vdots & \ddots & \vdots \\
\sum_{k=1}^n b_{ik}a_{k1} & \sum_{k=1}^n b_{ik}a_{k2} & \dots & \sum_{k=1}^n b_{ik}a_{kn} \\
\vdots & \vdots & \ddots & \vdots \\
a_{n1} & a_{n2} & \dots & a_{nn} \\
\end{array}\right)}  \\
&= \determinant {\footnotesize \left(\begin{array}{cccc}
a_{11} & a_{12} & \dots & a_{1n} \\
\vdots & \vdots & \ddots & \vdots \\
b_{ii}a_{i1} & b_{ii}a_{i2} & \dots & b_{ii}a_{in} \\
\vdots & \vdots & \ddots & \vdots \\
a_{n1} & a_{n2} & \dots & a_{nn} \\
\end{array}\right)} \\
&= b_{ii}\determinant{(A)}
\end{align*}
Therefore,
\[
(\determinant(A))' = \sum_{i=1}^n b_{ii}\determinant(A) = \trace{(B)}\determinant(A)
\qedhere\]
\end{proof}

\rref{lem:Liouville} proves the arithmetic fact \m{(\determinant(Y))' = - \trace{(\matpolyn{G}{x})}\determinant(Y)}
used in the proof of \rref{lem:vdbx} to derive rule \irref{vdbx}.
In that proof, $Y'=-YG$. Transposing yields $(Y^T)' = -G^TY^T$, so:
\[
(\determinant(Y))' = (\determinant(Y^T))'
= -\trace{(G^T)}\determinant(Y^T)\\
= -\trace{(G)}\determinant(Y)
\]

\subsubsection{Differential Radical Invariants}

Next, we complete the derivation of rule \irref{dRI} from derived rule \irref{vdbx}.
In this derivation, we use a \irref{DIeq} step which may be slightly unfamiliar since it differs from our usage in \irref{dI}. Rewriting axiom \irref{DIeq} with \irref{testb}, and abbreviating $\rrfvar\mdefequiv \dbox{\pevolvein{\D{x}=\usDE}{\usQ\argx}}{(\usp\argx)'=0}$ proves the following formula:
\[(\ivr \limply \rrfvar) \limply \big(\dbox{\pevolvein{\D{x}=\usDE}{\ivr}}{\usp=0} \lbisubjunct (\ivr \limply \usp=0)\big)\]
Propositionally, if we only consider of the ``$\lylpmi$'' direction of the nested equivalence we have the provable formula:
\[(\ivr \limply \rrfvar) \limply \big((\ivr \limply \usp=0) \limply \dbox{\pevolvein{\D{x}=\usDE}{\ivr}}{\usp=0}\big) \]
Thus, the formula \(\lnot{\ivr} \limply \dbox{\pevolvein{\D{x}=\usDE}{\ivr}}{\usp=0}\) proves propositionally from \irref{DIeq+testb}, since $\lnot{\ivr}$ implies both formulas on the left of the implication above.
Intuitively, the formula states that if the domain constraint $\ivr$ is false in an initial state, then the box modality in the conclusion is trivially true, because no trajectories stay in $\ivr$.

Therefore, as \(\ivr \limply \dbox{\pevolvein{\D{x}=\usDE}{\ivr}}{\usp=0}\) is propositionally equivalent to \(\lnot{\ivr} \lor \dbox{\pevolvein{\D{x}=\usDE}{\ivr}}{\usp=0}\) the following formula (which we use below) proves propositionally from \irref{DIeq+testb}:
\[(\ivr \limply \dbox{\pevolvein{\D{x}=\usDE}{\ivr}}{\usp=0}) \limply \dbox{\pevolvein{\D{x}=\usDE}{\ivr}}{\usp=0} \]

\begin{proof}[Proof of \rref{thm:DRI}]
\def\cmp{=}%
Let $p$ be a polynomial satisfying both premises of the \irref{dRI} proof rule, and let $\vecpolyn{p}{x}_i \mdefeq \lied[i-1]{\genDE{x}}{p}$ for $i = 1,2,\dots,N$, i.e.
\[\vecpolyn{p}{x} \mdefeq
\left(\begin{array}{l}p\\ \lied[1]{\genDE{x}}{p}\\ \vdots \\\lied[N-1]{\genDE{x}}{p}\end{array}\right) \]
The component-wise Lie derivative of $\vecpolyn{p}{x}$ is: $(\lied[]{\genDE{x}}{\vec{p}})_i = \lie[]{\genDE{x}}{\vec{p}_i} = \lied[i]{\genDE{x}}{p}$.

We start by setting up for a proof by \irref{vdbx}. In the first step, we used \irref{DIeq+testb} to assume $\ivr$ is true initially (see above). On the left premise after the cut, arithmetic equivalence \m{\landfold_{i=0}^{N-1} \lied[i]{\genDE{x}}{p} = 0 \lbisubjunct \vecpolyn{p}{x}\cmp0} is used to rewrite the succedent to the left premise of \irref{dRI}.
{\footnotesize
\begin{sequentdeduction}[array]
\linfer[DIeq+testb]{
\linfer[cut]{
  \linfer[qear]{
    \lsequent{\Gamma,\ivr} {\landfold_{i=0}^{N-1} \lied[i]{\genDE{x}}{p} = 0 }
  }
  {\lsequent{\Gamma,\ivr} {\vecpolyn{p}{x} \cmp0}} !
  \linfer[Mb]{
    \lsequent{\vecpolyn{p}{x} \cmp0} {\dbox{\pevolvein{\D{x}=\genDE{x}}{\ivr}}{\vecpolyn{p}{x} \cmp0}}
  }
  {\lsequent{\vecpolyn{p}{x} \cmp0} {\dbox{\pevolvein{\D{x}=\genDE{x}}{\ivr}}{p=0}}}
}
  {\lsequent{\Gamma,\ivr} {\dbox{\pevolvein{\D{x}=\genDE{x}}{\ivr}}{p\cmp0}}}
}
  {\lsequent{\Gamma} {\dbox{\pevolvein{\D{x}=\genDE{x}}{\ivr}}{p\cmp0}}}
\end{sequentdeduction}
}%

The right premise continues by \irref{vdbx} with the following choice of $\matpolyn{G}{x}$, with $1$ on its superdiagonal, and $g_i$ cofactors in the last row:
{\footnotesize
\begin{sequentdeduction}[default]
\linfer[vdbx]{
  {\lsequent{\ivr} {\lied[]{\genDE{x}}{\vec{p}}=\overbrace{
\left(\begin{array}{ccccc}
0      & 1      & 0      & \dots & 0      \\
0      & 0      & \ddots & \ddots & \vdots \\
\vdots & \vdots & \ddots & \ddots & 0      \\
0      & 0      & \dots & 0      & 1\\
g_0    & g_1    & \dots & g_{N-2}& g_{N-1} \end{array}\right)
}^{\matpolyn{G}{x}} \itimes\vecpolyn{p}{x}}}
}
    {\lsequent{\vecpolyn{p}{x} = 0} {\dbox{\pevolvein{\D{x}=\genDE{x}}{\ivr}}{\vecpolyn{p}{x}=0}}}
\end{sequentdeduction}
}%
The open premise requires us to prove a component-wise equality on two vectors, \ie, $(\lied[]{\genDE{x}}{\vec{p}})_i = (\matpolyn{G}{x}\itimes\vecpolyn{p}{x})_i$ for $1 \leq i \leq N$. For $i < N$, explicit matrix multiplication yields:
\[ (\lied[]{\genDE{x}}{\vec{p}})_i = \lied[i]{\genDE{x}}{p} = (\vecpolyn{p}{x})_{i+1} = (G\itimes\vecpolyn{p}{x})_i \]
Therefore, all but the final component-wise equality prove trivially by \irref{qear}. The remaining premise is:
\[\lsequent{\ivr}{(\lied[]{\genDE{x}}{\vec{p}})_{N} = (G\itimes\vecpolyn{p}{x})_N}\]
The LHS of this equality simplifies to:
\[ (\lied[]{\genDE{x}}{\vec{p}})_{N} = \lied[N]{\genDE{x}}{p}\]
The RHS simplifies to:
\[ (G\itimes\vecpolyn{p}{x})_N = \sum_{i=1}^{N} g_{i-1} (\vec{p})_i = \sum_{i=1}^{N} g_{i-1} \lied[i-1]{\genDE{x}}{p} = \sum_{i=0}^{N-1} g_{i} \lied[i]{\genDE{x}}{p}\]
Therefore, real arithmetic equivalently reduces the remaining open premise to the right premise of \irref{dRI}.
\end{proof}

\subsubsection{Completeness for Algebraic Invariants}

We now derive axiom \irref{DRI}, which makes use of \irref{LpRfullxeqy} from \rref{lem:localprogresssemialg} and \irref{ContOpen} from \rref{cor:ContOpen}.

\begin{proof}[Proof of \rref{thm:algcomplete}]
In the ``$\lylpmi$'' direction, we use \irref{dRI}, by setting $N$ to the rank of $p$, so that the succedent of its left premise is exactly $\sigliedzero{\genDE{x}}{p}$. The right premise closes by real arithmetic, since $N$ is the rank of $p$, it must, by definition satisfy the rank identity \rref{eq:differential-rank}.
{\footnotesize
\renewcommand{\arraystretch}{1.3}
\begin{sequentdeduction}[array]
\linfer[dRI]{
  \linfer[implyl]{\lclose}
  {\lsequent{\ivr \limply \sigliedzero{\genDE{x}}{p},\ivr}{\sigliedzero{\genDE{x}}{p}}}!
  \linfer[qear]{\lclose}
  {\lsequent{}{\lied[N]{\genDE{x}}{p} = \sum_{i=0}^{N-1} g_i \lied[i]{\genDE{x}}{p}}}
}
  {\lsequent{\ivr \limply \sigliedzero{\genDE{x}}{p}}{\dbox{\pevolvein{\D{x}=\genDE{x}}{\ivr}}{p=0}}}
\end{sequentdeduction}
}%

For the ``$\limply$'' direction, we first reduce to the contrapositive statement by logical manipulation. An application of \irref{diamond} axiom turns the negated box modality in the succedent to a diamond modality. By \rref{eq:sigliedzerorearrangement} from \rref{prop:rearrangement}, we equivalently rewrite the negated differential radical formula in the antecedents to two progress formulas. We cut the first-order formula $\lexists{y}{\initassum}$ which proves trivially in real arithmetic to get an initial state assumption. Finally, we use \irref{ContOpen} to cut local progress for $\ivr$ because, by assumption, $\ivr$ characterizes an open semialgebraic set. Splitting with \irref{orl} yields two premises, which we label \textcircled{1} and \textcircled{2}.
{\footnotesize
\begin{sequentdeduction}[array]
\linfer[implyr+diamond]{
\linfer[notl+notr]{
\linfer[qear]{
\linfer[cut]{
\linfer[existsl]{
\linfer[cut+ContOpen]{
\linfer[orl]{
  \textcircled{1} ! \textcircled{2}
}
  {\lsequent{\initassum,\dprogressinsmall{\D{x}=\genDE{x}}{\ivr}, \sigliedgt{\genDE{x}}{p} \lor \sigliedgt{\genDE{x}}{(-p)} }{\ddiamond{\pevolvein{\D{x}=\genDE{x}}{\ivr}}{p\neq 0}}}
}
  {\lsequent{\initassum,\ivr, \sigliedgt{\genDE{x}}{p} \lor \sigliedgt{\genDE{x}}{(-p)} }{\ddiamond{\pevolvein{\D{x}=\genDE{x}}{\ivr}}{p\neq 0}}}
}
  {\lsequent{\lexists{y}{\initassum},\ivr, \sigliedgt{\genDE{x}}{p} \lor \sigliedgt{\genDE{x}}{(-p)}}{\ddiamond{\pevolvein{\D{x}=\genDE{x}}{\ivr}}{p\neq 0}}}
}
  {\lsequent{\ivr, \sigliedgt{\genDE{x}}{p} \lor \sigliedgt{\genDE{x}}{(-p)}}{\ddiamond{\pevolvein{\D{x}=\genDE{x}}{\ivr}}{p\neq 0}}}
}
  {\lsequent{\ivr, \lnot{(\sigliedzero{\genDE{x}}{p})}}{\ddiamond{\pevolvein{\D{x}=\genDE{x}}{\ivr}}{p\neq 0}}}
}
  {\lsequent{\ivr, \lnot\ddiamond{\pevolvein{\D{x}=\genDE{x}}{\ivr}}{\lnot p=0}}{\sigliedzero{\genDE{x}}{p}}}
}
  {\lsequent{\dbox{\pevolvein{\D{x}=\genDE{x}}{\ivr}}{p=0}}{\ivr \limply \sigliedzero{\genDE{x}}{p}}}
\end{sequentdeduction}
}%

Continuing on \textcircled{1}, because we already have $\sigliedgt{\genDE{x}}{p}$ in the antecedents, \irref{Lpgtfullxeqy} derives local progress for $p>0$: $\dprogressin{\D{x}=\genDE{x}}{p>0}$. Unfolding the $\ddnext$ abbreviation, an application of \irref{decompand} allows us to combine the two local progress formulas in the antecedent.
{\footnotesize
\begin{sequentdeduction}[array]
\linfer[cut+Lpgtfullxeqy]{
\linfer[decompand]{
  \lsequent{\dprogressinsmall{\D{x}=\genDE{x}}{\ivr \land (p>0 \lor x=y)} }{ \ddiamond{\pevolvein{\D{x}=\genDE{x}}{\ivr}}{p\neq0}}
}
  {\lsequent{\dprogressinsmall{\D{x}=\genDE{x}}{\ivr},\dprogressin{\D{x}=\genDE{x}}{p > 0 } }{ \ddiamond{\pevolvein{\D{x}=\genDE{x}}{\ivr}}{p\neq0}}}
}
  {\lsequent{\initassum,\dprogressinsmall{\D{x}=\genDE{x}}{\ivr},\sigliedgt{\genDE{x}}{p}}{ \ddiamond{\pevolvein{\D{x}=\genDE{x}}{\ivr}}{p\neq0}}}
\end{sequentdeduction}
}%
We continue by removing the syntactic abbreviation for $\ddnextsmall$. Since $\overbrace{\ivr \land (p > 0 \lor x=y)}^{\rrfvar} \limply \ivr$ is a propositional tautology, we use \irref{gddR} to strengthen the evolution domain constraint in the succedent. This allows us to use Kripke axiom \irref{Kd} which reduces our succedent to the box modality. We finish the proof with a \irref{dW} step, because the formula $p>0 \lor x=y$ in the domain constraint $\rrfvar$ implies the succedent by real arithmetic.
{\footnotesize
\begin{sequentdeduction}[array]
\linfer[]{
\linfer[gddR]{
\linfer[Kd]{
\linfer[dW]{
\linfer[qear]{ \lclose }
  {\lsequent{\rrfvar}{(x\neq y \limply p\neq0)}}
}
  {\lsequent{}{\dbox{\pevolvein{\D{x}=\genDE{x}}{\rrfvar}}{(x\neq y \limply p\neq0)}}}
}
  {\lsequent{\ddiamond{\pevolvein{\D{x}=\genDE{x}}{\rrfvar}}{x\neq y} }{\ddiamond{\pevolvein{\D{x}=\genDE{x}}{\rrfvar}}{p\neq0}}}
}
  {\lsequent{\ddiamond{\pevolvein{\D{x}=\genDE{x}}{\rrfvar}}{x\neq y} }{\ddiamond{\pevolvein{\D{x}=\genDE{x}}{\ivr}}{p\neq0}}}
}
  {\lsequent{\initassum,\dprogressinsmall{\D{x}=\genDE{x}}{\rrfvar} }{ \ddiamond{\pevolvein{\D{x}=\genDE{x}}{\ivr}}{p\neq0}}}
\end{sequentdeduction}}%
The remaining premise \textcircled{2} follows similarly, except that the progress formula $\sigliedgt{\genDE{x}}{(-p)}$ enables the cut $\dprogressin{\D{x}=\genDE{x}}{{-}p > 0}$. It leads to the same conclusion of $p\neq 0$ in the postcondition.
\end{proof}

We now prove \rref{cor:testfree} using the characterization of algebraic invariants of ODEs from~\rref{thm:algcomplete}. We shall consider the fragment of \dL programs generated by the following grammar, where $\coalgvar$ denotes a first-order formula of real arithmetic that characterizes the \emph{complement} of a real algebraic variety, or equivalently, a formula of the form $r \neq 0$ where $r$ is a polynomial:
\[ \alpha,\beta \bebecomes \pumod{x}{e}  \alternative \ptest{\coalgvar} \alternative \pevolvein{\D{x}=\genDE{x}}{\coalgvar}  \alternative \pchoice{\alpha}{\beta} \alternative \alpha;\beta \alternative \prepeat{\alpha}\]

\begin{proof}[Proof of \rref{cor:testfree}]
Firstly, since $P$ is algebraic, it is equivalent to a formula $p=0$ for some polynomial $p$ so we may, by real arithmetic, assume that it is written in this form. Similarly, we shall, by real arithmetic, assume that $\coalgvar$ already has the form $r\neq 0$.

We proceed by structural induction on the form of $\alpha$ following \cite[Theorem 1]{DBLP:conf/lics/Platzer12b}, and show that for some (computable) polynomial $q$, we can derive the equivalence \(\dbox{\alpha}{p=0} \lbisubjunct q=0\) in \dL.
\begin{itemize}
\item Case $ \pevolvein{\D{x}=\genDE{x}}{r \neq 0}$. The set of states characterized by $r \neq 0$ is open. Thus, \rref{thm:algcomplete} derives the equivalence \(\dbox{\pevolvein{\D{x}=\genDE{x}}{r \neq 0}}{p=0} \lbisubjunct (r \neq 0 \limply \sigliedzero{\genDE{x}}{p})\).
Let $N$ be the rank of $p$ so that $\sigliedzero{\genDE{x}}{p}$ expands to \(\landfold_{i=0}^{N-1}  \lied[i]{\genDE{x}}{p} = 0\). Let \(q \mdefeq r(\sum_{i=0}^{N-1} (\lied[i]{\genDE{x}}{p})^2)\), giving the provable real arithmetic equivalence \((r \neq 0 \limply \sigliedzero{\genDE{x}}{p}) \lbisubjunct q = 0\). Rewriting with this derives the equivalence,
\[\dbox{\pevolvein{\D{x}=\genDE{x}}{r \neq 0}}{p=0} \lbisubjunct q=0\]

\item Case $\pumod{x}{e}$. By axiom \irref{assignb}, $\dbox{\pumod{x}{e}}{p(x)=0} \lbisubjunct p(e)=0$. As a composition of polynomials, $p(e)$ is a polynomial.

\item Case $\ptest{r \neq 0}$. By axiom \irref{testb}, $\dbox{\ptest{r\neq 0}}{p=0} \lbisubjunct (r \neq 0 \limply p = 0)$. Let $q \mdefeq rp$, giving the provable real arithmetic equivalence $(r \neq 0 \limply p = 0)\lbisubjunct q=0$. Rewriting with this derived equivalence yields the derived equivalence:
\[\dbox{\ptest{r\neq 0}}{p=0} \lbisubjunct q=0\]

\item Case $\pchoice{\alpha}{\beta}$. By \irref{choiceb}, $\dbox{\pchoice{\alpha}{\beta}}{p=0} \lbisubjunct \dbox{\alpha}{p=0} \land \dbox{\beta}{p=0}$. By the induction hypothesis on $\alpha,\beta$, we may derive $\dbox{\alpha}{p=0} \lbisubjunct q_1=0$ and $\dbox{\beta}{p=0} \lbisubjunct q_2 =0$ for some polynomials $q_1,q_2$. Moreover, $q_1=0 \land q_2=0 \lbisubjunct q_1^2+q_2^2=0$ is a provable formula of real arithmetic. Rewriting with the derived equivalences yields the derived equivalence:
\[\dbox{\pchoice{\alpha}{\beta}}{p=0} \lbisubjunct q_1^2+q_2^2=0\]

\item Case $\alpha;\beta$. By \irref{composeb}, $\dbox{\alpha;\beta}{p=0} \lbisubjunct \dbox{\alpha}{\dbox{\beta}{p=0}}$. By the induction hypothesis on $\beta$, we derive $\dbox{\beta}{p=0} \lbisubjunct q_2=0$. By rewriting with this equivalence, we derive $\dbox{\alpha;\beta}{p=0} \lbisubjunct \dbox{\alpha}{q_2=0}$. Now, by the induction hypothesis on $\alpha$, we derive $\dbox{\alpha}{q_2=0} \lbisubjunct q_1=0$ for some $q_1$. Rewriting with the derived equivalences yields the derived equivalence:
\[\dbox{\alpha;\beta}{p=0} \lbisubjunct q_1=0\]

\item Case $\prepeat{\alpha}$. This case relies on the fact that the polynomial ring $\polynomials{\reals}{x}$ (and $\polynomials{\rationals}{x}$) over a finite number of indeterminates $x$ is a Noetherian domain, i.e., every ascending chain of ideals is finite. We first construct the following sequence of polynomials $q_i$:
\[q_0 \mdefeq p,\quad q_{i+1} \mdefeq f_i\]
where $f_i$ is the polynomial satisfying the derived equivalence \(f_i \lbisubjunct \dbox{\alpha}{q_i=0}\) obtained by applying the induction hypothesis on $\alpha$ with postcondition $q_i=0$.
Since the ring of polynomials over the (finite set) of variables mentioned in $\alpha$ or $p$ is Noetherian, the following chain of ideals is finite:
\[ \ideal{q_0} \subset \ideal{q_0,q_1} \subset \ideal{q_0,q_1,q_2} \subset \dots \]
Thus, there is some smallest $k$ such that $q_k$ satisfies the following polynomial identity, with polynomial cofactors $g_i$:
\begin{equation}
  q_k = \sum_{i=0}^{k-1} g_i q_i
  \label{eq:HPNoetherianalgebraicloop}
\end{equation}
We claim that $\landfold_{i=0}^{k-1} q_i=0 \lbisubjunct \dbox{\prepeat{\alpha}}{p=0}$ is derivable. Since the following real arithmetic equivalence is provable \(\sum_{i=0}^{k-1} q_i^2 = 0 \lbisubjunct \landfold_{i=0}^{k-1} q_i=0\), this claim yields the derived equivalence \(\dbox{\prepeat{\alpha}}{p=0} \lbisubjunct \sum_{i=0}^{k-1} q_i^2 = 0\), as required. We show both directions of the claim separately.
\begin{itemize}
\item[``$\lylpmi$''] This direction is straightforward using $k$ times the iteration axiom \irref{iterateb} together with \irref{band}. By construction, we may successively replace $\dbox{\alpha}{q_i}$ with $q_{i+1}$, which gives us the required implication.
{\footnotesize
\renewcommand*{\arraystretch}{1.3}
\begin{sequentdeduction}[array]
\linfer[iterateb]{
\linfer[iterateb+band]{
\linfer[iterateb+band]{
\linfer[iterateb+band]{
\linfer[]{
\linfer[]{
  \lclose
}
  {\lsequent{q_0=0 \land q_1=0 \land q_2=0 \land \cdots \land q_{k-1} = 0}{\landfold_{i=0}^{k-1} q_i=0}}
}
  {\lsequent{p=0 \land \dbox{\alpha}{p=0} \land \dbox{\alpha}{\dbox{\alpha}{p=0}} \land \cdots}{\landfold_{i=0}^{k-1} q_i=0}}
}
  {\cdots}
}
  {\lsequent{p=0 \land \dbox{\alpha}{p=0} \land \dbox{\alpha}{\dbox{\alpha}{\dbox{\prepeat{\alpha}}{p=0}}}}{\landfold_{i=0}^{k-1} q_i=0}}
}
  {\lsequent{p=0 \land \dbox{\alpha}{\dbox{\prepeat{\alpha}}{p=0}}}{\landfold_{i=0}^{k-1} q_i=0}}
}
  {\lsequent{\dbox{\prepeat{\alpha}}{p=0}}{\landfold_{i=0}^{k-1} q_i=0}}
\end{sequentdeduction}
}%

\item[``$\limply$''] We strengthen the postcondition of the box modality to $\landfold_{i=0}^{k-1} q_i=0$ (recall that $q_0 \mdefeq p$, so $\landfold_{i=0}^{k-1} q_i=0 \limply p=0$ is a propositional tautology), and prove it as a loop invariant. By an application of \irref{band} followed by \irref{andr}, we may split the postcondition to its constituent conjuncts (indexed by $0 \leq i \leq k-1$). By construction, we can equivalently replace each $\dbox{\alpha}{q_i}$ with $q_{i+1}$.
{\footnotesize
\renewcommand*{\arraystretch}{1.3}
\begin{sequentdeduction}[array]
\linfer[Mb]{
\linfer[loop]{
\linfer[band]{
\linfer{
\linfer[qear]{
  \lclose
}
  {\lsequent{\landfold_{i=0}^{k-1} q_i=0}{\landfold_{i=0}^{k-1} q_{i+1}=0}}
}
  {\lsequent{\landfold_{i=0}^{k-1} q_i=0}{\landfold_{i=0}^{k-1} \dbox{\alpha}{q_i=0}}}
}
  {\lsequent{\landfold_{i=0}^{k-1} q_i=0}{\dbox{\alpha}{\landfold_{i=0}^{k-1} q_i=0}}}
}
  {\lsequent{\landfold_{i=0}^{k-1} q_i=0}{\dbox{\prepeat{\alpha}}{\landfold_{i=0}^{k-1} q_i=0}}}
}
  {\lsequent{\landfold_{i=0}^{k-1} q_i=0}{\dbox{\prepeat{\alpha}}{p=0}}}
\end{sequentdeduction}
}%

The conjuncts for $0 \leq i < k-1$ close trivially because $q_{i+1}=0$ is in the antecedent. The last conjunct for $i=k-1$ is $q_k=0$, which follows from the antecedent using \rref{eq:HPNoetherianalgebraicloop}.
\qedhere
\end{itemize}
\end{itemize}
\end{proof}

\subsection{Completeness for Semialgebraic Invariants with Semialgebraic Evolution Domain Constraints}
\label{app:semialginvariants}

The following generalized version of \irref{sAI} for \rref{thm:sAI}, which also handles the evolution domain constraints derives, from \irref{realindin+Lpiffxeqy}.

\begin{theorem}[Semialgebraic invariants for \rref{thm:sAI} with semialgebraic domain constraints]
For semialgebraic $\ivr,\rfvar$, with progress formulas $\sigliedsai{\genDE{x}}{\ivr},\sigliedsai{\genDE{x}}{\rfvar},\sigliedsai[-]{\genDE{x}}{\ivr}, \sigliedsai[-]{\genDE{x}}{(\lnot{\rfvar})}$ w.r.t.\ their respective normal forms \rref{eq:normalform}, this rule derives from the \dL calculus with \irref{RealIndInAx+DadjointAx+ContAx+UniqAx}.
\[
\dinferenceRule[sAIQ|sAI{$\&$}]{Semialgebraic invariant with domains}
{\linferenceRule
  {
  \lsequent{\rfvar,\ivr,\sigliedsai{\genDE{x}}{\ivr}}{\sigliedsai{\genDE{x}}{\rfvar}} &
  \lsequent{\lnot{\rfvar},\ivr,\sigliedsai[-]{-\genDE{x}}{\ivr}}{\sigliedsai[-]{-\genDE{x}}{(\lnot{\rfvar})}}
  }
  {\lsequent{\rfvar}{\dbox{\pevolvein{\D{x}=\genDE{x}}{\ivr}}{\rfvar}}}
}{}\]
\end{theorem}
\begin{proof}
Rule \irref{sAIQ} derives directly from rule \irref{realindin} derived in \rref{cor:realindin}, \rref{cor:bigsmallequiv}, and the characterization of semialgebraic local progress \irref{Lpiffxeqy} from \rref{cor:localprogresscomplete}. The $\initassum$ assumptions provided by \irref{realindin} are used to convert between the local progress modalities and the semialgebraic progress formulas by \irref{Lpiffxeqy}, but they do not need to remain in the premises of \irref{sAIQ} (by weakening).
\end{proof}

Recalling our earlier discussion for rule \irref{realindin}, we may again use \irref{V} to keep any additional context assumptions that do not depend on variables $x$ for the ODEs $\D{x}=\genDE{x}$ in rule \irref{sAIQ}, because it immediately derives from \irref{realindin}, which supports constant contexts.

We now prove a syntactic completeness theorem for \irref{sAIQ}, from which \rref{thm:semialgcompleteness} follows as a special case (where $\ivr \equiv \ltrue$).

\begin{theorem}[Semialgebraic invariant completeness for \rref{thm:semialgcompleteness} with semialgebraic domains]
\label{thm:semialgcompletenessdom}
For semialgebraic $\ivr,\rfvar$, with progress formulas $\sigliedsai{\genDE{x}}{\ivr},\sigliedsai{\genDE{x}}{\rfvar},\sigliedsai[-]{\genDE{x}}{\ivr}, \sigliedsai[-]{\genDE{x}}{(\lnot{\rfvar})}$ w.r.t.\ their respective normal forms \rref{eq:normalform}, this axiom derives in \dL with $\irref{RealIndInAx+DadjointAx+ContAx+UniqAx}$.
\[\dinferenceRule[semialgiffQ|SAI{$\&$}]{Semialgebraic invariant with domains axiom}
{\linferenceRule[equivl]
  {\lforall{x}{\big(\rfvar {\land} \ivr {\land} \sigliedsai{\genDE{x}}{\ivr} {\limply} \sigliedsai{\genDE{x}}{\rfvar}\big)} {\land}\lforall{x}{\big(\lnot{\rfvar} {\land} \ivr {\land} \sigliedsai[-]{-\genDE{x}}{\ivr} {\limply} \sigliedsai[-]{\-genDE{x}}{(\lnot{\rfvar})}\big)}}
  {\axkey{\lforall{x}{(\rfvar \limply \dbox{\pevolvein{\D{x}=\genDE{x}}{\ivr}}{\rfvar})}}}
}{}\]

In particular, the \dL calculus is complete for invariance properties of the following form with semialgebraic formulas $\rfvar,\ivr$:
\[\lsequent{\rfvar} {\dbox{\pevolvein{\D{x}=\genDE{x}}{\ivr}}{\rfvar}}\]
\end{theorem}
\begin{proof}
We abbreviate the left and right conjunct on the RHS of \irref{semialgiffQ} as \textcircled{a}  and \textcircled{b}, respectively.

The ``$\lylpmi$'' direction derives directly by an application of \irref{sAIQ}. The antecedents \textcircled{a} and \textcircled{b} are first-order formulas of real arithmetic which are quantified over $x$, the variables evolved by the ODE $\D{x}=\genDE{x}$. They may, therefore, be kept as constant context in the antecedents of the premises when applying rule \irref{sAIQ}.

{\footnotesize
\begin{sequentdeduction}[array]
\linfer[allr]{
\linfer[implyr]{
\linfer[sAIQ]{
  \linfer[alll+implyl]{
  \lclose
  }
  {\lsequent{\textcircled{a},\rfvar,\ivr,\sigliedsai{\genDE{x}}{\ivr}}{\sigliedsai{\genDE{x}}{\rfvar}}} !
  \linfer[alll+implyl]{
  \lclose
  }
  {\lsequent{\textcircled{b},\lnot{\rfvar},\ivr,\sigliedsai[-]{-\genDE{x}}{\ivr}}{\sigliedsai[-]{-\genDE{x}}{(\lnot{\rfvar})}}}
}
  {\lsequent{\textcircled{a}, \textcircled{b},\rfvar}{\dbox{\pevolvein{\D{x}=\genDE{x}}{\ivr}}{\rfvar}}}
}
  {\lsequent{\textcircled{a}, \textcircled{b}}{(\rfvar {\limply} \dbox{\pevolvein{\D{x}=\genDE{x}}{\ivr}}{\rfvar})}}
}
  {\lsequent{\textcircled{a}, \textcircled{b}}{\lforall{x}{(\rfvar {\limply} \dbox{\pevolvein{\D{x}=\genDE{x}}{\ivr}}{\rfvar})}}}
\end{sequentdeduction}
}%

In the ``$\limply$'' direction, we show the contrapositive statement in both cases. For \textcircled{b}, we use the derived invariant reflection axiom (\irref{reflect}) to turn the invariance assumption for the forwards ODE to an invariance assumption for the backwards ODE. The open premises with \textcircled{a} and \textcircled{b} in the succedent are labeled \textcircled{1} and \textcircled{2} respectively.
{\footnotesize\renewcommand{\linferPremissSeparation}{~~~}%
\begin{sequentdeduction}[array]
\linfer[andr]{
  \lsequent{\lforall{x}{(\rfvar {\limply} \dbox{\pevolvein{\D{x}=\genDE{x}}{\ivr}}{\rfvar})}}{\textcircled{a}} !
  \linfer[reflect]{
  \lsequent{\lforall{x}{(\lnot{\rfvar} {\limply} \dbox{\pevolvein{\D{x}=-\genDE{x}}{\ivr}}{\lnot{\rfvar}})}}{\textcircled{b}}
  }
  {\lsequent{\lforall{x}{(\rfvar {\limply} \dbox{\pevolvein{\D{x}=\genDE{x}}{\ivr}}{\rfvar})}}{\textcircled{b}}}
}
  {\lsequent{\lforall{x}{(\rfvar {\limply} \dbox{\pevolvein{\D{x}=\genDE{x}}{\ivr}}{\rfvar})}}{\textcircled{a} \land \textcircled{b}}}
\end{sequentdeduction}
}%

Continuing on \textcircled{1}, we expand \textcircled{a} and dualize on both sides of the sequent. We choose $x$ as our witness for the (then) existentially quantified succedent. The \irref{qear} two steps respectively introduce an initial state assumption, and replace $\lnot{(\sigliedsai{\genDE{x}}{\rfvar})}$ with $\sigliedsai{\genDE{x}}{(\lnot{\rfvar})}$, by \rref{prop:negationrearrangement}. We continue by cutting in local progress for $\ivr$ and $\lnot{P}$ respectively in steps \textcircled{3}, \textcircled{4}. The open premise after both cuts is labeled \textcircled{5}. All three steps are continued below.
{\footnotesize
\begin{sequentdeduction}[array]
\linfer[diamond+notl+notr]{
\linfer[existsl]{
\linfer[existsr+andr]{
\linfer[cut]{
\linfer[existsl]{
\linfer[qear]{
\linfer[cut]{
  \textcircled{3} !
  \textcircled{4} !
  \textcircled{5}
}
  {\lsequent{\initassum,\rfvar, \ivr, \sigliedsai{\genDE{x}}{\ivr}, \sigliedsai{\genDE{x}}{(\lnot{\rfvar})}}{\ddiamond{\pevolvein{\D{x}=\genDE{x}}{\ivr}}{\lnot{\rfvar}}}}
}
  {\lsequent{\initassum,\rfvar, \ivr, \sigliedsai{\genDE{x}}{\ivr}, \lnot{(\sigliedsai{\genDE{x}}{\rfvar})}}{\ddiamond{\pevolvein{\D{x}=\genDE{x}}{\ivr}}{\lnot{\rfvar}}}}
}
  {\lsequent{\lexists{y}{\initassum},\rfvar, \ivr, \sigliedsai{\genDE{x}}{\ivr}, \lnot{(\sigliedsai{\genDE{x}}{\rfvar})}}{\ddiamond{\pevolvein{\D{x}=\genDE{x}}{\ivr}}{\lnot{\rfvar}}}}
}
  {\lsequent{\rfvar, \ivr, \sigliedsai{\genDE{x}}{\ivr}, \lnot{(\sigliedsai{\genDE{x}}{\rfvar})}}{\ddiamond{\pevolvein{\D{x}=\genDE{x}}{\ivr}}{\lnot{\rfvar}}}}
}
  {\lsequent{\rfvar, \ivr, \sigliedsai{\genDE{x}}{\ivr}, \lnot{(\sigliedsai{\genDE{x}}{\rfvar})}}{\lexists{x}{(\rfvar \land \ddiamond{\pevolvein{\D{x}=\genDE{x}}{\ivr}}{\lnot{\rfvar}})}}}
}
  {\lsequent{\lexists{x}{\big(\rfvar \land \ivr \land \sigliedsai{\genDE{x}}{\ivr} \land \lnot{(\sigliedsai{\genDE{x}}{\rfvar})}\big)}}{\lexists{x}{(\rfvar \land \ddiamond{\pevolvein{\D{x}=\genDE{x}}{\ivr}}{\lnot{\rfvar}})}}}
}
  {\lsequent{\lforall{x}{(\rfvar {\limply} \dbox{\pevolvein{\D{x}=\genDE{x}}{\ivr}}{\rfvar})}}{\lforall{x}{\big(\rfvar \land \ivr \land \sigliedsai{\genDE{x}}{\ivr} {\limply} \sigliedsai{\genDE{x}}{\rfvar}\big)}}}
\end{sequentdeduction}
}%
The proofs in \textcircled{3}, \textcircled{4} are similar and they prove local progress conditions within $\ivr$ and $\lnot\rfvar$ from $\sigliedsai{\genDE{x}}{\ivr}, \sigliedsai{\genDE{x}}{(\lnot{\rfvar})}$ respectively. On \textcircled{3}, the first step uses \irref{bigsmallequiv} to weaken to the $\ddnext$ modality since we have $\initassum$ and $\rfvar$ in the antecedents, and $y$ is constant along the ODE $\D{x}=\genDE{x}$. The proof for \textcircled{3} completes with \irref{LpRfullxeqy}.
{\footnotesize
\begin{sequentdeduction}[array]
\linfer[bigsmallequiv]{
\linfer[LpRfullxeqy]{
  \lclose
}
  {\lsequent{\initassum,\ivr, \sigliedsai{\genDE{x}}{\ivr}}{\dprogressin{\D{x}=\genDE{x}}{\ivr}}}
}
  {\lsequent{\initassum,\ivr, \sigliedsai{\genDE{x}}{\ivr}}{\dprogressinsmall{\D{x}=\genDE{x}}{\ivr}}}
\end{sequentdeduction}
}%
The proof for \textcircled{4} uses~\irref{LpRfullxeqy} directly.
{\footnotesize
\begin{sequentdeduction}[array]
\linfer[LpRfullxeqy]{
  \lclose
}
  {\lsequent{\initassum, \sigliedsai{\genDE{x}}{(\lnot{\rfvar})}}{\dprogressin{\D{x}=\genDE{x}}{\lnot{\rfvar}}}}
\end{sequentdeduction}
}%

The proof completes on \textcircled{5} by unfolding the $\ddnextsmall$ and $\ddnext$ modalities and combining the two local progress modalities in the antecedents using \irref{decompand}.
A \irref{Kd+dW} step allows us to turn the postcondition of the diamond modality in the antecedent to $\lnot{\rfvar}$ because at the endpoint where $x \neq y$ is satisfied the domain constraint must still be true which implies that $\lnot{\rfvar}$ is true at that endpoint. Formally, the \irref{dW} step proves with the tautology $\ivr \land (\lnot{\rfvar} \lor \initassum) \limply x \neq y \limply \lnot{\rfvar}$. A \irref{gddR} step completes the proof.
{\footnotesize
\begin{sequentdeduction}[array]
\linfer[]{
\linfer[decompand]{
\linfer[Kd+dW]{
\linfer[gddR]{
\linfer[andl]{
  \lclose
}
  {\lsequent{\ivr \land (\lnot{\rfvar} \lor \initassum)}{\ivr}}
}
  {\lsequent{\ddiamond{\pevolvein{\D{x}=\genDE{x}}{\ivr \land (\lnot{\rfvar} \lor \initassum)}}{\lnot{\rfvar}}}{\ddiamond{\pevolvein{\D{x}=\genDE{x}}{\ivr}}{\lnot{\rfvar}}}}
}
  {\lsequent{\ddiamond{\pevolvein{\D{x}=\genDE{x}}{\ivr \land (\lnot{\rfvar} \lor \initassum)}}{x \neq y}}{\ddiamond{\pevolvein{\D{x}=\genDE{x}}{\ivr}}{\lnot{\rfvar}}}}
}
  {\lsequent{\ddiamond{\pevolvein{\D{x}=\genDE{x}}{\ivr}}{x \neq y}, \ddiamond{\pevolvein{\D{x}=\genDE{x}}{\lnot{\rfvar} \lor \initassum}}{x \neq y}}{\ddiamond{\pevolvein{\D{x}=\genDE{x}}{\ivr}}{\lnot{\rfvar}}}}
}
  {\lsequent{\dprogressinsmall{\D{x}=\genDE{x}}{\ivr}, \dprogressin{\D{x}=\genDE{x}}{\lnot{\rfvar}}}{\ddiamond{\pevolvein{\D{x}=\genDE{x}}{\ivr}}{\lnot{\rfvar}}}}
\end{sequentdeduction}
}%

The remaining derivation from \textcircled{2} (with succedent \textcircled{b}) is similar using local progress for the backwards differential equations instead.

For an invariance property $\lforall{x}{(\rfvar\limply\dbox{\pevolvein{\D{x}=\genDE{x}}{\ivr}}{\rfvar})}$, we may, without loss of generality, assume that $\rfvar,\ivr$ (and $\lnot{\rfvar},\lnot{\ivr}$) are equivalently rewritten into appropriate normal forms when necessary with an application of rule \irref{qear}. Therefore, by \irref{semialgiffQ}, the \dL calculus reduces all such invariance questions to a first-order formula of real arithmetic which is decidable.
\end{proof}
\end{document}